\documentclass[12pt,a4paper,twoside]{amsart} 
\usepackage{latexsym}
\usepackage{amsmath}
\usepackage{amssymb}
\usepackage{amsfonts}
\usepackage{mathrsfs}               
\usepackage{bbm}                    
\usepackage{bbold}                  
\usepackage{textcomp}               
\usepackage{amstext}
\usepackage{enumerate}
\usepackage{amstext}
\newcommand{\BB}{\mathbb{B}}
\newcommand{\CC}{\mathbb{C}}
\newcommand{\FF}{\mathbb{F}}

\newcommand{\NN}{\mathbb{N}}

\newcommand{\RR}{\mathbb{R}}

\newcommand{\supp}{\mathrm{supp}}

\newcommand{\Ran}{\mathrm{Ran}}
\newcommand{\sgn}{\mathrm{sgn}}
\newcommand{\diz}{\mathrm{div}}
\newcommand{\Op}{\mathrm{Op}}

\newcommand{\zero}{0}
\newcommand{\id}{\mathbbm{1}}
\newcommand{\klg}{\leqslant} 
\newcommand{\grg}{\geqslant}          
\newcommand{\ve}{\varepsilon}
\newcommand{\vp}{\varphi}
\newcommand{\vk}{\varkappa}
\newcommand{\vr}{\varrho}
\newcommand{\vt}{\vartheta}
\newcommand{\vs}{\varsigma}
\newcommand{\vo}{\varpi}
\newcommand{\wt}[1]{\widetilde{#1}}
    
\newcommand{\SPn}[2]{\langle \,#1\,|\,#2\, \rangle} 
\newcommand{\SPb}[2]{\big\langle \,#1\,\big|\,#2\, \big\rangle} 

\newcommand{\JB}[1]{\langle#1\rangle}
\newcommand{\ol}[1]{\overline{#1}} 

\newcommand{\wh}[1]{\widehat{#1}}  
\newcommand{\mr}[1]{\mathring{#1}} 
\newcommand{\bigO}{\mathcal{O}}    
\newcommand{\V}[1]{\mathbf{#1}}
\newcommand{\valpha}{\boldsymbol{\alpha}}
\newcommand{\vgamma}{\boldsymbol{\gamma}}

\newcommand{\vnu}{\boldsymbol{\nu}}
\newcommand{\vPi}{\boldsymbol{\Pi}}

\newcommand{\vsigma}{\boldsymbol{\sigma}}

\newcommand{\LO}{\mathscr{L}}      

\newcommand{\schwartz}{\mathscr{S}}
\newcommand{\dom}{\mathcal{D}}

\newcommand{\spec}{\mathrm{\sigma}}
\newcommand{\specac}{\mathrm{\sigma}_{\mathrm{ac}}}

\newcommand{\D}[1]{D_{#1}}
\newcommand{\dA}{d_{\mathrm{A}}}      
\newcommand{\Figx}{\mathbb{f}}                 
\newcommand{\SMS}{\mathfrak{S}}              
\newcommand{\CF}{\omega}                     
\newcommand{\wtG}{\widetilde{\Gamma}}
\newcommand{\cO}{\mathcal{O}} 

\newcommand{\cP}{\mathcal{P}} 

\newcommand{\cF}{\mathcal{F}}\newcommand{\cR}{\mathcal{R}}

\newcommand{\cK}{\mathcal{K}}
\newcommand{\cL}{\mathcal{L}}

\newcommand{\sA}{\mathscr{A}}\newcommand{\sN}{\mathscr{N}}

\newcommand{\sD}{\mathscr{D}} 
\newcommand{\sE}{\mathscr{E}}
\newcommand{\sF}{\mathscr{F}}
\newcommand{\sG}{\mathscr{G}}
\newcommand{\sH}{\mathscr{H}}
\newcommand{\sI}{\mathscr{I}}
\newcommand{\sV}{\mathscr{V}}
\newcommand{\sK}{\mathscr{K}}\newcommand{\sW}{\mathscr{W}}
         
\newcommand{\sM}{\mathscr{M}}       
 
\newcommand{\fL}{\mathfrak{L}}

\newcommand{\fM}{\mathfrak{M}}

\newcommand{\fs}{\mathfrak{s}}
\newtheorem{theorem}{Theorem}[section]
\newtheorem{lemma}[theorem]{Lemma}
\newtheorem{proposition}[theorem]{Proposition}
\newtheorem{corollary}[theorem]{Corollary}
\newtheorem{hypothesis}[theorem]{Hypothesis}
\theoremstyle{remark}
\newtheorem{remark}[theorem]{Remark}
\newtheorem{example}[theorem]{Example}
\numberwithin{equation}{section}
\title[Green kernel asymptotics for the Dirac operator]{
Semi-classical Green kernel asymptotics
for the Dirac operator}
\author{Oliver Matte}
\address{Oliver Matte\\
Institut f\"ur Mathematik\\
TU Clausthal\\
Erzstra{\ss}e 1\\
D-38678 Clausthal-Zellerfeld, Germany\\
{\em On leave from:} Mathematisches Institut\\
Ludwig-Maximilians-Universit\"at\\
Theresienstra{\ss}e 39\\
D-80333 M\"unchen, Germany.}
\email{matte@math.lmu.de}
\author{Claudia Warmt} 
\address{Claudia Warmt\\
Mathematisches Institut\\
Ludwig-Maximilians-Universit\"at\\
Theresienstra{\ss}e 39\\
D-80333 M\"unchen, Germany.}
\email{warmt@math.lmu.de}
\date{\today}
\begin{document}

\begin{abstract}
We consider a semi-classical Dirac operator in
$d\in\NN$ spatial dimensions with a smooth potential
whose partial derivatives of any order are bounded
by suitable constants. We prove that the distribution
kernel of the inverse operator evaluated at two distinct
points fulfilling a certain hypothesis can be represented
as the product of an exponentially decaying factor
involving an associated Agmon distance
and some amplitude admitting a complete asymptotic expansion
in powers of the semi-classical parameter.
Moreover, we find an explicit formula
for the leading term in that expansion. 
\end{abstract}

\maketitle

\section{Introduction and main results}

\noindent
The free Dirac operator in $d\in\NN$ spatial dimensions
is the matrix-valued
partial differential operator given by
\begin{equation}\label{Dirac1}
\D{h,0}\,:=\,\valpha\cdot(-ih\nabla)+\alpha_0\,:=\,
\sum_{k=1}^d\alpha_k\,(-ih\,\partial_{x_k})+\alpha_0\,,\qquad
h\in(0,1]\,.
\end{equation}
The Dirac matrices $\alpha_0,\ldots,\alpha_d$ appearing here
are hermitian $(d_*\times d_*)$-matrices satisfying the Clifford algebra
relations
\begin{equation}\label{Clifford}
\{\alpha_k\,,\,\alpha_\ell\}\,=\,2\,\delta_{k\ell}\,\id\,,
\qquad k,\ell=0,1,\ldots,d\,.
\end{equation}
According to the representation theory of Clifford algebras
such matrices exist and the minimal choice of their
dimension $d_*\in2\NN$ is $d_*=2^{[(d+1)/2]}$.
The special choice of the Dirac matrices is immaterial
for our purposes; only the relations \eqref{Clifford}
are used explicitly below. It is well-known that,
as an operator acting in the Hilbert space $L^2(\RR^d,\CC^{d_*})$,
$\D{h,0}$ is essentially self-adjoint
on $C_0^\infty(\RR^d,\CC^{d_*})$ and self-adjoint
on $H^1(\RR^d,\CC^{d_*})$.
Its Fourier transform can be easily diagonalized
(compare \eqref{adam5} and \eqref{adam6} below)
revealing that its spectrum is purely absolutely continuous and
given as
\begin{equation}\label{spec-D0}
\spec(\D{h,0})\,=\,\specac(\D{h,0})\,=\,(-\infty,-1]\cup[1,\infty)\,.
\end{equation}
Next, we add a smooth potential, $V$, 
to the free Dirac operator,
\begin{equation}\label{Dirac2}
\D{h,V}\,:=\,\D{h,0}\,+\,V\,\id_{d_*}\,.
\end{equation}
We shall always assume that $V$ has the following properties.

\begin{hypothesis}\label{hyp-V} 
$V\in C^\infty(\RR^d,\RR)$ and,
for every multi-index $\alpha\in\NN_0^d$, 
\begin{equation}\label{martin10}
\sup_{x\in\RR^d}|\partial_x^\alpha V(x)|\,<\,\infty\,.
\end{equation}
Moreover,
there is some
$\delta\in(0,1)$ 
such that
\begin{equation}\label{luise1}
-1+\delta\,\klg\,V(x)\,\klg\,-\delta\,,\qquad
x\in\RR^d.
\end{equation}
\end{hypothesis}

\smallskip

\noindent
In view of \eqref{spec-D0} the previous hypothesis clearly
implies that $\D{h,V}$ is self-adjoint on $H^1(\RR^d,\CC^{d_*})$
and continuously invertible.
In fact, its symbol,
$$
\wh{D}_V(x,\xi)\,:=\,\valpha\cdot\xi+\alpha_0+V(x)\,,\qquad
(x,\xi)\in\RR^{2d},
$$
is uniformly elliptic in the sense that
$$
\big|\det\big(\wh{D}_V(x,\xi)\big)\big|=
\big(1+|\xi|^2-V^2(x)\big)^{d_*/2}\grg(2\delta-\delta^2)^{d_*/2}>0\,,
\quad
(x,\xi)\in\RR^{2d}.
$$
Therefore, the inverse $\D{h,V}^{-1}$ is given by some matrix-valued
$h$-pseudo-differential operator whose distribution kernel,
$\RR^d\times\RR^d\ni(x,y)\mapsto \D{h,V}^{-1}(x,y)$, 
is smooth away from the diagonal.
Our goal is to study the semi-classical
asymptotics of this kernel, for fixed $x\not=y$.

To formulate our main result we first introduce an associated
Agmon distance, $\dA$, on $\RR^d$. It is the Riemannian distance 
corresponding to a metric conformally
equivalent to the Euclidean one on $\RR^d$, namely
\begin{equation}\label{def-G}
G(x)\,:=\,(1-V^2(x))\,\id_d\,,\qquad x\in\RR^d.
\end{equation}
The Agmon distance is thus given as
\begin{equation}\label{def-dA}
\dA(x,y)\,:=\,
\inf_{q:y\rightsquigarrow x}\int \SPb{\dot q}{G(q)\,\dot q}^{1/2},
\qquad x,y\in\RR^d,
\end{equation}
where the infimum is taken over all piecewise
smooth paths $q:[0,b]\to\RR^d$, for some $b>0$,
such that $q(0)=y$ and $q(b)=x$.
We also introduce an associated Hamilton function,
\begin{equation}\label{def-H}
H(x,p)\,:=\,-\sqrt{1-|p|^2}-V(x)\,,\qquad x,p\in\RR^d\,,\;|p|<1\,,
\end{equation}
and recall the following fact (see, e.g., \cite[pp. 197]{GH}):
If a smooth curve ${\gamma\choose\vo}:I\to\RR^{2d}$ on a non-trivial
interval $I$ is a solution of
the Hamiltonian equations
\begin{equation}\label{Ham-Gl-H}
\frac{d}{dt}{x\choose p}\,=\,{\nabla_p H\choose-\nabla_xH}(x,p)
\end{equation}
such that
\begin{equation}\label{H=0gvo}
H(\gamma(t),\vo(t))\,=\,0\,,\qquad t\in I\,,
\end{equation}
then $\gamma$ is a geodesic for the Agmon metric $G$.
Let $\exp_y:T_y\RR^d\to\RR^d$ 
denote the exponential map at $y\in\RR^d$ associated
to the Riemannian metric $G$.
We recall that two points $x,y\in\RR^d$ are called conjugate to each other
iff the derivative $\exp_y'(v)$ is singular, where $v\in\RR^d$
is chosen such that $\exp_y(v)=x$.
In this article we shall restrict our attention to arguments
of the Green kernel fulfilling the following hypothesis.

\begin{hypothesis}\label{hyp-geo-Dirac}
$x_\star,y_\star\in\RR^d$, $x_\star\not=y_\star$, and,
up to reparametrization, there is a unique minimizing
geodesic
from $y_\star$ to $x_\star$. Moreover, $x_\star$ and $y_\star$
are not conjugate to each other. 
\end{hypothesis}

\smallskip

\noindent
We recall that Hypothesis~\ref{hyp-geo-Dirac} is always
fulfilled, for fixed $y_\star$, provided that $x_\star$
is sufficiently close to $y_\star$.
To state our main results we also introduce the
-- in general non-orthogonal -- projections $\Lambda^\pm$
defined by
\begin{equation}\label{adam7}
\Lambda^\pm(\zeta):=
\frac{1}{2}\,\id\,\pm\,\frac{1}{2}\,S(\zeta)\,,
\quad S(\zeta):=
\frac{\valpha\cdot\zeta\,+\,\alpha_0}{
\sqrt{1+\zeta^2}}\,,
\quad \zeta\in\CC^d,\;|\Im\zeta|<1\,;
\end{equation}
compare Subsection~\ref{ssec-proj}.
Here and henceforth we abbreviate $\zeta^2:=\zeta_1^2+\dots+\zeta_d^2$,
for every $\zeta\in\CC^d$, and 
$\sqrt{\cdot}$ denotes the branch of the square root
slit on the negative real axis 
satisfying $\Re\sqrt{\cdot}>0$.
The following theorem presents the main result of this article
in the case $d\grg2$.

\begin{theorem}\label{mainthm}
Let $d\grg2$ and
assume that $V$ fulfills Hypothesis~\ref{hyp-V} and
$x_\star,y_\star$ fulfill Hypothesis~\ref{hyp-geo-Dirac}.
Let ${\gamma\choose\vo}:[0,\tau]\to\RR^{2d}$ be a smooth curve
solving \eqref{Ham-Gl-H} and satisfying \eqref{H=0gvo}
such that $\gamma(0)=y_\star$ and $\gamma(\tau)=x_\star$.
Then, as $h>0$ tends to zero, 
\begin{align}
\D{h,V}^{-1}(x_\star,y_\star)&=
\frac{1}{h^d}\cdot\nonumber
\frac{(1-V^2(x_\star))^{\frac{d-2}{4}}(1-V^2(y_\star))^{\frac{d-2}{4}}}{
\det\big[\exp_{y_\star}'(\exp_{y_\star}^{-1}(x_\star))\big]^{1/2}}
\cdot
\frac{(1+\bigO(h))\,e^{-\dA(x_\star,y_\star)/h}}{
\big(2\pi\,\dA(x_\star,y_\star)/h\big)^{\frac{d-1}{2}}}
\\
&\qquad\cdot U(\tau)\,(-V(y_\star))\label{asymp-DV}
\Lambda^+(i\vo(0))\,,
\end{align}
where $U(t)$, $t\in[0,\tau]$, is a unitary matrix
such that $U$ solves the matrix-valued initial value problem
$$
\frac{d}{dt}U(t)\,=\,-
\frac{i\valpha}{2}\cdot\frac{\nabla V(\gamma(t))}{V(\gamma(t))}\,U(t)\,,
\;t\in[0,\tau]\,,
\qquad U(0)=\id\,.
$$
The term abbreviated by $\bigO(h)$ in \eqref{asymp-DV} admits
a complete asymptotic expansion in powers of $h$.
\end{theorem}

\begin{proof}
This theorem follows from \eqref{adam1}, \eqref{adam2},
Proposition~\ref{prop-christa},
and Lemma~\ref{le-christa} below.
\end{proof}

\begin{remark}\label{rem-sym}
(i) The factor $\vr(x,y):=\det\big[\exp_y'(\exp_y^{-1}(x))\big]^{1/2}$
is familiar from the asymptotic expansion of the heat
kernel associated to $G$, where it also appears in the denominator of
the leading coefficient. In particular, is it known to be
symmetric, $\vr(x,y)=\vr(y,x)$.

\smallskip

\noindent(ii)
It follows from Remark~\ref{rem-B0} that
$M(x_\star,y_\star):=U(\tau)\,(-V(y_\star))\Lambda^+(i\vo(0))
=(-V(x_\star))\Lambda^+(i\vo(\tau))\,\alpha_0
\,U(\tau)\,(-V(y_\star))\Lambda^+(i\vo(0))$.
Using the latter formula we verify in the same remark that
$M(x_\star,y_\star)^*=M(y_\star,x_\star)$ so that \eqref{asymp-DV}
has the correct symmetry property of the kernel of a
matrix-valued self-adjoint operator.

\smallskip

\noindent(iii) In Appendix~\ref{app-BMT} we explain,
in the case $d=3$, the connection
between the term $U(\tau)\,(-V(y_\star))\Lambda^+(i\vo(0))$ 
and the BMT equation for the Thomas precession of a classical
spin along a particle trajectory.
The BMT equation is discussed in connection with the semi-classical
analysis of the time evolution generated by $D_{h,V}$ in
\cite{BolteKeppeler1999,RubinowKeller1963}.
\hfill$\Diamond$
\end{remark}

\smallskip

\noindent
Next, we state our main result in the case $d=1$, where we do not need any
restriction on the entries $x\not=y$ of the Green kernel.

\begin{theorem}\label{mainthm-d=1}
Let $x,y\in\RR$, $x\not=y$, and assume that $V$ fulfills
Hypothesis~\ref{hyp-V} with $d=1$. 
Let ${\gamma\choose\vo}:[0,\tau]\to\RR^{2d}$ be a smooth curve
solving \eqref{Ham-Gl-H} and satisfying \eqref{H=0gvo}
such that $\gamma(0)=y$ and $\gamma(\tau)=x$.
Then, as $h>0$ tends to zero, 
\begin{align}\nonumber
\D{h,V}^{-1}(x,y)
&=\frac{1}{h}\cdot
\frac{(1+\bigO(h))\,
\exp\Big(-\Big|\int_{y}^x\big(1-V^2(t)\big)^{1/2}\,dt\Big|\Big/h\Big)}{
(1-V^2(x))^{1/4}(1-V^2(y))^{1/4}}
\\\label{asymp-DV-d=1}
&\qquad\,\cdot
\big(\cos(\vt(\tau))\,\id-i\sin(\vt(\tau))\,\alpha_1\big)\,
(-V(y))\,\Lambda^+(i\omega(0))\,,
\end{align}
where 
$$
\vt(\tau)\,:=\,\int_0^\tau
\frac{V'(\gamma(t))}{2V(\gamma(t))}\,dt\,.
$$
The term abbreviated by $\bigO(h)$ in \eqref{asymp-DV-d=1} admits
a complete asymptotic expansion in powers of $h$.
\end{theorem}

\begin{proof}
This theorem follows from \eqref{adam1}, \eqref{adam2},
and Proposition~\ref{prop-christa-d=1}.
\end{proof}

\begin{remark}
The choice of the sign of $V$ in Hypothesis~\ref{hyp-V}
is not important for our results.
We could equally well consider smooth potentials
$V:\RR^d\to\RR$ satisfying \eqref{martin10}
and $0<\delta\klg V\klg1-\delta$.
In fact, this immediately follows from the following trivial
observation: If $\alpha_0,\ldots,\alpha_d$ are
Dirac matrices and $D_{h,V}$ is defined as in \eqref{Dirac1}
and \eqref{Dirac2} with some positive $V$, then
$D_{h,V}=-\wt{D}_{h,-V}$, where 
$\wt{D}_{h,-V}:=\wt{\valpha}\cdot(-ih\nabla)+\wt{\alpha}_0-V$.
Here the new Dirac matrices
$\wt{\alpha}_j:=-\alpha_j$, $j=0,\ldots,d$,
again satisfy \eqref{Clifford} and, hence,
Theorems~\ref{mainthm} and~\ref{mainthm-d=1}
are applicable to $\wt{D}_{h,-V}$.
In doing so we first observe that, for two
given points $x_\star,y_\star$, the validity
of Hypothesis~\ref{hyp-geo-Dirac} does not depend
on the sign of $V$ since the Agmon metric
$G=(1-V^2)\,\id$ depends only on $V^2$.
Moreover, the
expression in the first line of the right hand
side of \eqref{asymp-DV}, which we denote
by $\Delta(x_\star,y_\star)$, does not depend on the sign of $V$ either,
since the Agmon distance and the exponential map
are defined by means of $G$. We have, however,
to introduce a new Hamilton function,
$$
\wt{H}(x,p)\,:=\,-\sqrt{1-|p|^2}+V(x)\,,\qquad
x\in\RR^d,\;|p|<1\,.
$$ 
Let ${\wt{\gamma}\choose\wt{\vo}}:[0,\wt{\tau}]$
be a Hamiltonian trajectory solving
\eqref{Ham-Gl-H} and \eqref{H=0gvo} with
$H$ replaced by $\wt{H}$ such that
$\wt{\gamma}(0)=y_\star$ and $\wt{\gamma}(\wt{\tau})=x_\star$.
Then the last line of \eqref{asymp-DV} has to be changed as follows.
Since $D_{h,V}^{-1}=-\wt{D}_{h,-V}^{-1}$ we obtain
\begin{align}\label{retno1}
D_{h,V}^{-1}(x_\star,y_\star)\,&=\,-
\Delta(x_\star,y_\star)\,\wt{U}(\wt{\tau})\,V(y_\star)
\Lambda^-(i\wt{\vo}(0))\,,
\end{align}
where the projection $\Lambda^-$ is again defined with
the original $\alpha_j$ and 
$\wt{U}(t)$, $t\in[0,\wt{\tau}]$, is a unitary matrix
such that $\wt{U}$ solves the matrix-valued initial value problem
$$
\frac{d}{dt}\wt{U}(t)\,=\,
\frac{i\valpha}{2}\cdot
\frac{\nabla V(\wt{\gamma}(t))}{{V}(\wt{\gamma}(t))}\,\wt{U}(t)\,,
\;t\in[0,\wt{\tau}]\,,
\qquad \wt{U}(0)=\id\,.
$$
By Remark~\ref{rem-sym}(ii) we may multiply the right hand side
of \eqref{retno1} from the left with 
$V(x_\star)\,\Lambda^-(i\wt{\vo}(\wt{\tau}))\,(-\alpha_0)$.
Similar replacements have to be made in Formula \eqref{asymp-DV-d=1}
for the one-dimensional case. 
\hfill$\Diamond$
\end{remark}

\begin{example}
Assume that $V=E$ is some constant spectral parameter
$E\in(-1,1)$. Then we can compute the Green kernel
of $D_{h,E}=\valpha\cdot(-ih\nabla)+\alpha_0+E$
by means of the Fourier transform and find the well-known expression
\begin{align}\nonumber
D_{h,E}^{-1}(x,y)&=
\frac{D_{h,-E}}{(2\pi)^{d/2}\,h^d}\,
\Big(\frac{|x-y|}{h\sqrt{1-E^2}}\Big)^{1-d/2}\,
K_{d/2}\big(\sqrt{1-E^2}\,|x-y|/h\big)
\\
&=\nonumber
\frac{(1-E^2)^{d/4}}{(2\pi)^{d/2}\,h^d}\,\Big(\frac{r}{h}\Big)^{1-d/2}\,
\Big\{-i\frac{\valpha\cdot\V{r}}{r}\,K_{d/2}'\big(\sqrt{1-E^2}\,r/h\big)
\\
&\quad+\label{G0}
\Big(\alpha_0-E+ih\,(d/2-1)\frac{\valpha\cdot\V{r}}{r^2}\Big)
\frac{K_{d/2}\big(\sqrt{1-E^2}\,r/h\big)}{\sqrt{1-E^2}}
\Big\},
\end{align}
where we abbreviate $r:=|x-y|$ and $\V{r}:=x-y$
in the second line. For large $\rho$, the
Bessel function of the second kind,
$K_{d/2}$, behaves asymptotically as
$K_{d/2}(\rho)=(\pi/2\rho)^{1/2}\,e^{-\rho}\,(1+\bigO(1/\rho))$
and $K_{d/2}'(\rho)=-(\pi/2\rho)^{1/2}\,e^{-\rho}\,(1+\bigO(1/\rho))$.
Moreover, it is clear that $\sqrt{1-E^2}\,\frac{x-y}{|x-y|}$ is the constant
momentum of the Hamiltonian trajectory running from $y$
to $x$ in the level set $\{p^2=1-E^2\}$ and
we readily verify that
$\dA(x,y)=\sqrt{1-E^2}\,|x-y|$,
$(1-V^2(x))^{1/4}(1-V^2(y))^{1/4}=\sqrt{1-E^2}$,
$\exp_y'=\id$, and $U=\id$.
Consequently, the leading asymptotics in \eqref{G0} agrees with 
the value predicted by Theorems~\ref{mainthm}
and~\ref{mainthm-d=1}.
\hfill$\Diamond$
\end{example}

\smallskip

\noindent
An asymptotic expansion analogous to \eqref{asymp-DV} 
has been derived earlier in \cite{Matte2008}
for a certain class of $h$-pseudo-differential operators
whose symbols are periodic in the momentum variables.
In the general case encountered in \cite{Matte2008}
the Agmon metric is
replaced by a suitable Finsler metric. This is due to the fact that
the figuratrix at $x\in\RR^d$, 
\begin{equation}\label{def-Figx}
\Figx_x\,:=\,\big\{\,p\in \RR^d\,:\,H(x,p)\,=\,0\,\big\}\,,
\end{equation}
which is well-defined due to \eqref{luise1}
and just a sphere with radius $\sqrt{1-V^2(x)}$,
is replaced by the boundary of some more general strictly convex 
body in more general situations. 
We also remark that the exponential decay of eigenfunctions and the
semi-classical tunneling effect for
the Dirac operator is studied by means of the Agmon metric in
\cite{Wang1985}.

We briefly outline the strategy of our proofs and the
organization of this article. The first step is to conjugate
the Dirac operator with exponential weights $e^{\vp/h}$
where $\vp$ is essentially given as the Agmon distance
to $y_\star$; compare
\eqref{adam1}--\eqref{def-DVvp} below. 
The choice of $\vp$ is explained more precisely
in Section~\ref{sec-vp}. Its construction is the same as in
\cite{Matte2008} and we shall not repeat the details of the
proofs in Section~\ref{sec-vp}.
The distribution kernel of a parametrix of the conjugated
Dirac operator yields the prefactor in front of
the exponential in \eqref{asymp-DV} and
\eqref{asymp-DV-d=1}.
The symbol of the conjugated Dirac operator,
whose properties are also discussed in Section~\ref{sec-vp},
is given by a non-hermitian matrix having two different
$(d_*/2)$-fold degenerate eigenvalues,
one with a non-negative real part
and another one with a strictly negative real part.
Only the part corresponding to the eigenvalue with
non-negative real part
contributes to the asymptotics of the distribution kernel.
To obtain the asymptotics we first construct, roughly speaking, 
a parametrix for
a ``heat equation'' (backwards in time for the part
of the symbol belonging to the eigenvalue in the left 
complex half-plane)
by means of a WKB construction. Since the  
eigenvalues are complex we use
a Fourier integral operator with complex-valued
phase function as an ansatz for the parametrix and
work with almost analytic extensions.
In order to solve the associated complex time dependent
Hamilton-Jacobi equation we adapt the constructions
of \cite{Ku,MeSj2}. In Section~\ref{sec-CHJ} we provide a
self-contained discussion of the time dependent
Hamilton-Jacobi equation that proceeds along the lines of
\cite{MeSj2} and provides some alternative arguments to
control the derivatives of certain error terms and
implicit functions. 
To solve the transport equations in our 
WKB construction we employ a strategy based on the
Clifford algebra structure we learned from \cite{Yajima1982a}.
This strategy is adapted to our setting in Section~\ref{sec-T}.
(WKB constructions for the usual Dirac equation which also apply
to non-scalar potentials can be found in 
\cite{Khochman2007,RubinowKeller1963}.)
In Section~\ref{sec-parametrix} we construct a parametrix for
the conjugated Dirac operator by integrating the parametrix
for the ``heat equation'' with respect to the time variable
and adding a term accounting for the part of
its symbol left out in the WKB construction.
Finally, in Section~\ref{sec-asymp} we compute the
asymptotics of $e^{\vp(x)/h}\,D_{h,V}^{-1}(x,y)\,e^{-\vp(y)/h}$ by means of
a stationary phase expansion in the time variable
and the momentum variables of the Fourier integral operator.
The main text is followed by an appendix where the
BMT equation for Thomas precession is related to our results.


\section{Construction of a weight function}
\label{sec-vp}

\subsection{Eigenvalues and eigenprojections of the symbol
of the conjugated Dirac operator}\label{ssec-proj}

\noindent
Let us describe the first step in the derivation of the
asymptotics~\eqref{asymp-DV} and~\eqref{asymp-DV-d=1}.
We fix two distinct points,
$x_\star$ and $y_\star$, in $\RR^d$ fulfilling 
Hypothesis~\ref{hyp-geo-Dirac} and seek for some bounded
weight function $\vp\in C^\infty(\RR^d,\RR)$
satisfying
\begin{equation}\label{adam1}
\vp(x_\star)-\vp(y_\star)\,=\,\dA(x_\star,y_\star)\,.
\end{equation}
Since $\vp$ is bounded and smooth it is then clear that
\begin{equation}\label{adam2}
\D{h,V}^{-1}(x_\star,y_\star)\,=\,
e^{-\vp(x_\star)/h}\,
\D{h,V,\vp}^{-1}(x_\star,y_\star)
\,e^{\vp(y_\star)/h},
\end{equation}
where $\RR^d\times\RR^d\ni(x,y)\mapsto\D{h,V,\vp}^{-1}(x,y)$
denotes the distribution kernel of the inverse of the conjugated
Dirac operator
\begin{align}
\D{h,V,\vp}\,&:=\,\nonumber
e^{\vp/h}\,\D{h,V}
\,e^{-\vp/h}
\\
&\,=\,\label{def-DVvp}
\valpha\cdot(-ih\,\nabla+i\nabla\vp)\,+\,\alpha_0\,+\, V\,\id_{d_*}\,.
\end{align}
The Green kernel asymptotics of $\D{h,V,\vp}$ thus yield the
prefactor in front of the exponential in \eqref{asymp-DV} 
and~\eqref{asymp-DV-d=1}
provided that $\vp$ is chosen in the right way.
To motivate the partial differential equation determining $\vp$ 
we observe that the -- in general non-hermitian --
matrix
\begin{equation}\label{adam3}
\wh{D}_{V}(x,\zeta)\,:=\,
\valpha\cdot\zeta\,+\,\alpha_0\,+\, V(x)\,\id\,,
\qquad (x,\zeta)\in\RR^d\times\CC^d,
\end{equation}
has two $(d_*/2)$-fold degenerate complex eigenvalues, namely
\begin{equation}\label{adam4}
\lambda_{\pm}(x,\zeta)\,:=\,
\pm\sqrt{1+\zeta^2}\,+\,V(x)\,,
\qquad (x,\zeta)\in\RR^d\times\CC^d,\;|\Im\zeta|<1\,.
\end{equation}
The eigenprojections corresponding to the eigenvalues
in \eqref{adam4} are given by \eqref{adam7}.
In fact, a straightforward exercise using \eqref{Clifford},
which implies $(\valpha\cdot\zeta)^2=\zeta^2\,\id$ and,
hence, $(\valpha\cdot\zeta+\alpha_0)^2=(\zeta^2+1)\,\id$
reveals that, for $\zeta\in\CC^d$, $|\Im\zeta|<1$,
\begin{align}\label{adam5}
\Lambda^+(\zeta)+\Lambda^-(\zeta)\,&=\,
\id\,,\qquad S(\zeta)^2\,=\,\id\,,
\qquad
\Lambda^\pm(\zeta)^2\,=\,\Lambda^\pm(\zeta)\,,
\\
\qquad\label{adam6}
\wh{D}_{V}(x,\zeta)\,\Lambda^\pm(\zeta)
\,&=\,\lambda_{\pm}(\zeta)\,\Lambda^\pm(\zeta)\,.
\end{align}
Using $\cos^2(\theta/2)=(1+\cos(\theta))/2\grg\cos(\theta)$,
$\theta\in(-\pi/2,\pi/2)$,
we further observe for later reference 
that
\begin{equation*}
\Re\sqrt{z}=\sqrt{|z|}\,\cos(\theta/2)
\grg\sqrt{|z|\,\cos(\theta)}=\sqrt{\Re z}\,,\quad
z=|z|\,e^{i\theta}\in\CC\,,\;\Re z>0\,.
\end{equation*}
In particular,
\begin{equation}\label{est-sqrt2}
\Re\sqrt{1+\zeta^2}\,\grg\,\sqrt{1+(\Re\zeta)^2-(\Im\zeta)^2}\,,\qquad
\zeta\in\CC^d\,,\;|\Im\zeta|<1\,.
\end{equation}


\subsection{Agmon's distance as an optimal weight function}

\noindent
In this subsection we treat the partial differential equation
determining the weight function $\vp$.
This equation is the eikonal equation corresponding to the
Hamilton function introduced in \eqref{def-H} which is related
to the eigenvalue $\lambda_+$ as
\begin{equation}\label{luise0}
H(x,p)
\,=\,-\sqrt{1-p^2}-V(x)\,=\,-\lambda_+(x,ip)\,,
\qquad x\in\RR^d\,,\;|p|<1\,.
\end{equation}
We start with an elementary proposition covering the one-dimensional
case.

\begin{proposition}\label{prop-vp-d=1}
Let $d=1$ and assume that $V$ fulfills
Hypothesis~\ref{hyp-V}. Then the following assertions hold true:

\smallskip

\noindent(i) For $y\in\RR$,
the unique solution in $C^1(\RR,\RR)$ of the initial value problem 
$$
H(x,\phi'(x))\,=\,-\sqrt{1-\phi'(x)^2}-V(x)\,=\,0\,,\quad
x\in\RR\,,\quad\pm\phi'>0\,,\quad \phi(y)=0\,,
$$
is given by the smooth function
$$
\phi(x)\,=\,\pm\int_y^x\sqrt{1-V^2(t)}\,dt\,,\qquad
x\in\RR\,.
$$
We have
$\phi(x)=\pm\dA(x,y)$, for $x\grg y$, and
$\phi(x)=\mp\dA(x,y)$, for $x<y$.

\smallskip

\noindent(ii)
Given $x_\star,y_\star\in\RR$, $x_\star\not=y_\star$, we find
some compact interval,
$K_0\subset\RR$, such that $x_\star,y_\star\in\mr{K}_0$
and some $\vp\in C^\infty(\RR,\RR)$ such that
$\vp(x)-\vp(y)=\dA(x,y)$, for all $x,y\in K_0$
with $\sgn(x-y)=\sgn(x_\star-y_\star)$, 
$\vp$ is constant near $\pm\infty$, and
\begin{equation}\label{vp1-d=1}
H(x,\vp'(x))\klg0\,,\quad x\in\RR\,,
\quad\textrm{and}\quad
H(x,\vp'(x))=0\;\;\Leftrightarrow\;\;x\in K_0\,.
\end{equation}
\end{proposition}

\begin{proof} (i):
Since $-1+\delta\klg V\klg-\delta$, every solution $\phi\in C^1(\RR,\RR)$
of $H(x,\phi'(x))=0$, $x\in\RR$, satisfies either $\phi'>0$ or
$\phi'<0$ on $\RR$. Thus, $H(x,\phi')=0$ is equivalent to either
$\phi'=\sqrt{1-V^2}$ or $\phi'=-\sqrt{1-V^2}$ and the first assertion
is evident. Since the expression $\int\SPn{\dot q}{G(q)\,\dot q}^{1/2}$
is invariant with respect to reparametrizations of the path $q$,
we may plug in $q(t)=y+t$, $t\in[0,x-y]$, for $y\klg x$, or
$q(t)=y-t$, $t\in[0,y-x]$, for $y>x$, to 
verify that $\phi(x)=\pm\sgn(x-y)\,\dA(x,y)$.

(ii): 
We choose some compact interval $K_0$ with $\mr{K}_0\ni x_\star,y_\star$ and
pick some $\theta\in C^\infty_0(\RR,[0,1])$ such
that $\theta(t)=1$ if and only if $t\in K_0$. Then we
define $\vp(x):=\pm\int_{y_\star}^x\theta(t)\,\sqrt{1-V^2(t)}\,dt$, $x\in K_0$,
where we choose the $+$-sign if and only if $x_\star>y_\star$.
By Part~(i) $\vp$ satisfies $H(x,\vp')=0$ on $K_0$ and,
for $x\notin K_0$, we have
$\vp'(x)^2=\theta(x)^2(1-V^2(x))<1-V^2(x)$, that is,
$H(x,\vp'(x))<0$. 
\end{proof}

\smallskip

\noindent
To discuss the multi-dimensional case
we denote the flow of the
Hamiltonian vector field corresponding to $H$ by
$\Phi=(X,P):\dom(\Phi)\to\RR^{2d}$, so that
$\dom(\Phi)=\{(t,x,p)\in\RR\times\RR^d\times B_1:\,
t\in I_{\max}(x,p)\}$ and
\begin{equation}\label{waltraud}
\partial_t\Phi\,=\,
\partial_t{X\choose P}\,=\,{\nabla_pH(X,P)\choose-\nabla_xH(X,P)}\qquad
\textrm{on}\;\;\dom(\Phi)\,.
\end{equation}
Here $I_{\max}(x,p)$ denotes the maximal interval
of existence for the initial value problem
$\dot\rho=(\nabla_p H(\rho),-\nabla_xH(\rho))$, $\rho(0)=(x,p)$,
and $B_1=\{p\in\RR^d:p^2<1\}$.

In order to recall a result from \cite{Matte2008}
we remark that,
if Hypothesis~\ref{hyp-geo-Dirac}
is satisfied, there is, 
up to reparametrization, again
only one minimizing
geodesic running in the opposite direction
from $x_\star$ to $y_\star$.
Moreover, we can prolong the geodesic from $y_\star$ to $x_\star$
or from $x_\star$ to $y_\star$ a little bit such that
it still remains minimizing. 
Now, we are prepared to recall the following special case of
\cite[Proposition~4.5]{Matte2008}:

\begin{proposition}\label{prop-vp}
Let $d\grg 2$ and
assume that $V$ fulfills Hypothesis~\ref{hyp-V} and 
$x_\star$ and $y_\star$ fulfill
Hypothesis~\ref{hyp-geo-Dirac}.
Then there exist a point, $y_0$, on the prolongation of the geodesic
from $x_\star$ to $y_\star$, a compact neighborhood, $K_0$,
of the geodesic segment from $y_\star$ to $x_\star$,
some open set, $\sW\subset T_{y_0}\RR^d=\RR^d$, 
which is star-shaped
with respect to zero, and some bounded function,
$\vp\in C^\infty(\RR^d,\RR)$, with bounded partial derivatives
of any order such that
the following holds true:
\begin{enumerate}
\item[(i)] For all $x\in\RR^d$, we have $|\nabla\vp(x)|<1$ and 
\begin{align}
H(x,\nabla\vp(x))&\klg0\,,\quad\textrm{and}\quad
H(x,\nabla\vp(x))=0\;\;\Leftrightarrow\;\; x\in K_0\,.\label{vp1}
\end{align}
\item[(ii)] $\vp(x)-\vp(y_\star)=\dA(x,y_\star)$, for all
$x$ on the
geodesic segment from $y_\star$ to $x_\star$. 
\item[(iii)] 
$\exp_{y_0}\!\!\upharpoonright_{\sW}\in C^1(\sW,\RR^d)\cap 
C^\infty(\sW\setminus\{0\},\RR^d)$ is injective on $\sW$
and
$$
K_0\,\subset\,\exp_{y_0}(\sW)\setminus\{y_0\}\,.
$$
\item[(iv)] For every $x\in K_0$, there is a unique pair 
$(\tau,p_0)\in(0,\infty)\times\Figx_{y_0}$
such that the projection of 
$[0,\tau]\ni t\mapsto\Phi(t,y_0,p_0)$ onto $\RR^d_x$
is a minimizing geodesic from $y_0$ to $x$.
We have
\begin{equation}\label{vp2}
\Phi(\tau,y_0,p_0)\,=\,(x,\nabla\vp(x))
\,=\,\big(X(\tau,y_0,p_0),\nabla\vp(X(\tau,y_0,p_0))\big)\,.
\end{equation}
\end{enumerate}
\end{proposition}

\begin{proof}
We proved this proposition in \cite{Matte2008} assuming
that $H\in C^\infty(\RR^{2d},\RR)$ and, for
all $x\in\RR^d$, the function 
$H(x,\cdot\,):\RR^d\rightarrow\RR$ is 
strictly convex, even, and $H(x,0)<0$.
On account of \eqref{luise1} we can easily modify
$H$ such that these conditions are satisfied.
To this end we first restrict $H$ given by \eqref{luise0}
to the set
$\RR^d\times\{p\in\RR^d:\,|p|\klg1-\delta^2/2 \}$
and then we pick an arbitrary smooth extension
to $\RR^{2d}$ of this restriction fulfilling
the condition imposed in \cite{Matte2008}.
The assertions of the present proposition do not depend
on the choice of the latter extension.
Another condition required in \cite{Matte2008} is that
\begin{equation}\label{luise2}
\inf\big\{\,F(x,\mathring{v})
\,:\;x\in\RR^d\,,\,\,\mathring{v}\in S^{d-1}\,\big\}\,>\,0\,,
\end{equation}
where $F:\RR^{2d}\to\RR$ is the Finsler structure given by
$F(x,v):=\SPn{v}{G(x)\,v}^{1/2}=\sqrt{1-V^2(x)}\,|v|$,
$x,v\in\RR^d$.
Of course, \eqref{luise2} is a trivial consequence of \eqref{luise1}.
\end{proof}


\subsection{Consequences for the symbol of the conjugated Dirac operator}

\noindent
In this subsection we collect some important properties
of the eigenvalues of the symbol 
\begin{equation}\label{xiao}
\wh{D}_{V,\vp}(x,\xi)\,:=\,\wh{D}_V(x,\xi+i\nabla\vp(x))\,,\qquad
(x,\xi)\in\RR^{2d},
\end{equation}
always assuming that $V$ fulfills Hypothesis~\ref{hyp-V},
$x_\star$ and $y_\star$ fulfill Hypothesis~\ref{hyp-geo-Dirac},
and that $\vp$ is the function provided by
Proposition~\ref{prop-vp-d=1}, for $d=1$, or Proposition~\ref{prop-vp},
for $d\grg2$.
We introduce the complex-valued symbols
\begin{align}
a_\pm(x,\xi)\,&:=\,\mp i\lambda_\pm\label{def-apm}
\big(x,\xi+i\nabla\vp(x)\big)
\\
&\,=\,\nonumber
- i\sqrt{1+(\xi+i\nabla\vp(x))^2}\,\mp\,iV(x)\,,
\qquad (x,\xi)\in\RR^{2d}.
\end{align}
Notice that, since $|\nabla\vp|<1$, 
$a_+$ and $a_-$ are well-defined, complex-valued
smooth functions on $\RR^{2d}$.
In the next two lemmata we collect some basic properties
of $a_\pm$ which are used in the sequel.
Henceforth, we abbreviate $(a_\pm)_{x\xi}'':=d_\xi\nabla_x a_\pm$,
$H_{px}'':=d_x\nabla_p H$, etc.

\begin{lemma}\label{Eigenschaftenvonaplus}
For all $x,\xi\in\RR^d$,
\begin{align}
\Im a_{+}(x,\xi)\,&\klg\, 0\,, \label{Gleichung1}
 \\
\Im a_{+}(x,\xi)\,&=\,0 \quad \Leftrightarrow 
\quad (x,\xi)\in K_{0}\times\{0\}\,. \label{Gleichung3}
\end{align} 
Moreover, we have, for all $x\in\RR^d$,
\begin{align}
a_+(x,0)\,&=\,iH(x,\nabla\vp(x))\,,
\\
\nabla_\xi a_+(x,0)\,&
=\,\nabla_pH(x,\nabla\vp(x))\,,\label{a5}
\\
(a_+)_{\xi x}''(x,0)\,&=\,H_{px}''(x,\nabla\vp(x))
+H_{pp}''(x,\nabla\vp(x))\,\vp''(x)\,,
\label{a7}
\\
(a_+)_{\xi\xi}''(x,0)\,&=\,-iH_{pp}''(x,\nabla\vp(x))\,.\label{a8}
\end{align}
In particular, for all $x\in K_0$,
\begin{align}
a_+(x,0)&=0\,,\quad\nabla_xa_+(x,0)=0\,,\quad (a_+)_{xx}''(x,0)=0\,,\label{a3}
\\
\nabla_\xi a_+(x,0)\,&
=\,-\nabla\vp(x)/V(x)\not=0\,.
\end{align}
\end{lemma}

\begin{proof}
By virtue of \eqref{est-sqrt2} 
and Proposition~\ref{prop-vp}(i) we obtain, for $x,\xi\in\RR^d$,
\begin{align*}
\Im a_{+}(x,\xi)
\,&=\,
-\Re\sqrt{1-(\nabla\vp(x))^{2}+\xi^{2}+2i\,\SPn{\xi}{\nabla\vp(x)}}-V(x)
\\
&\klg\, 
-\sqrt{1-(\nabla\vp(x))^{2}+\xi^{2}}-V\left(x\right) 
\,\klg\,0\,.
\end{align*}
Moreover, by \eqref{vp1} and the previous
inequalities, 
$\Im a_{+}(x,\xi)=0$ if and only if $\xi=0$ and $x\in K_0$,
which yields \eqref{Gleichung1} and \eqref{Gleichung3}.
Furthermore, 
$\nabla_pH(x,p)=-\nabla_p\lambda_+(x,ip)=-i\nabla_\xi\lambda_+(x,ip)$,
for $x\in\RR^d$, $|p|<1$, whence
$\nabla_pH(x,\nabla\vp(x))=\nabla_\xi a_+(x,0)$, for
every $x\in\RR^d$, which is \eqref{a5}.
\eqref{a7} follows from \eqref{a5}.
Next, $H_{pp}''(x,p)=-id_p\nabla_\xi\lambda_+(x,ip)
=(\lambda_+)_{\xi\xi}''(x,ip)$, $x\in\RR^d$, $|p|<1$,
which implies \eqref{a8}.
All remaining identities are obvious.
\end{proof}

\begin{lemma}\label{Eigenschaftenvonaminus}
For all $x,\xi\in\RR^d$,
\begin{align}
\Im a_{-}(x,\xi)&\klg-2\delta<0\,. \label{Vorzeichenaminus}
\end{align}
\end{lemma}

\begin{proof}
Since $H(x,\nabla\vp)\klg0$ on $\RR^d$ 
the inequality \eqref{est-sqrt2} implies
\begin{align*}
\Im a_{-}(x,\xi)\,&=\,
-\Re\sqrt{1-(\nabla\vp(x))^2+\xi^2+2i\,\SPn{\xi}{\nabla\vp(x)}}
+V(x)
\\
&\klg\,-\sqrt{1-(\nabla\vp(x))^2}
+V(x)\,\klg\,2V(x)\,\klg\,-2\,\delta\,,
\end{align*}
for all $x,\xi\in\RR^d$.
\end{proof}


\section{The time-dependent complex Hamilton-Jacobi equation}
\label{sec-CHJ}

\noindent
In this section we solve the complex Hamilton-Jacobi equation
for $\psi_\pm(t,x,\eta)$,
\begin{align}\label{elodie1}
\partial_t\psi_\pm+a_\pm\big(x,\nabla_x\psi_\pm\big)\,
&=\,\bigO\big((\Im\psi_\pm)^N\big)\,,\;\;N\in\NN\,,
\\\label{elodie2}
\psi_\pm(0,x,\eta)\,&=\,\SPn{\eta}{x}\,,
\qquad\Im\psi_\pm\grg0\,.
\end{align}
To this end we proceed along the lines of the constructions
in \cite{MeSj2}. Instead of solving the problem first for symbols 
which are homogeneous of degree one as in \cite{MeSj2}
and then using a standard
reduction to that case (see, e.g., \cite{Sjoestrand1980}) we carry
through all constructions for general symbols as
in \cite{Ku}; see also \cite{Matte2008} for symbols which are
periodic in the momentum variable.
Since $a_\pm$ is complex-valued we can only hope to solve
\eqref{elodie1} up error terms $\bigO((\Im\psi_\pm)^N)$.
These do, however, not any harm in the WKB construction
as we shall see in the proof of Proposition~\ref{prop-sarah};
see, in particular, \eqref{sigrid2}.
In order to control the derivatives of these error
terms and of certain implicit functions appearing in the constructions
we give some arguments alternative to those
in \cite{MeSj2}.
The final result of this section is Corollary~\ref{le-psi}
where \eqref{elodie1}\&\eqref{elodie2} is solved.

In the whole Section~\ref{sec-CHJ} we always assume 
that $V$ fulfills Hypothesis~\ref{hyp-V} and 
$x_\star$, $y_\star$ fulfill Hypothesis~\ref{hyp-geo-Dirac}.
$\vp$ is the function provided by Proposition~\ref{prop-vp-d=1},
for $d=1$, or Proposition~\ref{prop-vp},
for $d\grg2$.


\subsection{Estimates on the Hamiltonian and contact flows}
\label{ssec-flow}

\noindent
From now on the symbols $x$, $y$, $\xi$, and $\eta$ will denote
complex variables in $\CC^d$. We set $\rho=(x,\xi)$ and
$\partial_{x_j}=(\partial_{\Re x_j}-i\partial_{\Im x_j})/2$,
$\partial_{\ol{x}_j}=(\partial_{\Re x_j}+i\partial_{\Im x_j})/2$,
$\nabla_x=(\nabla_{\Re x}-i\nabla_{\Im x})/2$,
$\nabla_{\ol{x}}=(\nabla_{\Re x}+i\nabla_{\Im x})/2$, $\ldots\;$,
and analogously for complex variables other than $x$.
In the rest of this article we further extend
$\vp$ and $V$ almost analytically to smooth functions
defined on $\CC^d$ -- again denoted by the symbols
$\vp$ and $V$ -- so that, for every compact subset
$K\subset\CC^d$ and all $N\in\NN$, $\alpha\in\NN_0^{2d}$,
we find some $C_{N,K,\alpha}\in(0,\infty)$
such that
\begin{align*}
|\partial^\alpha_{(\Re x,\Im x)}\nabla_{\ol{x}}\vp(x)|\,
&\klg\,C_{N,K,\alpha}\,|\Im x|^N,\qquad x\in K\,,
\\
|\partial^\alpha_{(\Re x,\Im x)}\nabla_{\ol{x}}V(x)|\,
&\klg\,C_{N,K,\alpha}\,|\Im x|^N,\qquad x\in K\,.
\end{align*} 
Then we find some open neighborhood, $\Omega\subset\CC^{2d}$,
of $\RR^{2d}$ such that the symbols
$$
a_\pm(x,\xi)=- i\sqrt{1+(\xi+i\nabla\vp(x))^2}\mp iV(x)\,,
\qquad
(x,\xi)\in\Omega\,,
$$
are well-defined and almost analytic on $\Omega$.
For every compact subset
$K\subset\Omega$ and all $N\in\NN$, $\alpha\in\NN_0^{4d}$,
we find some $C_{N,K,\alpha}\in(0,\infty)$
such that
\begin{equation}\label{julia1}
|\partial^\alpha_{(\Re\rho,\Im\rho)}
\nabla_{\ol{\rho}}a_\pm(\rho)|\,
\klg\,C_{N,K,\alpha}\,|\Im\rho|^N,\quad\rho\in K\,.
\end{equation}
(In fact, $\nabla_{\ol{\xi}}a_\pm=0$ on $\Omega$.)
In this subsection we proceed along the lines of
\cite{Ku,MeSj2} to obtain estimates on the 
Hamiltonian and contact flows associated with $a_\pm$.
On ${\Omega}$ we introduce the Hamiltonian vector fields
$$
\sH_{a_\pm}\,:=\,\SPn{\nabla_{\xi}a_\pm}{\nabla_x}-
\SPn{\nabla_xa_\pm}{\nabla_\xi}\,,
$$
and the elementary actions
$$
\sA_\pm(\rho)\,:=\,\SPn{\nabla_\xi a_\pm(\rho)}{\xi}-a_\pm(\rho)\,,\qquad
\rho=(x,\xi)\in{\Omega}\,.
$$
Here and henceforth $\SPn{\cdot}{\cdot\!\cdot}$ denotes
the extension of the Euclidean scalar product
to a {\em bilinear} form on $\CC^d$. 
For later reference we infer from \eqref{a3} that
\begin{equation}\label{julia4}
\sA_+(x,0)\,=\,0\,,\qquad x\in K_0\,.
\end{equation}
We further add an extra variable, $s\in\CC$, to $(x,\xi)\in\Omega$ which
parameterizes the action and define the contact fields
$$
\sK_{a_\pm}\,:=\,-\sA_\pm\,\partial_s+\sH_{a_\pm}\,
=\,(a_\pm-\SPn{\nabla_\xi a_\pm}{\xi})\,
\partial_s+\sH_{a_\pm}\qquad\textrm{on}\;\;\CC\times\Omega\,.
$$
Finally, we introduce the real partial differential operators
$$
\wh{\sH}_{a_\pm}\,:=\,\sH_{a_\pm}+\ol{\sH}_{a_\pm}\,,\qquad
\wh{\sK}_{a_\pm}\,:=\,\sK_{a_\pm}+\ol{\sK}_{a_\pm}\,.
$$
Notice that, since 
$c\,\partial_z+\ol{c}\,\partial_{\ol{z}}
=(\Re c)\,\partial_{\Re z}+(\Im c)\,\partial_{\Im z}$,
the vector in $\CC^{2d}$ corresponding to $\wh{\sH}_{a_\pm}$ under
the identification $\partial_{\Re x_j}\leftrightarrow{\sf e}_j$,
$\partial_{\Im x_j}\leftrightarrow i{\sf e}_j$,
$\partial_{\Re \xi_j}\leftrightarrow{\sf e}_{d+j}$,
$\partial_{\Im \xi_j}\leftrightarrow i{\sf e}_{d+j}$,
$j=1,\ldots,d$, where $({\sf e}_1,\ldots,{\sf e}_{2d})$
is the canonical basis of $\CC^{2d}$, is just
\begin{equation}\label{sarah}
\wh{\sH}_{a_\pm}\,\leftrightarrow\,{\nabla_\xi a_\pm\choose-\nabla_xa_\pm}\,.
\end{equation}
We denote the flow of $\wh{\sH}_{a_\pm}$ as
$\kappa_t^\pm=(Q^\pm,\Xi^\pm):\dom(\kappa^\pm)\to\CC^{2d}$,
so that 
$\dom(\kappa^\pm)=\{(t,y,\eta)\in\RR\times\Omega:\,t\in J_{\max}^\pm(y,\eta)\}$,
where $J_{\max}^\pm(y,\eta)$ denotes the maximal interval of existence for
the $4d$-dimensional real initial value problem 
$\dot\rho=\wh{\sH}_{a_\pm}(\rho)$, $\rho(0)=(y,\eta)$, and
\begin{equation}\label{def-kappaQXi}
\partial_t\kappa^\pm=\partial_t
{Q^\pm\choose\Xi^\pm}=
{\nabla_\xi a_\pm(Q^\pm,\Xi^\pm)\choose-\nabla_x a_\pm(Q^\pm,\Xi^\pm)}
\end{equation}
on $\dom(\kappa^\pm)$.
The flow of $\wh{\sK}_{a_\pm}$ is then given by 
$(\vs^\pm,\kappa^\pm):\CC\times\dom(\kappa^\pm)\to\CC^{1+2d}$, where
$$
\vs^\pm(s,t,y,\eta):=\vs_t^\pm(s,y,\eta):=
s-\int_0^t\sA_\pm(\kappa_r^\pm(y,\eta))\,dr,\;\;\,s\in\CC\,,\,
(t,y,\eta)\in\dom(\kappa^\pm).
$$

\begin{lemma}\label{le-egon}
Let $x_0\in K_0$ and $I\subset I_{\max}(x_0,\nabla\vp(x_0))$
some interval such that $X(t,x_0,\nabla\vp(x_0))\in K_0$,
for all $t\in I$. Then $I\subset J_{\max}^+(x_0,0)$
and
$$
Q^+(t,x_0,0)\,=\,X\big(t,x_0,\nabla\vp(x_0)\big)\,,\quad
\Xi^+(t,x_0,0)\,=\,0\,,
\qquad t\in I\,.
$$
(Recall the notation introduced above \eqref{waltraud}.)
\end{lemma}

\begin{proof}
By \eqref{a5} and \eqref{a3} we have 
$\wh{\sH}_{a_+}=\nabla_pH(x,\nabla\vp)\cdot\nabla_{\Re x}$ on $K_0\times\{0\}$.
Moreover, 
$\partial_tX_t(x_0,\nabla\vp(x_0))=\nabla_pH(\Phi_t(x_0,\nabla\vp(x_0)))$,
and $\Phi_t(x_0,\nabla\vp(x_0))=(X_t,\nabla\vp(X_t))(x_0,\nabla\vp(x_0))$,
$t\in I$, by \eqref{vp2}.
\end{proof}

\begin{lemma}\label{le-konrad}
Let $\tau>0$ and assume that $\rho:[0,\tau]\to\RR^{2d}$ is a real
integral curve of $\wh{\sH}_{a_+}$ with 
$\rho(0)\in K_0\times\{0\}=\{\Im a_+=0\}$.
Then $\rho([0,\tau])\subset K_0\times\{0\}$.
\end{lemma}

\begin{proof}
By assumption
$\Im(\nabla_\xi a_+,-\nabla_x a_+)(\rho(t))=\frac{d}{dt}\,\Im\rho(t)=0$.
Since $a_+$ fulfills the Cauchy-Riemann differential equations
on the real domain, it follows that
$(\Im a_+)_{\Re \rho}'(\rho(t))=\Im(a_+)_\rho'(\rho(t))=0$.
Hence, the derivative of $(\Im a_+)\!\!\upharpoonright_{\RR^{2d}}$ 
vanishes along $\rho$, thus
$\Im a_+(\rho(0))=0$ implies $\rho([0,\tau])\subset\{\Im a_+=0\}$.
Using \eqref{Gleichung3} we conclude
$\rho([0,\tau])\subset K_0\times\{0\}$.
\end{proof}

\smallskip

\noindent
In what follows we consider the trajectories of $\wh{\sK}_{a_\pm}$
emanating from the planes
$$
\fL_0(\eta)\,:=\,\big\{(-\psi_0(y,\eta),y,\nabla_y\psi_0(y,\eta))
\,:\;y\in\CC^d\,\big\}\,,\quad 
$$ 
where $\psi_0(y,\eta)=\SPn{\eta}{y}$, $y\in\CC^d$, $\eta\in\RR^d$, so
that $\nabla_y\psi_0(y,\eta)=\eta$.
The reason why we restrict our attention to real $\eta$ is
that in this case $\psi_0$ trivially fulfills the inequality
\begin{equation}\label{ina1}
\Im\psi_0(y,\eta)-\SPn{\Im y}{\Re\nabla_y\psi_0(y,\eta)}\,\grg\,
-\bigO\big(|\Im(y,\eta)|^3\big)\,,\quad y\in\CC^d,
\end{equation}
which is used to derive the estimates of Lemma~\ref{le-fabrizio} below.
(We could equally well consider more general $\psi_0$ satisfying
\eqref{ina1} locally on compact subsets in Lemma~\ref{le-fabrizio}.)
We define $\SMS:\CC^{1+2d}\to\RR$ by
\begin{equation}\label{def-SMS}
\SMS(s,x,\xi):=-\Im s-\SPn{\Im x}{\Re\xi}=
-\Im\big(s+\SPn{x}{\Re\xi}\big),\quad
(s,x,\xi)\in\CC^{1+2d},
\end{equation}
which corresponds to the function $-\SPn{\Im x}{\Re\xi}$
considered in \cite{MeSj2} where the symbol is homogeneous
of degree one in $\xi$. The following lemma is a slight
modification of \cite[Proposition~3.1]{MeSj2}. For the convenience of the
reader we present its proof in Appendix~\ref{app-MeSj}.

\begin{lemma}\label{le-fabrizio}
Let $y_0,\eta_0\in\RR^{d}$.
In the minus-case we set $\tau=0$.
In the plus-case we pick some
$\tau\in J_{\max}^+(y_0,\eta_0)$, $\tau\grg0$.
If $\tau>0$ we 
assume that $\kappa_t^+(y_0,\eta_0)$ is {\em real}, for all
$t\in[0,\tau]$, and $\Im a_+(y_0,\eta_0)=0$. Then
there exist $\ve>0$, $C\in(0,\infty)$,
and some neighborhood $\mathscr{O}\subset \CC^{1+2d}$
of $(-\psi_0(y_0,\eta_0),y_0,\eta_0)$ such that the following inequalities
are satisfied on $\mathfrak{L}_0(\eta_0)\cap\mathscr{O}$, 
for all $0\klg r\klg t\klg\tau+\ve$ and $h\in[0,1]$,
\begin{align}\label{fabrizio1}
\SMS(\vs_t^\pm,\kappa_t^\pm)\,
&\grg\,\frac{1}{2}\int_0^t-\Im a_\pm(\Re\kappa_u^\pm)\,du
-C\,|\Im\kappa_t^\pm |^3,
\\
|\Im\kappa_t^\pm|^2+\SMS(\vs_t^\pm,\kappa_t^\pm)
&\grg
\frac{1}{C}\Big\{
|\Im\kappa_r^\pm|^2+\SMS(\vs_r^\pm,\kappa_r^\pm)
+\int_r^t-\Im a_\pm(\Re\kappa_u^\pm)\,du
\Big\},\label{fabrizio2}
\\
|\Im\kappa_r^\pm|^2\,&\klg\,C\,\big(|\Im\kappa_t^\pm|^2
+\SMS(\vs_t^\pm,\kappa_t^\pm)\big)\,.\label{fabrizio3}
\end{align}
The constants $\ve$ and $C$ can be chosen uniform when $\eta_0$
varies in some compact set.

In particular, if $(s,\rho)\in\mathfrak{L}_0(\eta_0)\cap\mathscr{O}$
and $\kappa_t^\pm(\rho)$ 
and $\vs_t^\pm(s,\rho)$ are both 
real, for some $t\in[0,\tau+\ve]$, then
$\kappa_r^\pm(\rho)$ is real for all $r\in[0,t]$.
\end{lemma}

\smallskip

\noindent
We recall that the Hamilton matrix, $\FF_{a_\pm}$, of
$a_\pm$ is given as (recall \eqref{sarah})
\begin{align}\label{def-FFa}
\FF_a\,\theta&=
\begin{pmatrix}
a_{\xi x}''&a_{\xi\xi}''\\-a_{xx}''&-a_{x\xi}''
\end{pmatrix}{\theta_x\choose\theta_\xi}
+
\begin{pmatrix}
a_{\xi \ol{x}}''&a_{\xi\ol{\xi}}''\\-a_{x\ol{x}}''&-a_{x\ol{\xi}}''
\end{pmatrix}{\ol{\theta}_x\choose\ol{\theta}_\xi},
\quad a\in\{a_+,a_-\},
\end{align}
where $\theta=(\theta_x,\theta_\xi)\in\CC^{2d}$ and
where the matrix on the right vanishes on
the real domain $\RR^{2d}\subset\Omega$.
Here and henceforth we write
$a_{x\ol{\xi}}''=\frac{1}{4}(d{\Re\xi}
+id{\Im\xi})(\nabla_{\Re x}-i\nabla_{\Im x})\,a$, etc.
Let $\sI:\CC^{2d}\to\CC^{2d}$ denote multiplication with $i$.
Then \eqref{julia1} implies
that, for every compact $K\subset\Omega$ and
all $N\in\NN$, $\alpha\in\NN_0^{4d}$, there is some
$C_{N,K,\alpha}\in(0,\infty)$ such that
\begin{align}\label{julia5}\frac{1}{2}\,
\big\|\,\partial_{(\Re\rho,\Im\rho)}^\alpha[\sI,\FF_a]\,\big\|=
\left\|\partial_{(\Re\rho,\Im\rho)}^\alpha
\begin{pmatrix}
a_{\xi \ol{x}}''&a_{\xi\ol{\xi}}''\\-a_{x\ol{x}}''&-a_{x\ol{\xi}}''
\end{pmatrix}
\right\|\klg
C_{N,K,\alpha}\,|\Im\rho|^N
\end{align}
on $K$, where $a$ again is $a_+$ or $a_-$.

\begin{corollary}\label{cor-fabrizio}
Let $(y_0,\eta_0)$, $\tau$, and $\ve$ be as in
Lemma~\ref{le-fabrizio}, set $T:=\tau+\ve$, and let $K\subset\CC^{2d}$
be some sufficiently small compact neighborhood of $(y_0,\eta_0)$.
Then, for all $N\in\NN$ and
$\alpha\in\NN_0^{4d+1}$, $\beta\in\NN_0^{3d+1}$,
there exist $C_{N,K,T,\alpha},C_{N,K,T,\beta}'\in(0,\infty)$
such that
\begin{equation}\label{fabrizio4a}
\big\|\partial_{(t,\Re\rho,\Im\rho)}^\alpha
d_{\ol{\rho}}\kappa_t^\pm(\rho)\big\|\klg 
C_{N,K,T,\alpha}\,\sup_{s\in[0,t]}|\Im\kappa_s^\pm(\rho)|^N,
\;\;\rho\in K,\,t\in[0,T],
\end{equation}
and, for $(y,\eta)\in K\cap(\CC^d\times\RR^d)$, and $t\in[0,T]$,
\begin{equation}\label{fabrizio4}
\big\|\partial_{(t,\Re y,\Im y,\Re\eta)}^\beta
d_{(\ol{y},\ol{\eta})}\kappa_t^\pm(y,\eta)\big\|\klg 
C_{N,K,T,\beta}'
\big(|\Im\kappa_t^\pm(y,\eta)|^2
+\SMS(\vs_t^\pm,\kappa_t^\pm)(y,\eta)\big)^N.
\end{equation}
\end{corollary}

\begin{proof}
Again we drop all $\pm$-indices in this proof.
If $\sI$ denotes multiplication by $i$, we have
$$
\big\|\partial_{(t,\Re\rho,\Im\rho)}^\alpha
d_{\ol{\rho}}\kappa_t(\rho)\big\|=\frac{1}{2}\,
\big\|\,[\sI,\partial_{(t,\Re\rho,\Im\rho)}^\alpha
\kappa_t'(\rho)\,]\big\|\,,\quad \rho\in K.
$$ 
Moreover, we know that $\kappa_t'(\rho)$, $\rho\in K$, satisfies
$\frac{d}{dt}\,\kappa_t'(\rho)=\FF_a(\kappa_t(\rho))\,\kappa_t'(\rho)$,
$t\in[0,T]$, $\kappa_s'(\kappa_t'(\rho))=\kappa_{t+s}'(\rho)$,
$t,t+s\in[0,T]$, and $\kappa_0'(\rho)=\id$.
In particular, $[\sI\,,\,\kappa_0'(\rho)]=0$.
Since we have
\begin{align*}
\frac{d}{dt}\,[\sI,\kappa_t'(\rho)]\,&=\,
\big[\sI,\,\FF_a(\kappa_t(\rho))\,\kappa_t'(\rho)\big]
\\
&=\,
\FF_a(\kappa_t(\rho))\,[\sI,\kappa_t'(\rho)]\,+\,
[\sI,\FF_a(\kappa_t(\rho))]\,\kappa_t'(\rho)\,,
\qquad t\in[0,T],
\end{align*}
it thus follows from Duhamel's formula that
\begin{align*}
[\sI,\kappa_t'(\rho)]\,=\,
\int_0^t\kappa_{t-s}'(\rho)\,[\sI,\FF_a(\kappa_s(\rho))]\,\kappa_s'(\rho)\,ds
\,,\qquad t\in[0,T],\;\rho\in K.
\end{align*}
Using $\sup_{t\in[0,T]}\sup_{\rho\in K}\|\kappa_t'(\rho)\|<\infty$,
and $[\sI,\FF_a(\kappa_s(\rho))]=\bigO\big(|\Im\kappa_s(\rho)|^N\big)$,
$N\in\NN$,
we deduce that the following estimate is satisfied in the case $\alpha=0$,
\begin{align}\label{fabrizio5}
\big\|[\sI,\partial^\alpha_{(\Re\rho,\Im\rho)}\kappa_t'(\rho)]\big\|&\klg
C_{K,N,T,\alpha}'\,\sup_{s\in[0,T]}|\Im\kappa_s(\rho)|^{N},\qquad \rho\in K\,,
\end{align}
for some $C_{K,N,T,\alpha}'\in(0,\infty)$.
If \eqref{fabrizio5} holds true, for all multi-indices of
length $\klg n\in\NN_0$, and $|\alpha|=n+1$, we write
\begin{align}\label{fabrizio5a}
&\frac{d}{dt}\,[\sI,\partial^\alpha_{(\Re\rho,\Im\rho)}\kappa_t']
=\FF_a(\kappa_t)\,
[\sI,\partial^\alpha_{(\Re\rho,\Im\rho)}\kappa_t']
+[\sI,\FF_a(\kappa_t)]\,
\partial^\alpha_{(\Re\rho,\Im\rho)}\kappa_t'
\\
+\!\!\sum_{0<\beta\klg\alpha}\!&
{\alpha\choose\beta}\big\{
[\sI,\partial^\beta_{(\Re\rho,\Im\rho)}\FF_a(\kappa_t)]\,
\partial^{\alpha-\beta}_{(\Re\rho,\Im\rho)}\kappa_t'
+\partial^\beta_{(\Re\rho,\Im\rho)}\FF_a(\kappa_t)\,
[\sI,\partial^{\alpha-\beta}_{(\Re\rho,\Im\rho)}\kappa_t']
\big\}.\nonumber
\end{align}
Here we know that
$\sup_{t\in[0,T]}\sup_{\rho\in K}
\|\partial^{\alpha-\beta}_{(\Re\rho,\Im\rho)}\kappa_t'(\rho)\|<\infty$,
for $0\klg\beta\klg\alpha$,
and $\|\,[\sI,\partial^{\beta}_{(\Re\rho,\Im\rho)}\FF_a(\kappa_t)]\,\|
=\bigO\big(|\Im\kappa_t(\rho)|^{N}\big)$ by \eqref{julia5}.
By the induction hypothesis 
$\|\,[\sI,\partial^{\alpha-\beta}_{(\Re\rho,\Im\rho)}\kappa_t']\,\|
=\bigO\big(\sup_{s\in[0,T]}|\Im\kappa_s(\rho)|^{N}\big)$, 
for $0<\beta\klg\alpha$.
Applying Duhamel's formula once more, using
$[\sI,\partial^\alpha_{(\Re\rho,\Im\rho)}\kappa_0']=0$,
we obtain \eqref{fabrizio5}, for $\alpha\in\NN_0^{4d}$, $|\alpha|=n+1$.
Taking successively time derivatives of \eqref{fabrizio5a}
we can apply a bootstrap argument to
include higher order time derivatives and to get 
\begin{align}\label{fabrizio5b}
\big\|[\sI,\partial^\alpha_{(t,\Re\rho,\Im\rho)}\kappa_t'(\rho)]\big\|&\klg
C_{K,N,T,\alpha}''\,\sup_{s\in[0,T]}|\Im\kappa_s(\rho)|^{N},\qquad \rho\in K\,,
\end{align}
for some $C_{K,N,T,\alpha}''\in(0,\infty)$ and all $\alpha\in\NN_0^{4d+1}$.
Combining \eqref{fabrizio5b} with
\eqref{fabrizio3} we arrive at the asserted estimates 
\eqref{fabrizio4a} and \eqref{fabrizio4}.
\end{proof}


\subsection{Approximate solution of the complex Hamilton-Jacobi equation}

\noindent
We recall the notation \eqref{def-kappaQXi} and set
\begin{align*}
\wt{\sD}^+\,&:=\,
\big\{\,(t,y,0)\in\RR^+_0\times K_0\times\{0\}\::\;
Q_{t'}^+(y,0)\in K_0\,,\;t'\in[0,t]\,\big\}\,,
\\
\sD^+\,&:=\,\big\{\,(t,x,0)\in\RR^+_0\times K_0\times\{0\}\::\;
Q_{-t'}^+(x,0)\in K_0\,,\;t'\in[0,t]\,\big\}\,,
\\
\wt{\sE}^+\,&:=\,(\{0\}\times\RR^{2d})\cup\wt{\sD}^+,\qquad
\sE^+\,:=\,(\{0\}\times\RR^{2d})\cup\sD^+,
\\
\sF^+\,&:=\big\{\,(t,Q_t^+(y,0),0,y,0)\,:\;(t,y,0)\in\wt{\sD}^+\,\big\}
\cup\big\{(0,x,\eta,x,\eta)\,:\;x,\eta\in\RR^d\big\}
\,.
\end{align*}
We can represent $\sF^+$ as the graph of the function
$(k^+,g^+):\sE^+\to\RR^{2d}$, where
\begin{align}\label{carla1}
k^+(t,x,0)&=Q^+(-t,x,0)=X(-t,x,\nabla\vp(x)),\quad(t,x,0)\in\sD^+,
\\
k^+(0,x,\eta)&=x,\qquad(x,\eta)\in\RR^{2d},\label{carla2}
\end{align}
and
\begin{equation}\label{carla3}
g^+(t,x,0)=0\,,\;\;(t,x,0)\in\sD^+,\qquad
g^+(0,x,\eta)=\eta\,,\;\;(x,\eta)\in\RR^{2d}.
\end{equation}
In the minus-case we simply set
\begin{align*}
\wt{\sE}^-:=\sE^-:=\{0\}\times\RR^{2d},\quad
\sF^-:=\big\{(0,x,\eta,x,\eta)\,:\;x,\eta\in\RR^d\big\}\,.
\end{align*}
Then $\sF^-$ is the graph of $(k^-,g^-):\sE^-\to\RR^{2d}$, where
\begin{equation}\label{carla4}
k^-(0,x,\eta)=x\,,\quad g^-(0,x,\eta)=\eta\,,\qquad (x,\eta)\in\RR^{2d}.
\end{equation}
The next lemma shows that we can extend $k^\pm$ and $g^\pm$ to
smooth functions defined in a neighborhood
of $\sE^\pm$ which represent the canonical relations given
by the flow of $\wh{\sH}_{a_\pm}$ as a graph in the
vicinity of $\sF^\pm$.

\begin{lemma}\label{le-gk}
There exist open neighborhoods, $\sG^\pm$ of $\ol{\sE}^\pm$
in $\RR^+_0\times\CC^{2d}$ and $\sH^\pm$ of $\ol{\sF}^\pm$ in
$\RR^+_0\times\CC^{2d}\times\Omega$,
and $k^\pm,g^\pm\in C^\infty(\sG^\pm,\CC^d)$, such that, for all
$(t,x,\xi,y,\eta)\in\sH^\pm$,
$$
(x,\xi)=\big(Q_t^\pm(y,\eta),\Xi_t^\pm(y,\eta)\big)\quad\Leftrightarrow\quad
\big(\,y=k^\pm(t,x,\eta)\;\wedge\;\xi=g^\pm(t,x,\eta)\,\big)\,.
$$
\end{lemma}

\begin{proof}
We define 
\begin{equation}\label{def-IFF}
F^\pm(t,x,\xi,y,\eta)\,:=\,\big(x-Q^\pm(t,y,\eta)\,,\,
\xi-\Xi^\pm(t,y,\eta)\big)\,,
\end{equation}
for all $(x,\xi,t,y,\eta)\in\CC^{2d}\times\dom(\kappa^\pm)$. 
In the following we regard $F^\pm$ as a $\RR^{4d}$-valued
function of $8d+1$ real variables.
At $t=0$ we have $F^\pm(0,x,\xi,y,\eta)=(x-y,\xi-\eta)$ and it is
trivial that $(F^\pm)_{(\Re\xi,\Im\xi,\Re y,\Im y)}'|_{t=0}:\RR^{4d}\to\RR^{4d}$
is invertible.
This already proves the assertion in the minus-case.
In the plus-case we have,
for general $(x,\xi,t,y,\eta)\in\CC^{2d}\times\dom(\kappa^+)$,
$$
(F^+)_{(\Re\xi,\Im\xi,\Re y,\Im y)}'(t,x,\xi,y,\eta)\,=\,
\begin{pmatrix}
\zero_{2d}&-(Q^+)_{(\Re y,\Im y)}'
\\
\id_{2d}&-(\Xi^+)_{(\Re y,\Im y)}'
\end{pmatrix}(t,y,\eta)\,.
$$
We know, however, that $Q^+(t,y,0)=X(t,y,\nabla\vp(y))$, 
for $(t,y,0)\in\wt{\sD}^+$, where we use the notation 
introduced in the paragraph preceding \eqref{waltraud}.
In view of \eqref{fabrizio5} we further have
$(\Im Q^+)_{\Im y}'(t,y,0)=(\Re Q^+)_{\Re y}'(t,y,0)=d_y[X(t,y,\nabla\vp(y))]$
and $(\Re Q^+)_{\Im y}'(t,y,0)=-(\Im Q^+)_{\Re y}'(t,y,0)=0$,
for $(t,y,0)\in\wt{\sD}^+$.
Likewise we have $\Xi^+(t,y,0)=0$ and, hence, 
$(\Xi^+)_{(\Re y,\Im y)}'(t,y,0)=0$,
for $(t,y,0)\in\wt{\sD}^+$.
It follows that
$$
(F^+)_{(\Re\xi,\Im\xi,\Re y,\Im y)}'(t,Q^+(t,y,0),0,y,0)=
\begin{pmatrix}
\zero_{2d}&-d_y[X(t,y,\nabla\vp(y))]\otimes\id_2
\\
\id_{2d}&\zero_{2d}
\end{pmatrix},
$$
for all $(t,y,0)\in\wt{\sD}^+$.
Moreover,
the matrix $d_y[X(t,y,\nabla\vp(y))]$ is invertible as 
a time zero fundamental matrix of some matrix-valued ODE.
In fact, it holds
$X(0,y,\nabla\vp(y))=y$, thus $d_y[X(0,y,\nabla\vp(y))]=\id$, and
\begin{align*}
\partial_td_y[X(t,y,\nabla\vp(y))]\,&=\,
\BB(t,y)\,d_y[X(t,y,\nabla\vp(y))]\,, 
\end{align*}
for every $(t,y,0)\in\wt{\sD}^+$,
where 
\begin{equation}\label{def-BB}
\BB(t,y):=\big(H_{px}''\big(x,\nabla\vp(x)\big)
+H_{pp}''\big(x,\nabla\vp(x)\big)\,\vp''(x)\big)
\big|_{x=X(t,y,\nabla\vp(y))}\,.
\end{equation}
Since $\sF^+$ can globally be represented as a graph
we conclude that functions $k^+$ and $g^+$ with the properties
stated in the assertion exist and are unique if
the neighborhood $\sG^+$ is chosen sufficiently small.
\end{proof}

\smallskip

\noindent
In the remaining part of this subsection we show that
a solution of \eqref{elodie1} is given by the following
formula well-known from classical mechanics,
\begin{align}\nonumber
\psi_\pm(t,x,\eta)&:=
(Z_\pm\circ K_\pm)(t,x,\eta)
\\
&=
\SPn{k^\pm(t,x,\eta)}{\eta}\label{clelia1}
+\int_0^t\sA_\pm\big(\kappa_r^\pm(k^\pm(t,x,\eta),\eta)\big)\,dr\,,
\end{align}
for $(t,x,\eta)\in\sG^\pm$, where
\begin{align*}
Z_\pm(t,y,\eta)\,&:=\,-\vs_t^\pm(-\SPn{y}{\eta},y,\eta)
\,=\,\SPn{y}{\eta}
+\int_0^t\sA_\pm(\kappa_r^\pm(y,\eta))\,dr\,,
\end{align*}
for $(t,y,\eta)\in\dom(\kappa^\pm)$, and
\begin{align}\label{def-Kpm}
K_\pm(t,x,\eta)\,:=\,\big(t,k^\pm(t,x,\eta),\eta\big)\,,
\qquad (t,x,\eta)\in\mathscr{G}^\pm.
\end{align}
To this end we further introduce the following canonical weights which are
used to control the error terms, 
\begin{align*}
\wtG_\pm\,&:=\,|\Im\kappa^\pm|^2
+\SMS(-Z_\pm,\kappa^\pm)\,,\quad\textrm{on}\;\;\dom(\kappa^\pm)\,,
\\
\Gamma_\pm\,&:=\,\wtG_\pm\circ K_\pm\,=\,
|\Im(x,g^\pm)|^2-\SPn{\Im x}{\Re g^\pm}+\Im\psi_\pm\,,
\quad\textrm{on}\;\;\sG^\pm.
\end{align*}
Here we used that $\kappa_t^\pm(k^\pm(t,x,\eta),\eta)=(x,g^\pm(t,x,\eta))$,
for all $(t,x,\eta)\in\sG^\pm$.

\begin{lemma}\label{le-wtN}
There is an open neighborhood, 
$\wt{\sN}_\pm\subset\RR_0^+\times\CC^d\times\RR^d$, of the closure of
$\wt{\sE}^\pm$ such that, for every compact subset
 $K\subset\wt{\sN}_\pm$, we find some $C_K\in(0,\infty)$
such that, for all $(t,y,\eta)\in K$ and $r\in[0,t]$,
\begin{align}\nonumber
&\Im Z_\pm(t,y,\eta)\grg\SPn{\Im Q^\pm}{\Re\Xi^\pm}(t,y,\eta)-
\frac{1}{2}\int_0^t\Im a_\pm(\Re\kappa_s(y,\eta))\,ds
\\\label{hannah1}
&\qquad\qquad\qquad-C_K\,\big|\Im(Q^\pm,\Xi^\pm)(t,y,\eta)\big|^3,
\\ \label{hannah2}
&\frac{1}{2}\,\big|\Im(Q^\pm,\Xi^\pm)(r,y,\eta)\big|^2\klg
\wtG_\pm(r,y,\eta)\klg C_K\,\wtG_\pm(t,y,\eta)\,,
\end{align}
\end{lemma}

\begin{proof}
We apply Lemma~\ref{le-fabrizio}, for every fixed $\eta_0\in\RR^d$,
recalling that the constant $C$ appearing there
can be chosen uniform when $\eta_0$ varies in a compact set.
In the minus case we always choose $\tau=0$ in Lemma~\ref{le-fabrizio}.
In the plus case we choose $\tau=0$, if $(y_0,\eta_0)\notin K_0\times\{0\}$.
If, however, $\eta_0=0$ and $y_0\in K_0$, then we choose
$\tau=\max\{t\grg0:\,X(r,y_0,\nabla\vp(y_0))\in K_0,\,r\in[0,t]\}$.
Then all assumptions of Lemma~\ref{le-fabrizio} are satisfied
because $\Im a_+(y_0,0)=0$ and
$\Im\kappa_t(y_0,0)=\Im(X(t,y_0,\nabla\vp(y_0)),0)=0$, $t\in[0,\tau]$,
by \eqref{Gleichung3} and Lemma~\ref{le-egon}, respectively.
\end{proof}

\smallskip

\noindent
First, we derive some estimates on the derivatives of
the implicit functions $k^\pm$ and $g^\pm$. To this end we put
\begin{equation}\label{martin11}
\sN_\pm:=K_\pm^{-1}(\wt{\sN}_\pm)\,,
\end{equation}
where $K_\pm$ is given by \eqref{def-Kpm} and $\wt{\sN}_\pm$ by
Lemma~\ref{le-wtN}, so that 
$\mathscr{N}_\pm\subset\mathscr{G}^\pm$ is
a neighborhood of 
$\ol{\mathscr{E}}^\pm$ in $\RR_0^+\times\CC^d\times\RR_d$.

\begin{lemma}\label{le-clelia}
Let $k^\pm,g^\pm\in C^\infty(\sG^\pm,\CC^d)$ be the implicit functions
provided by Lemma~\ref{le-gk}. Then, for all compact
subsets $K\subset\sN^\pm$ and all $N\in\NN$ and
$\alpha\in\NN_0^{3d+1}$, there is some $C_{N,K,\alpha}\in(0,\infty)$
such that, for all $(t,x,\eta)\in K$,
\begin{align*}
\big\|\partial^\alpha_{(t,\Re x,\Im x,\Re\eta)}
d_{\ol{x}}(k^\pm,g^\pm)(t,x,\eta)\big\|
\,\klg\,C_{N,K,\alpha}\,
\Gamma_\pm(t,x,\eta)^N.
\end{align*}
\end{lemma}

\begin{proof}
Dropping all $\pm$-indices and using the notation \eqref{def-IFF} we have
$$
\begin{pmatrix}
k_{(\Re x,\Im x)}'
\\
g_{(\Re x,\Im x)}'
\end{pmatrix}
=
\begin{pmatrix}
\zero_{2d}&-Q_{(\Re y,\Im y)}'
\\
\id_{2d}&-\Xi_{(\Re y,\Im y)}'
\end{pmatrix}^{-1}
\begin{pmatrix}
-\id_{2d}
\\
\zero_{2d}
\end{pmatrix},
$$
where all derivatives of $Q$ and $\Xi$ are
evaluated at $(t,k(t,x,\eta),\eta)$, for $(t,x,\eta)\in\sG$.
We denote the above matrices as $A$, $B$, and $C$, so that
$A=B^{-1}\,C$. Let $\sI_n$ represent multiplication with $i$
on $\CC^d=\RR^{2d}$,
that is, $\sI_n=\id_{n}\otimes{0\;-1\choose1\;\;\;\,0}$.
Writing
$[\sI,A]:=\sI_{2d} A-A\sI_d$, etc., we then have
\begin{equation}\label{clelia2}
[\sI,A]=[\sI,B^{-1}]\,C+B^{-1}\,[\sI,C]=
B^{-1}\,[B,\sI]\,B^{-1}\,C\,.
\end{equation}
Taking derivatives of \eqref{clelia2} we obtain,
for $\alpha\in\NN_0^{1+3d}$,
\begin{align*}
&[\sI,\partial_{(t,\Re x,\Im x,\eta)}^\alpha A]
\\
&=
\sum_{\beta+\gamma+\delta=\alpha}\!\!
c(\beta,\gamma,\delta)\,
\{\partial_{(t,\Re x,\Im x,\eta)}^\beta B^{-1}\}
[\partial_{(t,\Re x,\Im x,\eta)}^\gamma B,\sI]
\{\partial_{(t,\Re x,\Im x,\eta)}^\delta (B^{-1}C)\}\,,
\end{align*}
for some combinatorial constants $c(\beta,\gamma,\delta)\in(0,\infty)$.
Here the commutator $[\partial_{(t,\Re x,\Im x,\eta)}^\gamma B,\sI]$
has the form ($\rho=(y,\eta)$)
$$
\sum [\partial_{(t,\Re\rho,\Im\rho)}^{\gamma'} B,\sI]\cdot
\big(\textrm{Polynomial in the partial derivatives of}\;k\big)\,.
$$
We know from Corollary~\ref{cor-fabrizio}
and \eqref{hannah2}
that
\begin{align*}
\big\|\,[\sI,\partial_{(t,\Re\rho,\Im\rho)}^{\gamma'}B]\,\big\|\,\klg\,
C_{N,K,T,\gamma'}\,\Gamma(t,x,\eta)^N.
\end{align*}
\end{proof}

\smallskip

\noindent 
In order to show that the formula \eqref{clelia1} 
defines a solution of \eqref{elodie1} we adapt a standard
proof from classical mechanics and compare the differential of $Z_\pm$
with the pull-back under the map
$$
\Theta_\pm(t,y,\eta):=\big(t,\kappa_t^\pm(y,\eta)\big)=
\big(t,Q^\pm(t,y,\eta),\Xi^\pm(t,y,\eta)\big)\,,
\qquad (t,y,\eta)\in\dom(\kappa^\pm)\,,
$$
of the Cartan form,
$$
\CF_\pm\,:=\,\xi\,dx-a_\pm(x,\xi)\,dt\,.
$$
$\CF_\pm$ is
considered as a form on $\RR\times\CC^{2d}$, 
so that $dx_j=d\Re x_j+id\Im x_j$, and we
abbreviate $\xi\,dx:=\xi_1\,dx_1+\dots+\xi_d\,dx_d$, etc.

\begin{lemma}\label{le-Cartan}
(i) On every compact subset $K\subset\dom(\kappa^\pm)$
such that $|t|\klg t_0$ on $K$ we have, for all $N\in\NN$
and $\alpha\in\NN_0^{4d+1}$,
\begin{equation}\label{Cartan1}
\partial_{(t,\Re y,\Im y,\Re\eta,\Im\eta)}^\alpha
\big(dZ_\pm-\Theta_\pm^*\CF_\pm-y\,d\eta\big)
\,=\,
\bigO\big(\max_{|r|\klg t_0}
|\Im\kappa_r^\pm(y,\eta)|^N\big)\,.
\end{equation}
(ii) Let $\wt{\mathscr{N}}_\pm$ be the set appearing in
Lemma~\ref{le-wtN} (so that $\eta$ is real in the following). Then
$$
\partial_{(t,\Re y,\Im y,\eta)}^\alpha\big(
dZ_\pm-\Theta_\pm^*\CF_\pm-y\,d\eta\big)
\,=\,
t\,\bigO\big(\wtG^N_\pm\big)
$$
on $\wt{\mathscr{N}}_\pm$ and for all 
$N\in\NN$ and $\alpha\in\NN_0^{3d+1}$,
where the $\bigO$-symbols are uniform on compact subsets
of $\wt{\mathscr{N}}_\pm$.
\end{lemma}

\begin{proof}
We drop all $\pm$-indices,
set $\lambda:=dZ-\Theta^*\CF$, and use
\begin{equation}\label{erikah}
\partial_t{Z}\,=\,\SPn{\Xi}{\nabla_\xi a(Q,\Xi)}-a(Q,\Xi)
\,=\,\SPn{\Xi}{\partial_t{Q}}\,-\,a(Q,\Xi)
\end{equation}
to obtain
\begin{align*}
\lambda\,&=\,
(Z_{y}'-\SPn{\Xi}{Q_{y}'})\,dy+
(Z_{\overline{y}}'-\SPn{\Xi}{Q_{\overline{y}}'})\,d\overline{y}+
(Z_{\eta}'-\SPn{\Xi}{Q_{\eta}'})\,d\eta
\\
&\qquad+
(Z_{\overline{\eta}}'-\SPn{\Xi}{Q_{\overline{\eta}}'})\,d\overline{\eta}
\,,
\end{align*}
where $\SPn{\Xi}{Q_{y}'}\,dy$ abbreviates
$\SPn{\Xi}{\partial_{y_1}Q}\,dy_1+\dots+\SPn{\Xi}{\partial_{y_d}Q}\,dy_d$,
etc.
Now, let $\vk$ be one of the variables
$y_j,\overline{y}_j,\eta_j,\overline{\eta}_j$, $j=1,\ldots,d$,
and set $\lambda_\vk:=Z_{\vk}'-\SPn{\Xi}{Q_{\vk}'}$.
Using successively \eqref{erikah}, the Hamiltonian equations
\eqref{def-kappaQXi},
and the almost analyticity of $a$
we find
\begin{align*}
\partial_t{\lambda}_\vk\,&=\,
\partial_\vk\big(\SPn{\Xi}{\partial_t{Q}}-a(Q,\Xi)\big)
-\SPn{\partial_t{\Xi}}{\partial_\vk Q}-
\SPn{\Xi}{\partial_t\partial_\vk{Q}}
\\
&=\,
\SPn{\partial_\vk\Xi}{\nabla_\xi a(Q,\Xi)}-\partial_\vk \big(a(Q,\xi)\big)
+\SPn{\nabla_x a(Q,\Xi)}{\partial_{\vk}Q}
\\
&=\,
-\SPn{\nabla_{\overline{x}}a(Q,\Xi)}{\overline{\partial_{\ol{\vk}}Q}}
-\SPn{\nabla_{\overline{\xi}}a(Q,\Xi)}{\overline{\partial_{\overline{\vk}}\Xi}}\,.
\end{align*}
Thus, by virtue of \eqref{julia1}, $\partial_t{\lambda}_\vk=G_\vk$,
where
\begin{equation}\label{knut2}
\partial_{(t,\Re y,\Im y,\Re\eta,\Im\eta)}^\alpha G_\vk\,=\,
\bigO\big(|\Im(Q,\Xi)|^N\big)\,.
\end{equation}
The initial condition
$(Z,Q,\Xi)|_{t=0}=(\SPn{y}{\eta},y,\eta)$ implies
\begin{align*}
Z_y'|_{t=0}&=\eta\,,\;\;\;& 
Z_\eta'|_{t=0}&=y\,,\quad\qquad
Z_{\overline{y}}'|_{t=0}=Z_{\overline{\eta}}'|_{t=0}=0\,,
\\
Q_y'|_{t=0}\,&=\,\id\,,\;\;\;&
Q_\eta'|_{t=0}&=
Q_{\overline{y}}'|_{t=0}=Q_{\overline{\eta}}'|_{t=0}=0
\,,
\end{align*}
and we conclude from $\partial_t{\lambda}_\vk=G_\vk$ that
$$
\lambda(t,y,\eta)-y\,d\eta=
\sum_{\vk}\int_0^tG_\vk(r,y,\eta)\,dr\,,\qquad
(t,y,\eta)\in\dom(\kappa)\,,
$$
which together with \eqref{knut2} yields~(i). Part~(ii) now follows from 
Lemma~\ref{le-wtN}. 
\end{proof}

\smallskip

\noindent
Since we have
\begin{equation}\label{sascha}
Z_\pm\circ K_\pm=\psi_\pm\,,\quad
Q^\pm\circ K_\pm=x\,,\quad
\Xi^\pm\circ K_\pm=g^\pm\,,
\end{equation}
on $\sG^\pm$ we arrive at the following result:

\begin{proposition}\label{prop-psi-1}
Let $\psi_\pm$ be defined by \eqref{clelia1}.
Then
\begin{align}
\partial_{(t,\Re x,\Im x,\Re\eta)}^\alpha
\big(\partial_t\psi_\pm+a_\pm(x,\nabla_x\psi_\pm)\big)
\,&=\,
\bigO(\Gamma_\pm^N)\,,
\label{irma1}
\\
\partial_{(t,\Re x,\Im x,\Re\eta)}^\alpha(\nabla_x\psi_\pm-g^\pm)
\,&=\,
\bigO(\Gamma_\pm^N)\,,
\label{irma2}
\\
\partial_{(t,\Re x,\Im x,\Re\eta)}^\alpha(\nabla_\eta\psi_\pm-k^\pm)
\,&=\,
\bigO(\Gamma_\pm^N)
\,,\label{irma3}
\\
\partial_{(t,\Re x,\Im x,\Re\eta)}^\alpha\nabla_{\overline{x}}\psi_\pm
\,&=\,
\bigO(\Gamma_\pm^N)
\,,\label{irma4}
\\
\partial_{(t,\Re x,\Im x,\Re\eta)}^\alpha
\nabla_{\overline{\eta}}\psi_\pm
\,&=\,
\bigO(\Gamma_\pm^N)
\,,\label{irma5}
\end{align}
on $\mathscr{N}_\pm$,
for $N\in\NN$ and every multi-index $\alpha\in\NN_0^{3d+1}$.
All $\bigO$-symbols are uniform on compact subsets
of $\mathscr{N}_\pm$. 
\end{proposition}

\begin{proof}
Lemma~\ref{le-Cartan}(ii) and $\Gamma=\wtG\circ K$ imply
$$
\partial_{(t,\Re x,\Im x,\Re\eta)}^\alpha
K^*_\pm\big(dZ_\pm-\Theta_\pm^*\omega_\pm-y\,d\eta\big)
\,=\,\bigO(\Gamma_\pm^N)
$$
on $K^{-1}(\wt{\sN}_\pm)=\sN_\pm$.
On the other hand \eqref{sascha} shows that, on $\sG^\pm$,
\begin{align*}
K^*_\pm\big(dZ_\pm-\Theta_\pm^*\omega_\pm-y\,d\eta\big)
\,&=\,d(K_\pm^*Z_\pm)-(\Theta_\pm\circ K_\pm)^*\omega_\pm-k^\pm\,d\eta
\\
&=\,
d\psi_\pm-g^\pm\,dx+a_\pm(x,g^\pm)\,dt-k^\pm\,d\eta\,.
\end{align*}
\end{proof}

\smallskip

\noindent 
The last corollary of this section summarizes the properties of $\psi_\pm$
on the real domain, where the weight $\Gamma_\pm$ can actually
be replaced by $\Im\psi_\pm$, so that we arrive at the desired solution
of the problem \eqref{elodie1}\&\eqref{elodie2}.

\begin{corollary}\label{le-psi}
(i) There is some real neighborhood, $\mathscr{M}_\RR^\pm$, 
of $\sE^\pm$ in $\RR_0^+\times\RR^{2d}$ such that,
for all $(t,x,\eta)\in\mathscr{M}_\RR^\pm$,
\begin{align}\label{psigrgg}
\Im\psi_\pm(t,x,\eta)\,&\grg\,\frac{1}{\bigO(1)}|\,\Im g^\pm(t,x,\eta)|^2,
\\
\Im\psi_+(t,x,\eta)\,&=\,0\quad\Leftrightarrow\quad
(t,x,\eta)\in\sE^+,\label{psiaufE2}
\\
\Im\psi_-(t,x,\eta)\,&=\,0\quad\Leftrightarrow\quad
t=0\,,\;x,\eta\in\RR^d.\label{psiaufE3}
\end{align}
Consequently,
\begin{equation}\label{psigrgGamma}
\Im\psi_\pm\,\grg\,\frac{1}{\bigO(1)}\:\Gamma_\pm\qquad
\textrm{on}\;\;\mathscr{M}_\RR^\pm\,,
\end{equation}
so that
\eqref{irma1}--\eqref{irma5}
hold true on $\mathscr{M}_\RR^\pm$ 
with the right hand sides replaced
by $\bigO_N((\Im\psi_\pm)^N)$. In particular, 
\begin{equation}\label{t-dep-Ham-Jac-eq}
\partial_{(t,\Re x,\Im x,\Re\eta)}^\alpha
\big(\partial_t\psi_\pm+a_\pm(x,\nabla_x\psi_\pm)\big)\,=\,
\bigO\big((\Im \psi_\pm)^N\big)
\qquad\textrm{on}\;\;\mathscr{M}_\RR^\pm\,.
\end{equation}
(All $\bigO$-symbols are uniform on compact subsets
of $\mathscr{M}_\RR^\pm$.)

\smallskip 

\noindent
(ii) For all $(t,x,0)\in\sD^+$ and $\beta\in\NN_0^{d+1}$, 
\begin{equation}\label{lotta99}
\partial^\beta_{(t,x)}\psi_+(t,x,0)=0\,,\qquad \nabla_\eta\psi_+(t,x,0)=
X(-t,x,\nabla\vp(x))\,. 
\end{equation}
\end{corollary}

\begin{proof}
(i): On account of \eqref{hannah1} and (\ref{sascha}),
\begin{equation}\label{hannah3}
\Im \psi_\pm-\SPn{\Im x}{\Re g^\pm}
\,\grg\,-\bigO(1)\,
|\Im (x,g^\pm)|^3\qquad\textrm{on}\;\;\sN_\pm\,.
\end{equation}
We recall that $\sN_\pm$ is a neighborhood of $\sE^\pm$
in $\RR_0^+\times\CC^d\times\RR^d$ and, hence,
$(\sN_\pm)_\RR:=\sN_\pm\cap(\RR_0^+\times\RR^{2d})$ is
a neighborhood of $\sE^\pm$ in $\RR_0^+\times\RR^{2d}$.
Now, let $(t_0,x_0,\eta_0)\in\sE^\pm$
and let $K\subset(\sN_\pm)_\RR$ be a compact neighborhood
of $(t_0,x_0,\eta_0)$ in $(\sN_\pm)_\RR$. 
By choosing $\ve_0>0$ sufficiently small we can ensure that
$(t,x+\ve\,\Im g^\pm(t,x,\eta),\eta)\in K'\subset(\sN_\pm)_\RR$,
for every $(t,x,\eta)\in K$ and $|\ve|<\ve_0$, where
$K'$ is compact, too.
According to \eqref{hannah3} there exist 
$C,C'\in(0,\infty)$ such that, for all $(t,x,\eta)\in K$,
\begin{align*}
\Im\psi_\pm(t,x-\ve\,\Im g^\pm(t,x,\eta),\eta)
\,&\grg\,
-C\,\big|\Im g^\pm\big(t,x-\ve\,\Im g^\pm(t,x,\eta),\eta\big)\big|^3
\\
&\grg\,-C'\,\big|\Im g^\pm(t,x,\eta)\big|^3
\,.
\end{align*}
Taylor expanding the left hand side of the previous estimate
with respect to $x$ using \eqref{irma2} and \eqref{irma4}
we obtain
\begin{align*}
\Im\psi_\pm&\grg\ve\,|\Im g^\pm|^2-
C''\,\big(|\Im g^\pm|^3+\ve^2\,|\Im g^\pm|^2\big)
-\ve\,C_{N_0}\,|\Im \psi_\pm|^{N_0}\qquad
\textrm{on}\;\;K\,,
\end{align*}
for some ${N_0}\in\NN$, ${N_0}\grg2$, and $C'',C_{N_0}\in(0,\infty)$.
Now, we choose $\ve\in(0,\frac{1}{2C''})$, $\ve<\ve_0$, such that
$\ve\,C_{N_0}\,|\Im\psi_\pm|^{{N_0}-1}<1/2$ on $K$.
Then $\Im\psi_\pm+(1/2)|\Im\psi_\pm|\grg(\ve/2)|\Im g^\pm|^2-C''|\Im g^\pm|^3$
on $K$.
Next, we recall from \eqref{carla3} and \eqref{carla4}
that $\Im g^\pm=0$ on $\sE^\pm$.
Therefore, 
we may further ensure that
$C''|\Im g|<\ve/4$ on $K$ by possibly restricting the
compact neighborhood $K$ of $(t_0,x_0,\eta_0)$ suitably 
and we obtain \eqref{psigrgg}.
Finally, the set $\sM^\pm_\RR$ is defined
as the union of all sets $\mr{K}$ obtained as above for
every $(t_0,x_0,\eta_0)\in\sE^\pm$.

Next, we prove \eqref{psiaufE2} and \eqref{psiaufE3}. 
First, let $(t,x,\eta)\in\sE^\pm$.
At $t=0$ we have
$\Im\psi_\pm(0,x,\eta)=\Im\SPn{x}{\eta}=0$.
If $(t,x,0)\in\sD^+$, then we know that
$\kappa(r,k(t,x,0),0)\in K_0\times\{0\}$,
for all $r\in[0,t]$, whence
$\sA_+(\kappa(r,k(t,x,0),0))=0$, $r\in[0,t]$,
due to \eqref{julia4}. Recalling the definition \eqref{clelia1} of
$\psi_+$ we see that $\psi_+(t,x,0)=0$, which also proves
the first assertion of (ii).

Conversely, assume that $(t,x,\eta)\in\mathscr{M}_\RR^\pm$ with $t>0$
and $\Im\psi_\pm(t,x,\eta)=0$. Then \eqref{psigrgg}
implies that $\Im g^\pm(t,x,\eta)=0$
and \eqref{hannah2} and \eqref{sascha}
show that $(y,\eta):=(k^\pm(t,x,\eta),\eta)$
and $(x,g^\pm(t,x,\eta))$ are connected by a purely real integral curve of
$\wh{\sH}_{a_\pm}$.
Then $\Im a_\pm(y,\eta)<0$
implies $\Im\psi_\pm(t,x,\eta)>0$ on account of
\eqref{Gleichung1}, \eqref{Vorzeichenaminus}, and \eqref{hannah1}.
In the minus-case we thus get a contradiction to
\eqref{Vorzeichenaminus} showing that there is no
$(t,x,\eta)\in\mathscr{M}_\RR^-$ with $t>0$ and $\Im\psi_-(t,x,\eta)=0$.
In the plus-case it follows that $\Im a_+(y,\eta)=0$,
that is, $y\in K_0$ and $\eta=0$ by \eqref{Gleichung3}.
Lemma~\ref{le-konrad} implies that $(t,y,0)\in\wt{\sD}^+$, thus
$(t,x,0)\in{\sD}^+$.

Finally, Part~(ii) follows from 
\eqref{Gleichung3}, \eqref{carla1}, \eqref{carla3}, and
\eqref{irma1}--\eqref{irma3}.
\end{proof}


\section{The transport equations}\label{sec-T}

\subsection{Formal ansatz for a parametrix}

\noindent
In order to construct a parametrix for
the conjugated Dirac operator $\D{h,V,\vp}$
we split, roughly speaking, $\D{h,V,\vp}$ micro-locally into a plus
and a minus part by means of the projections introduced in \eqref{adam7}.
For each of these parts, again roughly speaking, we construct
parametrices for
the corresponding ``heat equations'' (backwards in time in the
minus case) and integrate the latter with respect to the
time variable. The parametrices for
the heat equations are obtained as Fourier integral
operators with complex-valued phase functions.
More precisely, our ansatz for the Green kernel reads
\begin{align*}
\D{h,V,\vp}^{-1}(x,y)
&=
\sum_{\sharp\in\{+,-\}}
\sharp\int_0^\infty\!\!\int_{\RR^d}
e^{i\psi_\sharp(t,x,\eta)/h-i\SPn{y}{\eta}/h}
\sum_{\nu=0}^\infty h^\nu\,B_\sharp^\nu(t,x,\eta)\,
\frac{d\eta\,dt}{(2\pi h)^d\,h}
\\\
&\quad+ \check{q}(x,x-y)\,.
\end{align*}
The symbol $q$ additionally appearing here accounts for the elliptic part of 
$\D{h,V,\vp}$ and is constructed in Section~\ref{sec-parametrix} below. 
To find equations determining
$\psi_\pm$ and $B_\pm^\nu$
we calculate formally
\begin{align}\nonumber
e^{-i\psi_\pm/h}&
\,\big(\pm h\,\partial_t\,+\,
\valpha\cdot(-ih\,\nabla+i\nabla\vp)+\alpha_0+V\big)
e^{i\psi_\pm/h}\sum_{\nu=0}^\infty h^\nu\,B_\pm^\nu
\\
&=\,\nonumber
\big(\pm i\,\partial_t\psi_\pm\,+\,
\valpha\cdot(\nabla_x\psi_\pm+i\nabla\vp)+\alpha_0+V\big)\,B_\pm^0
\\
&\qquad+\nonumber
\sum_{\nu=1}^\infty h^\nu
\Big\{(\pm\partial_t-i\valpha\cdot\nabla)B_\pm^{\nu-1}
\\
&\qquad\qquad\quad\label{sarah1}
+\big(\pm i\partial_t\psi_\pm+\valpha\cdot
(\nabla_x\psi_\pm+i\nabla\vp)+\alpha_0+V\big)\,B_\pm^\nu
\Big\}
\,\stackrel{!}{=}\,0\,.
\end{align}
In the sequel we fix a smooth cut-off function,
$\chi\in C_0^\infty(\CC^{2d})$, such that
$\chi\equiv1$ on some small real neighborhood of
$K_0\times\{0\}$, $0\klg\chi\klg1$ on $\RR^{2d}$,
and such that $\supp(\chi)$ is contained in some
small complex neighborhood of $K_0\times\{0\}$.
We assume that $\chi$ is an almost analytic extension
of $\chi\!\!\upharpoonright_{\RR^{2d}}$, so that
$$
\big|\partial^\alpha_{(\Re y,\Im y,\Re\eta,\Im\eta)}
\nabla_{(\ol{y},\ol{\eta})}\chi(y,\eta)\big|
\,\klg\,C_{N,\alpha}\,|\Im(y,\eta)|^N,\qquad
(y,\eta)\in\CC^{2d},
$$
for all $N\in\NN$, $\alpha\in\NN_0^{4d}$, and suitable
constants $C_{N,\alpha}\in(0,\infty)$.

Let us suppose for the moment
that the matrix-valued amplitudes $B_\pm^0$ satisfy
\begin{align*}
\mathbf{(T_0):}\quad
\partial_{(t,\Re x,\Im x,\eta)}^\alpha&\big\{
B_\pm^0(t,x,\eta)-
\Lambda^\pm\big(x,\nabla_x\psi_\pm(t,x,\eta)+i\nabla\vp(x)\big)
\,B_\pm^0(t,x,\eta)\big\}
\\
&=\,\bigO(\Gamma_\pm^N)\,,
\\
B_\pm^0(0,x,\eta)\,&=\,
\chi(x,\eta)\,\Lambda^\pm\big(x,\eta+i\nabla\vp(x)\big)\,.
\end{align*}
From \eqref{sarah1} we further obtain the transport equations
\begin{align*}
(\mathbf{T}_{\vnu})_{\vnu\grg\mathbf{1}}:\quad
&\partial_{(t,\Re x,\Im x,\eta)}^\alpha\big\{
(\mp\partial_{t}+ i\valpha\cdot\nabla)\,B^{\nu-1}_{\pm}\,
\\
&-(\pm i\partial_{t}\psi_{\pm}+\valpha\cdot(\nabla_x\psi_{\pm}+
i\nabla\vp)+\alpha_0+V)\,B^{\nu}_{\pm}\big\}=
\bigO_{N}(\Gamma_{\pm}^{N})\,,
\\
B^{\nu}_+|_{t=0}
\,&=\,-B^{\nu}_-|_{t=0}\,.
\end{align*}
If $\mathbf{(T_0)}$ is fulfilled, then the matrix in front of
$B_\pm^0$ in \eqref{sarah1} can be replaced by one of its
eigenvalues and we find the eikonal equations
$$
\pm i\,\partial_t\psi_\pm\,\pm
\sqrt{1+(
\nabla_x\psi_\pm+i\nabla\vp)^2}\,+\,V\,=\,\cO(\Gamma_\pm^N)\,,
$$
which, according to the definition \eqref{def-apm}, 
are equivalent to the problems \eqref{elodie1} solved
in Section~\ref{sec-CHJ}.
Again the error terms $\bigO(\Gamma_\pm^N)$ in the transport
equations cannot be avoided because the transport equations
are complex-valued. They
do, however, not destroy the WKB construction as we shall see
later on in Proposition~\ref{prop-sarah}.

\subsection{Solution of the transport equations}

\noindent
In the rest of Section~\ref{sec-T} we assume that $V$ fulfills
Hypothesis~\ref{hyp-V}, $x_\star$ and $y_\star$ fulfill
Hypothesis~\ref{hyp-geo-Dirac}, and that $\vp$ and 
$\psi_\pm$ are the functions given by Propositions~\ref{prop-vp-d=1}
and~\ref{prop-vp} and by \eqref{clelia1}, respectively.
All $\bigO$-symbols are uniform on compact subsets and
the variable $\eta$ will be real.

The transport equations $\mathbf{(T_0)},\mathbf{(T_1)},\ldots\,$
are solved by means of a 
strategy we learned from \cite{Yajima1982a}.
Following this strategy we have, however, to keep track of
the error terms and to put some factors $i$ and some
additional minus signs in the right
places. For we consider a ``heat equation'' with different time
directions for the plus and minus parts of the symbol rather than
the usual Dirac equation whose scattering theory
is discussed in \cite{Yajima1982a}.
For the convenience of the reader
it thus makes sense to give a self-contained discussion of
the transport equations.
As a first step we introduce gamma-matrices
$$
\gamma_0\,:=\,\alpha_0\,,\quad\gamma_j\,:=\,-\alpha_0\,\alpha_j\,,
\quad j=1,\ldots,d\,,
$$
so that
\begin{align}\label{gamma1}
(\gamma_0)^2&=\id\,,\qquad (\gamma_j)^2=-\id\,,\quad j=1,\ldots,d\,,
\\ \label{gamma2}
\{\gamma_\mu\,,\,\gamma_\nu\}&=0\,,\qquad 0\klg \mu<\nu\klg d\,.
\end{align}
Furthermore, we set
\begin{align*}
\partial^\pm_0&:=\pm i\partial_t\,,\qquad
\partial^\pm_j:=-\partial_{x_j}\,,
\\
\Pi^\pm_0&:=\pm\sqrt{1+(\nabla_x\psi_\pm+i\nabla\vp)^2}\,,
\quad \Pi^\pm_j:=\partial_{x_j}\psi_\pm+i\partial_{x_j}\vp\,,
\\
\vPi^\pm&:=\big(\Pi_1^\pm,\ldots,\Pi_d^\pm\big)\,,
\end{align*}
where $j=1,\ldots,d$, and
\begin{align*}
\tilde{\partial}^\pm\,&:=\,
\sum_{\mu=0}^d\gamma_\mu\,\partial^\pm_\mu\,,\qquad
\tilde{\Pi}^\pm\,:=\,
\sum_{\mu=0}^d\gamma_\mu\,\Pi_\mu^\pm\quad\textrm{on}\;\;\sN_\pm\,.
\end{align*}
We recall that the sets $\sN_\pm$ have been introduced in
\eqref{martin11}.
We transform the transport equations into
a new sequence of equations on $\sN_\pm$ given as
\begin{align*}
{(\mathbf{K}_{\vnu})} 
\begin{cases}
\partial_{(t,\Re x,\Im x,\eta)}^\alpha
\big((\tilde{\Pi}^{\pm}-1)\,B^{\nu}_{\pm}
+i\tilde{\partial}^{\pm}B^{\nu-1}_{\pm}\big)\,
=\,\mathcal{O}_{N}(\Gamma_{\pm}^{N})\,, \quad N\in\NN\,,
\\ \partial_{(t,\Re x,\Im x,\eta)}^\alpha
(\tilde{\Pi}^{\pm}+1)\,\tilde{\partial}^{\pm}\,B^{\nu}_{\pm}
\,=\,\mathcal{O}_{N}(\Gamma_{\pm}^{N})\,, \quad N\in\NN\,,
\end{cases}
\end{align*}
where $\nu\in\NN_0$ and $B_\pm^{-1}:=0$.

\begin{lemma}\label{le-T-K}
If $B_\pm^\nu$, $\nu\in\NN_0$, satisfy the first equation
in $\mathbf{(K_1),(K_2),\ldots\;}$ on some neighborhood,
$\sN_\pm'\subset\sN_\pm$, of $\sE_\pm$ in
$\RR^+_0\times\CC^d\times\RR^d$,
then they satisfy the transport equations $\mathbf{(T_1),(T_2),\ldots}\;$
on $\sN_\pm'$, too.
\end{lemma}

\begin{proof}
Let $\nu\in\NN$.
Multiplying the first equation in $(\mathbf{K_\nu})$
with $\alpha_0$
and using $\alpha_0^2=\id$ we obtain
\begin{align}\nonumber
\partial_{(t,\Re x,\Im x,\eta)}^\alpha\big\{&
(\mp\partial_{t}+ i\valpha\cdot\nabla)\,B^{\nu-1}_{\pm}
+\big(\Pi_0^\pm+V
\\\label{andi}
&-\valpha\cdot(\nabla_x\psi_\pm+i\nabla\vp)-\alpha_0-V\big)\,B^{\nu}_{\pm}
\big\}
=\bigO_{N}(\Gamma_{\pm}^{N})\,.
\end{align}
In view of the Hamilton-Jacobi equation 
$\partial_{(t,\Re x,\Im x,\eta)}^\alpha
(\pm i\partial_{t}\psi_{\pm}+
\Pi_0^{\pm}+V)=\bigO_{N}(\Gamma_{\pm}^{N})$,
$N\in\NN$, we see that \eqref{andi} is equivalent to 
$(\mathbf{T}_{\vnu})$.
\end{proof}

\smallskip

\noindent
In what follows we abbreviate
\begin{equation}\label{gamma3}
\sigma_{\mu\nu}\,:=\,\frac{i}{2}\,[\gamma_{\mu}\,,\,\gamma_{\nu}]
\,=\,i\gamma_{\mu}\,\gamma_{\nu}\,,
\quad
({\mathrm{rot}^\pm}\,\Pi^{\pm})_{\mu\nu}\,
:=\,({\partial}_\mu^\pm\Pi^{\pm}_\nu)
-({\partial}_\nu^\pm\Pi^{\pm}_\mu)\,,
\end{equation}
for all $\mu,\nu=0,\ldots,d$, $\mu\not=\nu$.
We introduce
the following matrix-valued partial differential operator on $\sN_\pm$,
\begin{align}
L^{\pm}&:=\pm i\{\partial_t,\Pi^{\pm}_0\}+
\sum_{j=1}^d\{\partial_{x_j},\Pi^{\pm}_j\}
-i\!\!\sum_{0\klg\mu<\nu\klg d}\!\!
\sigma_{\mu\nu}\,({\mathrm{rot}^\pm}\,\Pi^{\pm})_{\mu\nu}\,.
\end{align}
On account of \eqref{gamma1}, \eqref{gamma2}, and \eqref{gamma3} 
we can re-write $L^\pm$ as
\begin{align}
L^\pm\,&=\,\sum_{\mu=0}^d\gamma_\mu^2\,\{\partial^\pm_\mu,\Pi^{\pm}_\mu\}
+\sum_{{\mu,\nu=0\atop\mu\not=\nu}}^d\gamma_\mu\,\gamma_\nu\,
(\partial^\pm_\mu\circ\Pi^{\pm}_\nu+\Pi^{\pm}_\mu\,\partial^\pm_\nu)
\,=\,\{\tilde{\partial}^\pm\,,\,\tilde{\Pi}^{\pm}\}\,.\label{Lpm}
\end{align}

\begin{lemma}\label{le-emely}
Let $\sN_\pm$ be the neighborhood of $\sE_\pm$
in $\RR_0^+\times\CC^d\times\RR^d$ introduced in \eqref{martin11} and
suppose that the matrix-valued function
$R\in C^\infty(\sN_\pm,\LO(\CC^{d_*}))$ 
satisfies 
\begin{equation}\label{emely-2}
\partial_{(t,\Re x,\Im x,\eta)}^\alpha d_{\ol{x}}
R\,=\,\bigO(\Gamma_\pm^N)\,,\qquad N\in\NN\,,\;
\alpha\in\NN_0^{3d+1}.
\end{equation}
Furthermore, let $c\in C^\infty(\CC^d\times\RR^d)$ satisfy
\begin{equation}\label{emely-1}
\partial_{(\Re x,\Im x,\eta)}^\alpha d_{\ol{x}}
c\,=\,\bigO(|\Im x|^N)\,,\qquad N\in\NN\,,\;
\alpha\in\NN_0^{3d}.
\end{equation}
Then there exist $B^\pm\in C^\infty(\sN_\pm,\LO(\CC^{d_*}))$ such
that
\begin{align}\label{emely0a}
\partial_{(t,\Re x,\Im x,\eta)}^\alpha(L^\pm\,B^\pm-R)&
=\bigO(\Gamma_\pm^N),
\;\; \partial^\beta_{(\Re x,\Im x,\eta)}(B^\pm|_{t=0}-c)=
\bigO(|\Im x|^N),
\\\label{emely0b}
\partial_{(t,\Re x,\Im x,\eta)}^\alpha d_{\ol{x}}B^\pm&=
\bigO(\Gamma_\pm^N),
\end{align}
for all $N\in\NN$, $\alpha\in\NN_0^{3d+1}$, and $\beta\in\NN_0^{3d}$.
If $C^\pm\in C^\infty(\sN_\pm,\LO(\CC^{d_*}))$ is another solution of
\eqref{emely0a} and \eqref{emely0b}, then
$\partial_{(t,\Re x,\Im x,\eta)}^\alpha(B^\pm-C^\pm)=\bigO(\Gamma_\pm^N)$,
for all $N\in\NN$ and $\alpha\in\NN_0^{3d+1}$.
\end{lemma}

\begin{proof}
By definition we have
\begin{align*}
L^{\pm}&=\pm2 i\Pi^{\pm}_0\,\partial_t
+2\vPi^\pm\cdot\nabla_x+M^\pm
\,,
\end{align*}
where 
$$
M^\pm\,:=\,\pm i(\partial_t\Pi^{\pm}_0)+\diz_x\vPi^\pm
-i\!\!\sum_{0\klg\mu<\nu\klg d}\!\!
\sigma_{\mu\nu}\,({\mathrm{rot}^\pm}\,\Pi^{\pm})_{\mu\nu}\,.
$$
We further set
$$
W^\pm\,:=\,-(\pm 2i\Pi_0^\pm)^{-1}\,M^\pm\,,\qquad
S^\pm\,:=\,(\pm 2i\Pi_0^\pm)^{-1}\,R\,,
$$
and consider the maximal solutions, $\wt{B}^\pm$, of the
ordinary differential equations
\begin{align}\label{emely1}
\partial_t\wt{B}^\pm(t,y,\eta)\,&=\,W^\pm\big(t,Q^\pm(t,y,\eta),\eta\big)
\,\wt{B}^\pm(t,y,\eta)+S^\pm\big(t,Q^\pm(t,y,\eta),\eta\big)
\end{align}
defined for $(t,y,\eta)\in\wt{\sN}_\pm'$
with boundary condition $\wt{B}(0,y,\eta)=c(y,\eta)$.
We further introduce the vector field 
$$
Z_\pm\,:=\,\nabla_\xi a_\pm(x,g^\pm)\cdot\nabla_x
$$ 
on $\sG^\pm$.
Then an appropriate restriction of $Q^\pm:\dom(\kappa^\pm)\to\CC^d$
is equal to the flow of the real vector field
$\wh{Z}_\pm:=Z_\pm+\ol{Z}_\pm$, because 
$g^\pm(t,Q^\pm(t,y,\eta),\eta)=\Xi^\pm(t,y,\eta)$,
for all $(t,y,\eta)\in\dom(\kappa^\pm)$ such that 
$(t,Q^\pm(t,y,\eta),\eta)\in\sG^\pm$.
Setting
$$
B^\pm(t,x,\eta)\,:=\,\wt{B}^\pm\big(t,k^\pm(t,x,\eta),\eta\big)
$$
we thus have
\begin{align*}
(\partial_t B^\pm+\wh{Z}_\pm\,B^\pm)&\big(t,Q^\pm(t,y,\eta),\eta\big)
\,=\,\frac{d}{dt}\,B^\pm\big(t,Q^\pm(t,y,\eta),\eta\big)
\\
&=\,W^\pm\big(t,Q^\pm(t,y,\eta),\eta\big)
\,\wt{B}^\pm(t,y,\eta)+S^\pm\big(t,Q^\pm(t,y,\eta),\eta\big)\,,
\end{align*}
because $k^\pm(t,Q^\pm(t,y,\eta),\eta)=y$.
In view of 
$(\pm 2i\Pi_0^\pm)^{-1}2\vPi^\pm=\nabla_\xi a_\pm(x,\nabla_x\psi_\pm)$
it follows that
\begin{align}\nonumber
(\pm&2i\Pi_0^\pm)^{-1}(L^\pm\,B^\pm-R)
\\
&=\,\nonumber
\big(\partial_t+\nabla_\xi a_\pm(x,\nabla_x\psi_\pm)\cdot\nabla_x 
-W^\pm\big)\,B^\pm-S^\pm
\\ \label{emely2a}
&=\,-\ol{\nabla_\xi a_\pm(x,g^\pm)}\cdot\nabla_{\ol{x}}B^\pm
+\big(\nabla_\xi a_\pm(x,\nabla_x\psi_\pm)-\nabla_\xi a_\pm(x,g^\pm)\big)\cdot
\nabla_xB^\pm.
\end{align}
From now on we drop all sub- and superscripts 
$\pm$ in the existence part of this
proof since they do not play any role anymore.
By virtue of \eqref{irma2} we know that
$\partial_{(t,\Re x,\Im x,\eta)}^\alpha(\nabla_x\psi-g)
=\bigO(\Gamma^N)$, whence
\begin{equation}\label{emely2b}
\partial_{(t,\Re x,\Im x,\eta)}^\alpha
\big(\nabla_\xi a(x,\nabla_x\psi)-\nabla_\xi a(x,g)\big)
\,=\,\bigO(\Gamma^N)\,,
\end{equation}
for all $N\in\NN$ and $\alpha\in \NN_0^{3d+1}$.
Furthermore, it is clear that all partial derivatives of
$B$ are bounded on compact subsets of $\sN$.
To study $d_{\ol{x}}B$ we write
$$
d_{\ol{x}}B=
(d_{y}\wt{B})(t,k,\eta)\,d_{\ol{x}}k
+(d_{\ol{y}}\wt{B})(t,k,\eta)\,d_{\ol{x}}\ol{k}\,.
$$
Here we know from Lemma~\ref{le-clelia} that
$\partial_{(t,\Re x,\Im x,\eta)}^\alpha d_{\ol{x}}k=\bigO(\Gamma^N)$
and it suffices to find a similar bound on
$d_{\ol{y}}\wt{B}$. To this end we differentiate the differential
equation \eqref{emely1} to get
\begin{align*}
\partial_td_{\ol{y}}\wt{B}&=
W\big(t,Q,\eta\big)\,d_{\ol{y}}\wt{B}
+\big(d_xW\big(t,Q,\eta\big)\,d_{\ol{y}}Q\big)\,\wt{B}
+\big(d_{\ol{x}}W\big(t,Q,\eta\big)\,d_{\ol{y}}\ol{Q}\big)\,\wt{B}
\\
&\quad+d_xS\big(t,Q,\eta\big)\,d_{\ol{y}}Q
+d_{\ol{x}}S\big(t,Q,\eta\big)\,d_{\ol{y}}\ol{Q}\,,
\end{align*}
which on account of \eqref{fabrizio4}, \eqref{irma4}, \eqref{emely-2},
and $|\Im Q|^2\klg\bigO(1)\wtG$, shows that
\begin{equation}\label{emely2}
\partial^\alpha_{(t,\Re y,\Im y,\eta)}\big(\partial_td_{\ol{y}}\wt{B}
-W\big(t,Q,\eta\big)\,d_{\ol{y}}\wt{B}\big)
\,=\,\bigO(\wtG^N)\,.
\end{equation}
Since $\wt{B}(0,y,\eta)=B(0,y,\eta)$ we also have
$\partial^\alpha_{(\Re y,\Im y,\eta)}d_{\ol{y}}\wt{B}|_{t=0}
=\partial^\alpha_{(\Re y,\Im y,\eta)}d_{\ol{y}}c=\bigO(|\Im y|^N)$.
Therefore, 
we first obtain from Duhamel's formula
and $\wtG(s,y,\eta)\klg\bigO(1)\,\wtG(t,y,\eta)$, $s\in[0,t]$, that
\begin{equation}\label{emely3}
d_{\ol{y}}\wt{B}\,=\,U_{t,0}\,\bigO(|\Im y|^N)+
\int_0^t U_{t,s}\,\bigO\big(\wtG(s,y,\eta)^N\big)\,ds
\,=\,\bigO\big(\wtG(t,y,\eta)^N\big)\,,
\end{equation}
where $U_{t,s}$ fulfills
$\partial_tU_{t,s}=W(t,Q(t,y,\eta),\eta)\,U_{t,s}$, $U_{s,s}=\id$.
Now, suppose we have shown that
$\partial^\beta_{(\Re y,\Im y,\eta)}d_{\ol{y}}\wt{B}=\bigO\big(\wtG^N\big)$
is valid, for all $\beta\in\NN_0^{3d}$ with $|\beta|\klg n\in\NN_0$
and let $\alpha\in\NN_0^{3d}$, $|\alpha|=n+1$.
Then we obtain
\begin{equation}\label{emely4}
\partial_t\partial^\alpha_{(\Re y,\Im y,\eta)}d_{\ol{y}}\wt{B}
-W(t,Q,\eta)\,
\partial^\alpha_{(\Re y,\Im y,\eta)}d_{\ol{y}}\wt{B}
\,=\,\bigO(\wtG^N)
\end{equation}
from \eqref{emely2}, and we again conclude that
$\partial^\alpha_{(\Re y,\Im y,\eta)}d_{\ol{y}}\wt{B}=\bigO\big(\wtG^N\big)$.
Using the differential equation \eqref{emely2} and the usual
bootstrap argument to include the time derivatives
we further see that
$\partial^\alpha_{(t,\Re y,\Im y,\eta)}d_{\ol{y}}\wt{B}=\bigO\big(\wtG^N\big)$,
for all $N\in\NN$ and $\alpha\in \NN_0^{3d+1}$.
Altogether we arrive at \eqref{emely0b}
and from \eqref{emely2a} and \eqref{emely2b} we infer that
\eqref{emely0a} holds true also.

Now, suppose that $C^\pm$ is another solution of \eqref{emely0a}
and \eqref{emely0b}.
Then we have $(\pm2i\Pi_0^\pm)^{-1}L^\pm(B^\pm-C^\pm)=G^\pm$, where 
$\partial_{(t,\Re x,\Im x,\eta)}^\alpha G^\pm=\bigO(\Gamma_\pm^N)$,
or (compare \eqref{emely2a})
\begin{align*}
\big(&\partial_t+\wh{Z}_\pm-W^\pm\big)(B^\pm-C^\pm)\,=\,G^\pm
\\
&+\big\{\ol{\nabla_\xi a_\pm(x,g^\pm)}\cdot\nabla_{\ol{x}}
-\big(\nabla_\xi a_\pm(x,\nabla_x\psi_\pm)-\nabla_\xi a_\pm(x,g^\pm)\big)\cdot
\nabla_x\big\}(B^\pm-C^\pm)\,.
\end{align*}
Let $\wt{S}^\pm$ denote the right hand side of the previous
identity.
Setting $\wt{E}^\pm:=(B^\pm-C^\pm)(t,Q^\pm,\eta)$ we then have
\begin{align*}
\partial_t\wt{E}^\pm-W^\pm(t,Q^\pm,\eta)\,\wt{E}^\pm\,&
=\,\wt{S}^\pm(t,Q^\pm,\eta)\,,
\qquad \partial_{(\Re y,\Im y,\eta)}^\beta\wt{E}^\pm|_{t=0}=\bigO(|\Im y|^N)\,,
\end{align*}
where $\partial_{(t,\Re x,\Im x,\eta)}^\alpha\wt{S}^\pm=\bigO(\Gamma_\pm^N)$.
By
the same induction argument as the one used above 
to discuss $d_{\ol{y}}\wt{B}$ we infer that
$\partial_{(t,\Re y,\Im y,\eta)}^\alpha\wt{E}^\pm=\bigO(\wtG_\pm^N)$,
which implies 
$\partial_{(t,\Re x,\Im x,\eta)}^\alpha(B^\pm-C^\pm)=\bigO(\Gamma_\pm^N)$,
for $N\in\NN$, $\alpha\in\NN_0^{3d+1}$.
\end{proof}

\smallskip

\noindent
The next lemma will be applied to the choices
$f^\pm=i\tilde{\partial}^{\pm}B^{\nu}_{\pm}$ and 
$u^{\pm}=B^{\nu+1}_{\pm}$.

\begin{lemma}\label{le-klaus}
Let the matrix-valued functions 
$f^\pm,u^{\pm}\in C^\infty(\sN_\pm,\LO(\CC^{d_*}))$ 
satisfy the equations
\begin{align}
\partial_{(t,\Re x,\Im x,\eta)}^\alpha(\tilde{\Pi}^{\pm}+1)\,f^\pm
\,&=\,
\bigO(\Gamma_{\pm}^{N}),
\label{klaus1}
\\
\partial_{(t,\Re x,\Im x,\eta)}^\alpha\big(
L^{\pm}\,u^{\pm}+\tilde{\partial}^{\pm}f^\pm\big)
\,&=\,
\bigO(\Gamma_{\pm}^{N}),
\label{klaus2}
\\\label{klaus2b}
\partial_{(t,\Re x,\Im x,\eta)}^\alpha d_{\ol{x}}u^\pm\,&=\,\bigO(\Gamma_\pm^N),
\\
\label{klaus2c}
\partial_{(t,\Re x,\Im x,\eta)}^\alpha d_{\ol{x}}f^\pm\,&=\,\bigO(\Gamma_\pm^N),
\end{align}
for every $N\in\NN$ and $\alpha\in\NN_0^{3d+1}$, and
\begin{align}
\partial_{(\Re x,\Im x,\eta)}^\beta
\big((\tilde{\Pi}^{\pm}(0,x,\eta)-1)\,
u^{\pm}(0,x,\eta)+f^\pm(0,x,\eta)\big)
\,&=\,\bigO(|\Im x|^N)\,,\label{klaus3}
\end{align}
for all $(0,x,\eta)\in\sN_\pm$, $N\in\NN$, and $\beta\in\NN_0^{3d}$.
Then $u^{\pm}$ fulfills the following equations on $\sN_\pm$,
for every $N\in\NN$ and $\alpha\in\NN_0^{3d+1}$,
\begin{align}
\partial_{(t,\Re x,\Im x,\eta)}^\alpha\big((\tilde{\Pi}^{\pm}-1)\,
u^{\pm}+f^\pm\big)\,
&=\,\bigO(\Gamma_{\pm}^{N})\,,\label{klaus4}
\\
\partial_{(t,\Re x,\Im x,\eta)}^\alpha(\tilde{\Pi}^{\pm}+1)\,
\tilde{\partial}^{\pm}u^{\pm}
\,&=\,\bigO(\Gamma_{\pm}^{N})\,. 
\label{klaus5}
\end{align}
\end{lemma}

\begin{proof}
On account of \eqref{gamma1} and \eqref{gamma2} we have
$
(\tilde{\Pi}^{\pm})^{2}=
(\Pi^{\pm}_{0})^{2}-(\nabla_x\psi_\pm+i\nabla\vp)^{2}=\id
$,
which together with \eqref{Lpm} gives
$
(\tilde{\Pi}^{\pm}-1)\,L^{\pm}=
L^{\pm}\,(\tilde{\Pi}^{\pm}-1)
$. From this we infer that
\begin{align*}
L^{\pm}\,\big((\tilde{\Pi}^{\pm}-1)\,u^{\pm}+f^\pm\big)
\,&=\,
(\tilde{\Pi}^{\pm}-1)\,L^{\pm}\,u^{\pm}
+(\tilde{\partial}^{\pm}\tilde{\Pi}^{\pm}
+\tilde{\Pi}^{\pm}\tilde{\partial}^{\pm})\,f^\pm
\\
&=\,
(\tilde{\Pi}^{\pm}-1)\,(L^\pm\,u^\pm+\tilde{\partial}^{\pm}f^\pm)
+\tilde{\partial}^{\pm}(\tilde{\Pi}^{\pm}+1)\,f^\pm\,,
\end{align*}
thus
$
\partial_{(t,\Re x,\Im x,\eta)}^\alpha
L^{\pm}\,\big((\tilde{\Pi}^{\pm}-1)\,u^{\pm}+f^\pm\big)
=\bigO_{N}(\Gamma_{\pm}^{N})
$.
Together with Lemma~\ref{le-emely} and \eqref{klaus3} this
shows that \eqref{klaus4} holds true. 
The identity \eqref{klaus5} follows from 
\eqref{klaus2} and \eqref{klaus4} because
$(\tilde{\Pi}^{\pm}+1)\,\tilde{\partial}^{\pm}u^\pm
=L^\pm\,u^\pm+\tilde{\partial}^{\pm}f^\pm-
\tilde{\partial}^{\pm}\big((\tilde{\Pi}^{\pm}-1)\,u^\pm+f^\pm\big)$.
\end{proof}

\smallskip

\noindent
By virtue of Lemma~\ref{le-emely} we may now 
define $B^{\nu}_{\pm}\in C^\infty(\sN_\pm,\LO(\CC^{d_*}))$, $\nu\in\NN_0$, 
by successively solving the differential equations
\begin{equation}
\partial_{(t,\Re x,\Im x,\eta)}^\alpha
\big(L^{\pm}\,B^{\nu}_{\pm}+
i\tilde{\partial}^{\pm}(\tilde{\partial}^{\pm}B^{\nu-1}_{\pm})\big)
\,=\,\bigO(\Gamma_\pm^N)\,,
\qquad B_\pm^{-1}:=0\,,
\label{bea1}
\end{equation}
on $\sN_\pm$ with initial conditions
\begin{equation}\label{bea2}
B^{0}_{\pm}(0,x,\eta)\,:=\,\chi(x,\eta)\,
\Lambda^\pm\big(\vPi^\pm(0,x,\eta)\big)\,=\,\chi(x,\eta)\,
\Lambda^\pm\big(\eta+i\nabla\vp(x)\big)\,,
\end{equation}
for $(0,x,\eta)\in\sN_\pm$, and 
\begin{align}\label{bea3}
B^{\nu}_{\pm}|_{t=0}&:=\,
-(2\Pi_0^\pm)^{-1}
\alpha_0\,
(i\tilde{\partial}^+ B^{\nu-1}_++i\tilde{\partial}^- B^{\nu-1}_-)\big|_{t=0}
\,,\qquad 
\nu\in\NN\,.
\end{align}
Here $\chi$ is the cut-off function introduced below \eqref{sarah1}.
We summarize the results of the previous constructions
in the following proposition.

\begin{proposition}
The matrix-valued amplitudes
$B^{\nu}_{\pm}\in C^\infty(\sN_\pm,\LO(\CC^{d_*}))$
defined by \eqref{bea1}--\eqref{bea3}
solve the original transport equations
$(\mathbf{T_0}),(\mathbf{T_1}),(\mathbf{T_2}), \ldots\; $
on $\sN_\pm$ and 
\begin{equation}\label{bea4}
\partial_{(t,\Re x,\Im x,\eta)}^\alpha(\tilde{\Pi}^\pm+1)\,
\tilde{\partial}^\pm B_\pm^\nu\,=\,\bigO(\Gamma_\pm^N)\,,
\qquad N\in\NN\,,\;\alpha\in\NN_0^{3d+1}.
\end{equation}
Moreover,
\begin{equation}\label{dbar-Bnu}
\partial_{(t,\Re x,\Im x,\eta)}^\alpha d_{\ol{x}}B_\pm^\nu\,
=\,\bigO(\Gamma_\pm^N)\,,
\qquad N\in\NN\,,\;\alpha\in\NN_0^{3d+1}.
\end{equation}
The supports of $B_\pm^\nu$, $\nu\in\NN_0$, are compact and contained
in some fixed compact neighborhood of $\sD_\pm$.
\end{proposition}

\begin{proof}
We argue by induction successively applying Lemma~\ref{le-klaus}.
For $\nu=0$, we set $f^\pm=0$ so that \eqref{klaus1} is
satisfied trivially and \eqref{klaus2} is just \eqref{bea1}.
The initial condition \eqref{klaus3} is implied by \eqref{bea2}.
By Lemma~\ref{le-klaus} we see that \eqref{bea4} is valid, for $\nu=0$,
and that
$\partial_{(t,\Re x,\Im x,\eta)}^\alpha(\tilde{\Pi}^\pm-1)\,B_\pm^0
=\bigO(\Gamma_\pm^N)$,
which implies 
$\partial_{(t,\Re x,\Im x,\eta)}^\alpha\Lambda^{\mp}(\nabla_x\psi_\pm+i\nabla\vp)
\,B_\pm^0=\bigO(\Gamma_\pm^N)$, which is $(\mathbf{T_0})$.
Next, assume that $n\in\NN_0$ and that
$B_\pm^\nu$ fulfills $(\mathbf{T}_{\vnu})$, \eqref{bea4}, and \eqref{dbar-Bnu},
for every $\nu=0,\ldots,n$. Then \eqref{klaus1}--\eqref{klaus2c}
are fulfilled with $f^\pm=i\tilde{\partial}^{\pm}B^{n}_{\pm}$ and
$u^{\pm}=B^{n+1}_{\pm}$ by \eqref{bea4} and the definition of
$B^{n+1}_{\pm}$. 
At $t=0$ we have
\begin{align}\nonumber
(\tilde{\Pi}^\pm-1)\,B^{n+1}_{\pm}\big|_{t=0}
&=(-2\Pi_0^\pm)^{-1}\,\big(\tilde{\Pi}^\pm-1\big)\,\alpha_0\,
(i\tilde{\partial}^+ B^{n}_++i\tilde{\partial}^- B^{n}_-)\big|_{t=0}
\\
&=\label{markus1}
(-2\Pi_0^\pm)^{-1}\,
\big(\Pi_0^\pm+\valpha\cdot\vPi^\pm-\alpha_0\big)\,
(i\tilde{\partial}^+ B^{n}_++i\tilde{\partial}^- B^{n}_-)\big|_{t=0}\,.
\end{align}
Furthermore, we observe that \eqref{bea4} with $\nu=n$ yields
\begin{align}\label{markus2}
\partial^\beta_{(\Re x,\Im x,\eta)}\big\{
\Pi_0^\pm\,(i\tilde{\partial}^\pm B^{n}_\pm)\big|_{t=0}&-
(\valpha\cdot\vPi^\pm-\alpha_0)\,(i\tilde{\partial}^\pm B^{n}_\pm)\big|_{t=0}
\big\}=\bigO(|\Im x|^N)\,,
\end{align}
because 
$\Gamma(0,x,\eta)=|\Im(x,\eta)|^2-\SPn{\Im x}{\Re \eta}
+\Im\SPn{\eta}{x}=|\Im x|^2$,
for real $\eta$, where
\begin{align}\label{markus3}
\vPi^\pm\big|_{t=0}&=\eta+i\nabla\vp(x)
\end{align}
is independent of the choice of the $\pm$-signs.
Combining \eqref{markus1}-\eqref{markus3} we arrive at
\begin{align*}
&\partial^\beta_{(\Re x,\Im x,\eta)}
\big\{(\tilde{\Pi}^\pm-1)\,B^{n+1}_{\pm}\big|_{t=0}\big\}
\\
&=
\partial^\beta_{(\Re x,\Im x,\eta)}\big\{(-2\Pi_0^\pm)^{-1}\,
\big(
(\Pi_0^\pm+\Pi_0^+)(i\tilde{\partial}^+ B^{n}_+)
+(\Pi_0^\pm+\Pi_0^-)(i\tilde{\partial}^- B^{n}_-)
\big)\big|_{t=0}\big\}
\\
&
=\partial^\beta_{(\Re x,\Im x,\eta)}
\big\{-i\tilde{\partial}^\pm B^{n}_\pm\big|_{t=0}\big\}
\qquad\textrm{mod}\;\bigO(|\Im x|^N)\,,\quad
\beta\in\NN_0^{3d}\,.
\end{align*}
In the last step we used
\begin{equation}\label{markus4}
\Pi_0^+\big|_{t=0}=\sqrt{1+(\eta+i\nabla\vp)^2}=-\Pi_0^-\big|_{t=0}\,.
\end{equation}
In conclusion we see that \eqref{klaus3} is satisfied, too.
Again by Lemma~\ref{le-klaus} we deduce that
\eqref{bea4} is fulfilled, for $\nu=n+1$, and that
$\partial_{(t,\Re x,\Im x,\eta)}^\alpha
((\tilde{\Pi}^\pm-1)\,B_\pm^{n+1}+i\tilde{\partial}^{\pm}B^{n}_{\pm})
=\bigO(\Gamma_\pm^N)$. By virtue of Lemma~\ref{le-T-K} we
conclude that $B_\pm^{n+1}$ fulfills the
differential equation in $(\mathbf{T_{n+1}})$.
The initial condition in $(\mathbf{T_{n+1}})$ is satisfied
because of \eqref{bea3} and \eqref{markus4}.

Finally, we recall that the estimates \eqref{dbar-Bnu}
follow from Lemma~\ref{le-emely}. The last assertion on the
supports of $B_\pm^\nu$ is clear from the constructions in
Lemma~\ref{le-emely}, the fact that the values of $t$
are bounded on $\sN_\pm$, and the fact that the initial
conditions $B_\pm^\nu|_{t=0}$ are supported in $\supp(\chi)$. 
\end{proof}

\smallskip

\noindent
In order to calculate the leading asymptotics of the
Green kernel of $D_{h,V}$ at $(x_\star,y_\star)$ we have to compute 
the value of $B_+^0(\tau,x,0)$, for $(\tau,x,0)\in\mr{\sD}^+$,
more explicitly. We recall that, by definition,
$(\tau,x,0)\in\mr{\sD}^+$ implies that there is some
$y\in K_0$ such that $Q^+(t,y,0)\in K_0$, $t\in[0,\tau]$, and
$Q^+(\tau,y,0)=x$.

\begin{lemma}\label{le-gustel}
Let $B_+^0$ be a solution of $(\mathbf{T_0})$ as in
\eqref{bea1} and \eqref{bea2}. Let $(\tau,x,0)\in\mr{\sD}^+$
and let $U(\cdot,y):[0,\tau]\to\LO(\CC^{d_*})$ be the solution of
$$
\frac{d}{dt}\,U(t,y)\,=\,-
\frac{i\valpha}{2}\cdot\frac{\nabla V(Q^+(t,y,0))}{V(Q^+(t,y,0))}\,U(t,y)\,,
\qquad U(0,y)=\id\,,
$$
on $[0,\tau]$, where $Q^+(\tau,y,0)=x$. Then 
\begin{equation}\label{B0aufsD}
B_+^0(\tau,x,0)=
\frac{(-V(y))^{1/2}\,\chi(y,0)}{
(-V(x))^{1/2}\,\det\big[d_y Q^+(\tau,y,0)\big]^{1/2}}\,U(\tau,y)\,
\Lambda^+(i\nabla\vp(y)).
\end{equation}
\end{lemma}

\begin{proof}
On account of Proposition~\ref{prop-psi-1} we have
$\psi_+(t,x,0)=0$, $\nabla_x\psi_+(t,x,0)=0$, thus
$\Pi_0^+(t,x,0)=\sqrt{1-\nabla\vp(x)^2}=-V(x)$ and
$\vPi^+(t,x,0)=i\nabla\vp(x)$, for all $(t,x,0)\in\sD^+$.
We further deduce that $\partial_t\Pi_0^+(t,x,0)=0$,
$\diz_x\vPi^+(t,x,0)=i\Delta\vp(x)$, 
${\partial}_i^+\Pi_j^+(t,x,0)=-i\partial_{x_i}\partial_{x_j}\vp(x)$,
${\partial}_0^+\Pi_j^+(t,x,0)=0$,
${\partial}_j^+\Pi_0^+(t,x,0)=\partial_{x_j}V(x)$, 
where $i,j\in\{1,\ldots,d\}$, thus
$$
-i\!\!\!\sum_{0\klg\mu<\nu\klg d}\sigma_{\mu\nu}\,
({\mathrm{rot}^+}\,\vPi^+)_{\mu\nu}(t,x,0)=
\gamma_0\,\vgamma\cdot(-\nabla V(x))\,=\,\valpha\cdot\nabla V(x)\,,
$$
for all $(t,x,0)\in\sD^+$, where
$\vgamma:=(\gamma_1,\ldots,\gamma_d)$.
For $(t,y,0)\in\wt{\sD}^+$, the differential equation
\eqref{emely1} determining $B_+^0$, for the choice 
$S^+=0$,
thus reads
\begin{equation}\label{hubert}
\partial_t\wt{B}_+^0(t,y,0)=
\frac{i\Delta\vp(Q^+(t,y,0))+\valpha\cdot\nabla 
V(Q^+(t,y,0))}{2iV(Q^+(t,y,0))}
\,\wt{B}_+^0(t,y,0)\,,
\end{equation}
for $t\in[0,\tau]$,
and we have the initial condition 
$
\wt{B}_+^0(0,y,0)=\chi(y,0)\,\Lambda^+(i\nabla\vp(y))
$.
Using the ansatz 
$\wt{B}^0_+(t,y,0)=\wt{b}(t,y)\,U(t,y)\,\wt{B}_+^0(0,y,0)$
with a scalar $\wt{b}$ it thus
remains to solve
\begin{equation*}
\partial_t\wt{b}(t,y)\,=\,
\frac{\Delta\vp(Q^+(t,y,0))}{2V(Q^+(t,y,0))}\,\wt{b}(t,y)\,,
\quad t\in[0,\tau]\,,
\quad \wt{b}(0,y)=\chi(y,0)\,.
\end{equation*}
Using Liouville's formula for the equation
$\partial_t Q^+(t,y,0)=F(Q^+(t,y,0))$, where 
$F(x):=\nabla_pH(x,\nabla\vp(x))=-V^{-1}(x)\,\nabla\vp(x)$, $x\in K_0$,
that is,
$$
\partial_t \det[d_yQ^+(t,y,0)]\,=\,\diz\, F(Q^+(t,y,0))\,\det[d_yQ^+(t,y,0)]\,,
\quad (t,y,0)\in\wt{\sD}^+,
$$
where $\diz F=-V^{-1}\Delta\vp-V^{-1}\nabla V\cdot F$,
it is, however, elementary to verify that
\begin{align*}
\partial_t\big\{&(-V(Q^+(t,y,0)))^{-1/2}\,
\det[d_yQ^+(t,y,0)]^{-1/2}\big\}
\\
&=\,
\frac{\Delta\vp(Q^+(t,y,0))}{2V(Q^+(t,y,0))}\,\big\{(-V(Q^+(t,y,0)))^{-1/2}\,
\det[d_yQ^+(t,y,0)]^{-1/2}\big\}\,.
\end{align*}
We deduce that $\wt{b}(t,y)$ is equal to the term in the curly
brackets $\{\cdots\}$ times 
$(-V(Q^+(0,y,0))^{1/2}=(-V(y))^{1/2}$
(so that $\wt{b}(0,y)=1$)
and the formula \eqref{B0aufsD} follows.
\end{proof}

\smallskip

\noindent
In the one-dimensional case the formula \eqref{B0aufsD}
for the solution of the first transport equation can still
be written a bit more explicitly.

\begin{lemma}\label{le-U-d=1}
Assume that $d=1$, let $B_+^0$ be a solution of $(\mathbf{T_0})$ as in
\eqref{bea1} and \eqref{bea2}, and let $(\tau,x,0)\in\mr{\sD}^+$. Then
\begin{align}\nonumber
B_+^0(\tau,x,0)\,=\,&
\frac{(-V(y))^{1/2}\,\chi(y,0)}{
(-V(x))^{1/2}\,(Q^+)_y'(\tau,y,0)^{1/2}}
\\
&\cdot\big(\cos(\vt(\tau))\,\id-i\sin(\vt(\tau))\,\alpha_1\big)
\,\Lambda^+(i\vp'(y)),
\end{align}
where $Q^+(\tau,y,0)=x$ and 
$$
\vt(t)\,:=\,\int_0^t\frac{V'(Q^+(s,y,0))}{2V(Q^+(s,y,0))}\,ds\,,
\qquad t\in[0,\tau]\,.
$$
\end{lemma}

\begin{proof}
On account of Lemma~\ref{le-gustel}
we just have to compute the solution of
$$
\frac{d}{dt}\,U(t,y)\,=\,
-i\,\frac{V'(Q^+(t,y,0))}{2V(Q^+(t,y,0))}\,\alpha_1\,U(t,y)\,,\quad
t\in[0,\tau]\,,\quad U(0,y)\,=\,\id\,,
$$
which is given by 
$U(t,y)=\cos(\vt(t))\,\id-i\sin(\vt(t))\,\alpha_1$,
$t\in[0,\tau]$.
\end{proof}

\begin{remark}\label{rem-B0}
We know that 
$\tilde{\Pi}^+(\tau,x_\star,0)=-\alpha_0\,\valpha\cdot i\nabla\vp(x_\star)+
\alpha_0\,(-V(x_\star))$, where $\nabla\vp(x_\star)=\vo(\tau)$. 
Then the identity $(\tilde{\Pi}^+-1)\,B_+^0(\tau,x_\star,0)=0$
implies
\begin{align*}
B_0^+(\tau,x_\star,0)\,&=\,
\Big(\frac{1}{2}+\frac{1}{2}\,\tilde{\Pi}^+(\tau,x_\star,0)\Big)
\,B_0^+(\tau,x_\star,0)
\\
&=\,
(-V(x_\star))\,\Lambda(i\vo(\tau))\,\alpha_0\,B_0^+(\tau,x_\star,0)\,.
\end{align*}
In the following we verify directly that taking the hermitian 
conjugate of $M(x_\star,y_\star):=(-V(x_\star))\Lambda^+(i\vo(\tau))\,\alpha_0
\,U(\tau)\,(-V(y_\star))\Lambda^+(i\vo(0))$ gives the same
expression as interchanging the roles of $x_\star$ and $y_\star$.

Notice that the Hamiltonian trajectory in $\{H=0\}$ whose
position space projection runs from $x_\star$ to $y_\star$
is given by ${\wt{\gamma}\choose\wt{\vo}}(t):={\gamma\choose-\vo}(\tau-t)$,
$t\in[0,\tau]$. Thus
$$
M(y_\star,x_\star)\,=\,
(-V(y_\star))\Lambda^+(-i\vo(0))\,\alpha_0\,\wt{U}(\tau)\,
(-V(x_\star))\Lambda^+(-i\vo(\tau))\,,
$$ 
where $\wt{U}$ solves
$$
\frac{d}{dt}\,\wt{U}(t)\,=\,i\valpha\cdot
w(\wt{\gamma}(t))\,\wt{U}(t)\,,\quad
t\in[0,\tau]\,,\qquad w:=-(2V)^{-1}\nabla V\,.
$$
Hence, in order to verify $M(y_\star,x_\star)=M(x_\star,y_\star)^*$
it suffices to show that $\alpha_0\,\wt{U}(\tau)=U(\tau)^*\,\alpha_0$.
To this end we recall that $U(\tau)$ is given by the
Dyson series $U(\tau)=\sum_{n=0}^\infty I_n(\tau)$, where
$$
I_n(\tau)\,:=\,
\int_{\tau\,\triangle_n}i\valpha\cdot
w({\gamma}(t_n))\cdots i\valpha\cdot
w({\gamma}(t_1))\,dt_1\ldots dt_n\,,
$$ 
$\tau\,\triangle_n:=\{0\klg t_1\klg\dots\klg t_n\klg\tau\}$ denoting
the $n$-dimensional standard simplex scaled by $\tau$. 
On the one hand, $\{\alpha_0,\alpha_j\}=0$ and $\alpha_j^*=\alpha_j$, 
$j=1,\ldots,d$, implies
$$
I_n(\tau)^*\,\alpha_0\,=\,
\alpha_0\int_{\tau\,\triangle_n}i\valpha\cdot
w({\gamma}(t_1))\cdots i\valpha\cdot
w({\gamma}(t_n))\,dt_1\ldots dt_n\,.
$$
On the other hand the substitution
$s_1=\tau-t_n,\ldots,s_n=\tau-t_1$ turns the latter
expression into
$$
\alpha_0
\int_{\tau\,\triangle_n}i\valpha\cdot
w(\wt{\gamma}(s_n))\cdots i\valpha\cdot
w(\wt{\gamma}(s_1))\,ds_1\ldots ds_n=:\alpha_0\,\wt{I}_n(\tau)\,,
$$   
where $\wt{U}(\tau)=\sum_{n=0}^\infty\wt{I}_n(\tau)$.
\hfill$\Diamond$
\end{remark}

\subsection{Approximate solution of 
$\boldsymbol{(\pm h{\partial}_t+D_{h,V,{\varphi}})\,u_\pm=0}$}

\noindent
Finally, we put the results of Section~\ref{sec-CHJ} and
the present section together.
To this end we pick smooth cut-off functions,
$\vr_\pm\in C^\infty(\RR_0^+\times\RR^{2d})$, such that
$\vr_\pm\equiv1$ in some neighborhood of $\sE_\pm$ and
$\supp(\vr_\pm)\subset\sN_\pm\cap(\RR_0^+\times\RR^{2d})$
and such that all partial derivatives of any order of $\vr_\pm$
are uniformly bounded. Furthermore, we define 
$B_\pm(\,\cdot\,;h)\in C_0^\infty(\RR_0^+\times\RR^{2d},\LO(\CC^{d_*}))$, 
$h\in(0,1]$, by
\begin{equation}\label{def-Bpm}
B_\pm(t,x,\eta;h)\,=\,\sum_{\nu=0}^\infty h^\nu\,\theta(h/\ve_\nu)\,
\vr_\pm(t,x,\eta)\,B_\pm^\nu(t,x,\eta)\,,
\end{equation}
for $(t,x,\eta)\in\RR_0^+\times\RR^{2d}$, $h\in(0,1]$,
where $\theta\in C^\infty(\RR,[0,1])$ equals
$1$ on $(-\infty,1]$ and $0$ on $[2,\infty)$.
($B_\pm^\nu$ is set equal to zero outside its original
domain of definition $\sN_\pm$.)
If $\ve_\nu\searrow0$ sufficiently fast, we have
\begin{equation}\label{sigrid1}
\sup_{{(t,x,\eta)\in\atop\RR_0^+\times\RR^{2d}}}
\Big\|\partial_{(t,x,\eta)}^\alpha
\Big(B_\pm(t,x,\eta;h)-\sum_{\nu=0}^Nh^\nu\,
\vr_\pm(t,x,\eta)\,B_\pm^\nu(t,x,\eta)\Big)\Big\|
\klg C_{N,\alpha}\,h^{N+1},
\end{equation}
for $N\in\NN$, $\alpha\in\NN_0^{2d+1}$, $h\in(0,h_{\alpha,N}]$, and suitable
constants $C_{N,\alpha},h_{\alpha,N}\in(0,\infty)$.

\begin{proposition}\label{prop-sarah}
There exist compactly supported
matrix-valued functions, 
$\check{r}_\pm(\,\cdot\,;h)\in C^\infty(\RR_0^+\times\RR^{2d},\LO(\CC^{d_*}))$,
such that
\begin{align*}
(\pm h\,\partial_t+D_{h,V,\vp})\,(e^{i\psi_\pm/h}\,B_\pm)=
\check{r}_\pm\,,\qquad h\in(0,1]\,,
\end{align*}
and, for all $N\in\NN$ and $\alpha\in\NN_0^{2d+1}$,
there is some $C_{N,\alpha}\in(0,\infty)$ such that
\begin{equation}\label{thea}
\sup_{(t,x,\eta)\in\RR_0^+\times\RR^{2d}}
\|\partial_{(t,x,\eta)}^\alpha\check{r}_\pm(t,x,\eta;h)\|\,\klg\,
C_{N,\alpha}\,h^N,\quad h\in(0,1]\,.
\end{equation}
\end{proposition}

\begin{proof}
Since all terms in \eqref{sarah1} corresponding to 
the different powers of $h$ are equal to some smooth
matrix-valued function on $\sN_\pm$ whose partial derivatives
of any order are equal to $\bigO(\Gamma_\pm^N)$, $N\in\NN$, and
since $\Gamma_\pm\klg\bigO(1)\,\Im\psi_\pm$ on the real domain,
\begin{align}\nonumber
(\pm h\,\partial_t+D_{h,V,\vp}&)\,[e^{i\psi_\pm/h}\,B_\pm]
=
\sum_{\nu=0}^\infty h^\nu\,\,\theta(h/\ve_\nu)\,
e^{i\psi_\pm/h}\,\bigO\big((\Im\psi_\pm)^N\big)
\\\nonumber
&+
\sum_{\nu=1}^\infty h^\nu\,\theta(h/\ve_\nu)\,
e^{i\psi_\pm/h}\,[(\pm\partial_t-i\valpha\cdot\nabla)\vr_\pm]
\,B_\pm^{\nu-1}
\\
&+ 
\sum_{\nu=1}^\infty h^\nu\,\big(\theta(h/\ve_{\nu-1})-\theta(h/\ve_\nu)\big)
\,e^{i\psi_\pm/h}\,(\pm\partial_t-i\valpha\cdot\nabla)[\vr_\pm\,B_\pm^{\nu-1}]
\,.\label{michi0}
\end{align}
Consequently,
all partial derivatives of
the first term on the right hand side of \eqref{michi0} 
are of order $\bigO(h^\infty)$ because
\begin{equation}\label{sigrid2}
e^{-\Im\psi_\pm/h}\,(\Im\psi_\pm)^N\,\klg\,N!\,h^N,\qquad N\in\NN\,.
\end{equation}
Furthermore, $\Im\psi_\pm>0$ on $\supp(\vr_\pm')$ and all partial
derivatives of
$R_\nu:=[(\pm\partial_t-i\valpha\cdot\nabla)\vr_\pm]\,B_\pm^{\nu}$
are locally bounded, whence
$\partial^\alpha_{(t,x,\eta)}[e^{i\psi_\pm/h}\,R_\nu]=\bigO(h^\infty)$.
All partial derivatives of
the third term on the right hand side of \eqref{michi0} 
are of order $\bigO(h^\infty)$, too, since
$\theta(h/\ve_{\nu-1})-\theta(h/\ve_\nu)=0$, for all
$h\in(0,\ve_\nu)$, $\nu\in\NN$.
\end{proof}


\section{A parametrix for $\D{h,V,\vp}$}
\label{sec-parametrix}

\noindent
Given $k,m\in\RR$, we write  
$b\in S^k(\JB{\xi}^m;\LO(\CC^{d_*}))$, for some map 
$b:\RR^{2d}\times(0,h_0]\to \LO(\CC^{d_*})$, if $h_0>0$,
$b(\,\cdot\,;h)\in C^\infty(\RR^{2d},\LO(\CC^{d_*}))$, $h\in(0,h_0]$, and,
for all $\alpha\in\NN_0^{2d}$, we find $h_\alpha,C_\alpha\in(0,\infty)$
such that
$$
\|\partial_{(x,\xi)}^\alpha b(x,\xi;h)\|\,\klg\,C_\alpha\,\JB{\xi}^m\,h^{-k},
\qquad x,\xi\in\RR^d,\;h\in(0,h_\alpha]\,.
$$
We further set
$$
S^{-\infty}(\JB{\xi}^{-\infty};\LO(\CC^{d_*}))\,:=\,\bigcap_{k\in\NN}
S^{-k}(\JB{\xi}^{-k};\LO(\CC^{d_*}))\,.
$$
In what follows we work with the semi-classical standard quantization
of matrix-valued symbols $b\in S^k(\JB{\xi}^m;\LO(\CC^{d_*}))$ determined 
by the oscillatory integrals
$$
\Op_h(b)\,f(x):=\int e^{i\SPn{\xi}{x-y}/h}\,b(x,\xi)\,f(y)
\frac{dyd\xi}{(2\pi h)^d}\,,\quad f\in\schwartz(\RR^d,\CC^{d_*})\,.
$$
Let $k_1,k_2,m_1,m_2\in\RR$.
We recall that, for $b\in S^{k_1}(\JB{\xi}^{m_1};\LO(\CC^{d_*}))$ and
$c\in S^{k_2}(\JB{\xi}^{m_2};\LO(\CC^{d_*}))$, the symbol, 
$b\#_h c\in S^{k_1+k_2}(\JB{\xi}^{m_1+m_2};\LO(\CC^{d_*}))$, 
of $\Op_h(b)\circ\Op_h(c)$
has the following asymptotic expansion in 
$S^{k_1+k_2}(\JB{\xi}^{m_1+m_2};\LO(\CC^{d_*}))$,
\begin{align*}
b\#_h c(x,\xi;h)&=
e^{ihD_\eta D_y}\,b(x,\eta)\,c(y,\xi)\big|_{{y=x\atop \eta=\xi}}
\asymp
\sum_{\alpha\in\NN_0^d}\frac{h^{|\alpha|}}{i^{|\alpha|}\alpha!}\,
(\partial^\alpha_\xi b)(x,\xi)\,(\partial_x^\alpha c)(x,\xi)\,.
\end{align*}
If the symbol 
$\wh{D}_{V,\vp}(x,\xi)\in S^0(\JB{\xi};\LO(\CC^{d_*}))$ 
defined in \eqref{xiao}
were invertible,
for every $(x,\xi)\in\RR^{2d}$, we had a well-known asymptotic
expansion, 
$\breve{q}(x,\xi)\asymp\sum_{\nu=0}^\infty h^\nu\,q_\nu(x,\xi)$, 
of the matrix-valued symbol of the inverse operator.
We can, however, write down this asymptotic expansion 
formally and determine $q_\nu(x,\xi)$ at every point
$(x,\xi)\in\RR^{2d}$ where $\wh{D}_{V,\vp}(x,\xi)$
is invertible.
We proceed in this way and pick some cut-off
function $\wt{\chi}\in C_0^\infty(\RR^{2d},[0,1])$
such that $\wt{\chi}\equiv1$ in a small neighborhood
of $K_0\times\{0\}$ and $\supp(\wt{\chi})\subset\{\chi=1\}$,
where $\chi$ has been introduced below \eqref{sarah1}.
Then it turns out that 
$q_\nu\,(1-\wt{\chi})\in S^0(\JB{\xi}^{-\nu-1};\LO(\CC^{d_*}))$.
Let $\tilde{q}$ be a Borel resummation 
of $\sum_{\nu=0}^\infty h^\nu\,q_\nu\,(1-\wt{\chi})$, so that
$$
\tilde{q}-\sum_{\nu=0}^{N-1} h^\nu\,q_\nu
\,(1-\wt{\chi})\in S^{-N}(\JB{\xi}^{-N-1};\LO(\CC^{d_*}))\,,
\qquad N\in\NN\,.
$$
Then $D_{h,V,\vp}\circ\Op_h(\tilde{q})$ is a pseudo-differential
operator whose symbol has the asymptotic expansion
$$
\wh{D}_{V,\vp}\#_h\tilde{q}(x,\xi)\,\asymp\,
1-\wt{\chi}(x,\xi)+\sum_{\nu=0}^\infty h^\nu\,\breve{r}_\nu(x,\xi)\,.
$$
Here each error term 
$\breve{r}_\nu$
contains some partial derivative of $\wt{\chi}$
whence $\supp(\breve{r}_\nu)\subset\{\chi=1\}$, 
for every $\nu\in\NN_0$.
Setting
$$
q\,:=\,\tilde{q}\#_h(1-\chi)
$$
we thus have
\begin{equation}\label{diana1}
\wh{D}_{V,\vp}\#_hq\,-\,(1-\chi)\in 
S^{-\infty}(\JB{\xi}^{-\infty};\LO(\CC^{d_*}))\,.
\end{equation}
Next, we define an operator 
$\cP_h:\schwartz(\RR^d,\CC^{d_*})\to \schwartz'(\RR^d,\CC^{d_*})$,
\begin{align}\nonumber
(\cP_h\,f)(x)&:=
\sum_{\sharp\in\{+,-\}}
\sharp\int_0^\infty\!\!\!\int_{\RR^{2d}}
e^{i\psi_\sharp(t,x,\eta)/h-i\SPn{\eta}{y}/h}\,B_\sharp(t,x,\eta;h)\,f(y)\,
\frac{dyd\eta dt}{(2\pi h)^dh}
\\\label{def-cP}
&\qquad+\Op_h(q)\,f(x)\,,
\qquad x\in\RR^d,\;f\in C_0^\infty(\RR^d,\CC^{d_*})\,.
\end{align}
Here the integrals in the first line are effectively
evaluated over some compact set so that $\cP_h$ is
obviously well-defined.
Furthermore, it is clear that $\cP_h$ can be represented as an integral
operator with kernel
\begin{align}\nonumber
\cP_h(x,y)&=\sum_{\sharp\in\{+,-\}}
\sharp\int_0^\infty\!\!\!\int_{\RR^{d}}
e^{i\psi_\sharp(t,x,\eta)/h-i\SPn{\eta}{y}/h}\,B_\sharp(t,x,\eta;h)\,
\frac{d\eta dt}{(2\pi h)^dh}
\\\label{kern-cP}
&\quad
+\,\check{q}(x,x-y)\,.
\end{align}
We recall that $\check{q}(x,y-x)$,
with $\check{q}(x,y)=(\cF_h^{-1})_{\xi\to y}q(x,y)$,
is the distribution kernel of $\Op_h(q)$.
Here we normalize the (component-wise) semi-classical Fourier transform as
$$
\hat{f}(\eta):=(\cF_h\,f)(\eta):=
\int_{\RR^d}e^{-i\SPn{\eta}{y}/h}\,f(y)\,dy\,,
\qquad \eta\in\RR^d,\;f\in L^1(\RR^d,\sV)\,,
$$
where $\sV$ is $\CC^{d_*}$ or $\LO(\CC^{d_*})$.
Integrating by parts by means of the operators
$$
\frac{1-ih\,(
\ol{\nabla_\eta\psi_\pm}-y)\cdot\nabla_\eta}{
1+|\nabla_\eta\psi_\pm-y|^2}
$$
and using the fact that $B_\pm$ is compactly supported 
it is easy to see that 
\begin{equation*}
\|\partial_x^\alpha\partial_y^\beta\cP_h(x,y)\|\,
\klg \,C_{N,\alpha,\beta}\,h^{-|\alpha|-|\beta|}
\,\JB{x}^{-N}\,\JB{y}^{-N},\qquad x,y\in\RR^d,
\end{equation*}
for all $N\in\NN$, $\alpha,\beta\in\NN_0^d$, 
and suitable constants $C_{N,\alpha,\beta}\in(0,\infty)$.

\begin{theorem}\label{thm-parametrix}
(i) There is some 
$\tilde{r}\in S^{-\infty}(\JB{\xi}^{-\infty};\LO(\CC^{d_*}))$,
such that
\begin{equation*}
D_{h,V,\vp}\,\cP_h\,=\,\id-\Op_h(\tilde{r})\,.
\end{equation*}
\noindent(ii) 
There exist $h_0\in(0,1]$ and
smooth kernels, $\cR_h\in C^\infty(\RR^{2d},\LO(\CC^{d_*}))$, $h\in(0,h_0]$,
such that
\begin{equation}\label{michi6}
\|\partial_x^\alpha\partial_y^\beta\cR_h(x,y)\|
\klg C_{N,\alpha,\beta}\,h^N\JB{x}^{-N}\JB{y}^{-N},\qquad
x,y\in\RR^d,\,h\in(0,h_0]\,,
\end{equation}
for all $N\in\NN$, $\alpha,\beta\in\NN_0^d$, and suitable
constants $C_{N,\alpha,\beta}\in(0,\infty)$, and such that
\begin{equation*}
D_{h,V}^{-1}(x,y)\,=\,e^{-(\vp(x)-\vp(y))/h}\,(\cP_h(x,y)+\cR_h(x,y))\,,
\qquad x\not=y\,.
\end{equation*}
\end{theorem}

\begin{proof}
(i): By Proposition~\ref{prop-sarah} we have,
for $f\in C_0^\infty(\RR^d,\CC^{d_*})$,
\begin{align}\nonumber
\D{h,V,\vp}&\sum_{\sharp\in\{+,-\}}
\sharp\int_0^\infty\!\!\int_{\RR^d}
e^{i\psi_\sharp(t,x,\eta)/h}\,B_\sharp(t,x,\eta;h)\,\hat{f}(\eta)\,
\frac{d\eta\,dt}{(2\pi h)^d\,h}
\\
&=\,-\nonumber
\sum_{\sharp\in\{+,-\}}
\int_{\RR^d}\int_0^\infty\partial_t\big(
e^{i\psi_\sharp(t,x,\eta)/h}\,B_\sharp(t,x,\eta;h)\big)
\,\hat{f}(\eta)\,
\frac{dt\,d\eta}{(2\pi h)^d}
\\
&\qquad+\sum_{\sharp\in\{+,-\}}\nonumber\sharp
\int_0^\infty\!\!\int_{\RR^d}
\check{r}_\sharp(t,x,\eta;h)\,\hat{f}(\eta)\,
\frac{d\eta\,dt}{(2\pi h)^d\,h}
\\
&=\,\nonumber
\int_{\RR^d}\chi(x,\eta)\,e^{i\SPn{\eta}{x}/h}\,\sum_{\sharp\in\{+,-\}}
\Lambda^\sharp\big(\eta+i\nabla\vp(x)\big)\,\hat{f}(\eta)\,
\frac{d\eta}{(2\pi h)^d}
\\
&\quad-\Op_h(\tilde{r}_1)f(x)\,,\label{freya}
\end{align}
where $\Lambda^++\Lambda^-=\id$.
We recall that the integral appearing here in the first
line is effectively an integral over some compact set. 
In particular,
the boundary at $\infty$ does not
give any contribution to the integral
$\int_0^\infty\partial_t(\cdots)\,dt$.
Moreover, we put
$$
\tilde{r}_1(x,\eta;h)\,:=-\sum_{\sharp\in\{+,-\}}\sharp
\int_0^\infty e^{-i\SPn{\eta}{x}/h}\,\check{r}_\sharp(t,x,\eta;h)\,
\frac{dt}{h}\,.
$$
In view of \eqref{thea} and since $\check{r}$ has a compact
support it is clear that $\tilde{r}_1$ is a symbol
in $S^{-\infty}(\JB{\eta}^{-\infty};\LO(\CC^{d_*}))$.
In fact, if some derivative $\partial_{(x,\eta)}^\alpha$ is applied
to $\tilde{r}_1$ the inverse powers of $h$ obtained by differentiating
the phase $e^{-i\SPn{\eta}{x}/h}$ are compensated for by derivatives of
$\check{r}_\sharp$, which are of order $\bigO(h^\infty)$.
On account of \eqref{diana1} we further have
$$
D_{h,V,\vp}\,\Op_h(q)\,f(x)\,=\,
\int_{\RR^d}e^{i\SPn{\eta}{x}/h}\,(1-\chi(x,\eta))\,\hat{f}(\eta)\,
\frac{d\eta}{(2\pi h)^d}-\Op_h(\tilde{r}_2)f(x)\,,
$$
for some $\tilde{r}_2\in S^{-\infty}(\JB{\eta}^{-\infty},\LO(\CC^{d_*}))$,
which proves (i) with 
$\tilde{r}:=\tilde{r}_1+\tilde{r}_2$.

(ii): For sufficiently small $h>0$, $\|\Op_h(\tilde{r})\|_{\LO(L^2)}\klg1/2$,
and $(\id-\Op_h(\tilde{r}))^{-1}=\Op_h(c)$, for some
$c\in S^0(1;\LO(\CC^{d_*}))$.
Now, $\Op_h(c)\,(\id-\Op_h(\tilde{r}))=\id$ is equivalent
to $\Op_h(c-1)=\Op_h(c)\,\Op_h(\tilde{r})$ which implies
$c-1\in S^{-\infty}(\JB{\eta}^{-\infty};\LO(\CC^{d_*}))$,
because $\tilde{r}\in S^{-\infty}(\JB{\eta}^{-\infty};\LO(\CC^{d_*}))$.
Of course,
$\cP_h\Op_h(c)=\cP_h+\cP_h\Op_h(c-1)$
and the distribution kernel of $\Op_h(c-1)$ -- let us
call it $\cK_h$ -- fulfills
$\|\partial_x^\alpha\partial_y^\beta\cK_h(x,y)\|
\klg C_{N,\alpha,\beta}\,h^N\,\JB{x-y}^{-N}$.
Therefore, 
using the remarks preceding the statement of this theorem 
it is easy to check that 
${\cR_h}:=\cP_h\,\Op_h(c-1)$ is an integral operator
with smooth kernel -- again denoted by the same symbol -- 
such that \eqref{michi6} holds true.
Using $D_{h,V,\vp}\,(\cP_h+{\cR_h})=\id$,
and
$D_{h,V,\vp}=e^{\vp/h}\,D_{h,V}\,e^{-\vp/h}$, where
$\vp$ is bounded with bounded partial derivatives of
any order, we conclude that
$e^{-\vp/h}(\cP_h+{\cR_h})\,e^{\vp/h}=D_{h,V}^{-1}$.
\end{proof}


\section{Calculation of the leading asymptotics}
\label{sec-asymp}

\noindent
In this section we calculate the asymptotics of the 
integral kernel \eqref{kern-cP} of the operator $\cP_h$ defined in
\eqref{def-cP} at the distinguished points
$x_\star$ and $y_\star$ fulfilling Hypothesis~\ref{hyp-geo-Dirac}.
On account of Theorem~\ref{thm-parametrix} this will complete
the proofs of our main Theorems~\ref{mainthm} and~\ref{mainthm-d=1}.
As in the statement of these theorems we
let ${\gamma\choose\vo}:[0,\tau]\to\RR^{2d}$ denote a smooth curve
solving \eqref{Ham-Gl-H} and satisfying \eqref{H=0gvo}
such that $\gamma(0)=y_\star$ and $\gamma(\tau)=x_\star$ and
set
$$
v_{y_\star}\,:=\,\frac{d}{dt}\gamma(0)\,,\qquad 
v_{x_\star}\,:=\,\frac{d}{dt}\gamma(\tau)\,.
$$

\begin{proposition}\label{prop-christa}
Let $d\in\NN$.
As $h\in(0,1]$ tends to zero,
\begin{align}\nonumber
&D_{h,V,\vp}^{-1}(x_\star,y_\star)
\\
&\;=\,\frac{1}{h^d}
\Big(\frac{h}{2\pi}\Big)^{\frac{d-1}{2}}
\frac{(1+\bigO(h))\,(-V(y_\star))^{1/2}\,U(\tau,y_\star)\,
\Lambda^+(i\nabla\vp(y_\star))}{(-V(x_\star))^{1/2}
\sqrt{\det\begin{pmatrix}
0&-v_{y_\star}^\top
\\
v_{x_\star}&id_\eta Q^+(\tau,y_\star,0)
\end{pmatrix}
}}\,,\label{christa0}
\end{align}
where we use the same notation as in \eqref{B0aufsD}.
\end{proposition}

\begin{proof}
By Theorem~\ref{thm-parametrix}(ii) it suffices to consider
only the kernel $\cP_h$.
First, a standard argument shows that the distribution kernel 
$\check{q}(x,x-y)$ of $\Op_h(q)$ in \eqref{kern-cP} does
not contribute to the asymptotic expansion in \eqref{christa0}.
In fact, $(x-y)^{2N}\check{q}(x,x-y)$
is the inverse Fourier transform of
$h^{2N}\,\Delta_\xi^Nq(x,\xi)$ at $x-y$,
where $\Delta_\xi^Nq(x,\xi)$ is absolutely integrable
with respect to $\xi$, for large $N\in\NN$.

Next, we consider the integral
$$
I_-(x_\star,y_\star)\,:=\,
\int_0^\infty\!\!\int_{\RR^d}
e^{i\psi_-(t,x_\star,\eta)/h-i\SPn{\eta}{y_\star}/h}
\,B_-(t,x_\star,\eta;h)\,
\frac{d\eta\,dt}{(2\pi h)^d\,h}\,.
$$
At $t=0$ we have $\psi_-(0,x_\star,\eta)-\SPn{\eta}{y_\star}
=\SPn{\eta}{x_\star-y_\star}$.
Since $x_\star\not=y_\star$ we can thus show by integration by parts with
respect to $\eta$ that
$$
\int_0^\ve\!\!\int_{\RR^d}
e^{i\psi_-(t,x_\star,\eta)/h-i\SPn{\eta}{y_\star}/h}
\,B_-(t,x_\star,\eta;h)\,
\frac{d\eta\,dt}{(2\pi h)^d\,h}\,=\,\bigO(h^\infty)\,,
$$
provided that $\ve>0$ is sufficiently small.
Since $\Im\psi_-(t,x_\star,\eta)>0$, for $t>0$, by \eqref{psiaufE3},
it further follows that
$\int_\ve^\infty\!\!\int_{\RR^d}e^{i\psi_-/h-i\SPn{\eta}{y_\star}/h}\,B_-\,d\eta dt
=\bigO(h^\infty)$, for every fixed $\ve>0$.
Consequently, $I_-(x_\star,y_\star)=\bigO(h^\infty)$.

Finally, we treat the integral
\begin{align*}
I_+(x_\star,y_\star)\,&:=\,
\int_0^\infty\!\!\int_{\RR^d}
e^{i\psi_+(t,x_\star,\eta)/h-i\SPn{\eta}{y_\star}/h}\,
B_+(t,x_\star,\eta;h)\,
\frac{d\eta\,dt}{(2\pi h)^d\,h}\,,
\end{align*}
which is the only term contributing to the asymptotic expansion.
We shall apply a complex stationary phase expansion with respect
to the $d+1$ variables $(t,\eta)$.
The critical points of the phase are given by
\begin{align}
0\,&=\,\partial_t\psi_+(t,x_\star,\eta)\,,\label{christa1}
\\
0\,&=\,\nabla_\eta\psi_+(t,x_\star,\eta)-y_\star\,.\label{christa2}
\end{align}
To find the asymptotics of $I_+(x_\star,y_\star)$ it certainly
suffices to determine all critical points $(t,\eta)$ 
with $\Im\psi_+(t,x_\star,\eta)=0$. We know, however, that
$\Im\psi_+(t,x_\star,\eta)=0$ implies $t=0$ or
$(t,x_\star,\eta)\in\sD^+$. As above, we 
infer by integration by parts that 
$\int_0^\ve\int_{\RR^d}e^{i\psi_+/h-i\SPn{\eta}{y_\star}/h}\,B_+\,d\eta dt
=\bigO(h^\infty)$, for some sufficiently small $\ve>0$.
The only critical point $(t,\eta)$ such that $(t,x_\star,\eta)\in\sD^+$
is, however, given by $(t,\eta)=(\tau,0)$.
The method of complex stationary phase
\cite{MeSj1} thus implies that
\begin{align}
I_+(x_\star,y_\star)\,&=\,
\frac{(2\pi h)^{\frac{d+1}{2}}}{(2\pi h)^d\,h}\,
\frac{e^{i\psi_+(\tau,x_\star,0)/h}\,B_+^0(\tau,x_\star,0)\,\big(1+\cO(h)\big)}{
\sqrt{\det\,\frac{1}{i}
\begin{pmatrix}
\partial^2_t\psi_+&\partial_td_\eta\psi_+
\\
\partial_t\nabla_\eta\psi_+&d_\eta\nabla_\eta\psi_+
\end{pmatrix}(\tau,x_\star,0)}}\,.\label{christa3}
\end{align}
Thanks to \eqref{lotta99} we know that 
$\psi_+(\tau,x_\star,0)=0=\partial_t^2\psi_+(\tau,x_\star,0)$
and differentiating the identity
$Q^+(t,k^+(t,x_\star,\eta),\eta)=x_\star$
we obtain
\begin{align*}
d_yQ^+\,\partial_t{k^+}(t,x_\star,\eta)
+d_{\overline{y}}Q^+&\,
\overline{\partial_t{k^+}(t,x_\star,\eta)}
\,=\,
-\partial_t{Q^+}\,,
\\
d_yQ^+\,d_\eta k^+(t,x_\star,\eta)+
d_{\overline{y}}Q^+&\,
\overline{d_\eta k^+(t,x_\star,\eta)}\,=\,
-d_\eta Q^+\,,
\end{align*}
where all derivatives of $Q^+$ are evaluated at 
$(t,k^+(t,x^\star,\eta),\eta)$.
Inserting $(t,\eta)=(\tau,0)$ we infer that
\begin{align}
d_yQ^+(\tau,y_\star,0)\,v_{y_\star}
\,&=\,
v_{x_\star}\,,\label{christa4}
\\
d_yQ^+(\tau,y_\star,0)\,
\tfrac{1}{i}d_\eta\nabla_\eta\psi_+(\tau,x_\star,0)
\,&=\,
i\,d_\eta Q^+(\tau,y_\star,0)\label{christa5}
\,,
\end{align}
The determinant appearing in \eqref{christa3} is thus equal to
\begin{equation}\label{christa6}
\det
\begin{pmatrix}
0&-v_{y_\star}^\top
\\
v_{y_\star}&d_yQ^+(\tau,y_\star,0)^{-1}\,i\,d_\eta Q^+(\tau,y_\star,0)
\end{pmatrix}\,.
\end{equation}
Multiplying the determinant \eqref{christa6} and the one
appearing in \eqref{B0aufsD}, which can be
written as 
$$
\det\big[d_yQ^+(\tau,y_\star,0)\big]\,=\,
\det\begin{pmatrix}
1&0\\0&d_yQ^+(\tau,y_\star,0)
\end{pmatrix}\,,
$$
and using $\chi(y_\star,0)=1$, \eqref{christa4}, and \eqref{christa5} 
we arrive at \eqref{christa0}.
\end{proof}

\smallskip

\noindent
Next, we re-write the formula \eqref{christa0} in the one-dimensional
case.

\begin{proposition}\label{prop-christa-d=1}
For $d=1$, we have, as $h\searrow0$,
\begin{align}\nonumber
D_{h,V,\vp}^{-1}(x_\star,y_\star)
&=\,\frac{1+\bigO(h)}{
h\,(1-V^2(x_\star))^{1/4}(1-V^2(y_\star))^{1/4}}
\\
&\quad\cdot
\big(\cos(\vt(\tau))\,\id-i\sin(\vt(\tau))\,\alpha_1\big)
\,(-V(y_\star))^{1/2}\,\Lambda(i\nabla\vp(y_\star))
\,,\label{christa0-d=1}
\end{align}
\end{proposition}

\begin{proof}
In the case $d=1$ the formula \eqref{christa0} reduces to
\begin{align*}
I(x_\star,y_\star)\,&=\,(1+\bigO(h))\,
\,
\frac{U(\tau,y_\star)\,(-V(y_\star))^{1/2}
\,\Lambda^+(i\vp'(y_\star))}{
h\,|V(x_\star)|^{1/2}|V(y_\star)|^{1/2}(v_{x_\star}\,v_{y_\star})^{1/2}}
\,.
\end{align*}
Here $v_{x_\star}$ and $v_{y_\star}$ have the same sign and
$|v_z|=(1-V^2(z))^{1/2}/|V(z)|$, for $z=x_\star,y_\star$, so that the factors
$|V(z)|$, $z=x_\star,y_\star$, cancel each other in the denominator.
Moreover, we have already calculated $U(\tau,y_\star)$
in the proof of Lemma~\ref{le-U-d=1}. 
\end{proof}

\smallskip

\noindent
In more than one dimension we can evaluate the determinant
in \eqref{christa0} more explicitly.
To this end we first prove the following lemma.

\begin{lemma}
Let $X$ be the position space projection of the
Hamiltonian flow associated with $H$ as defined in \eqref{waltraud}.
Then the following identity holds,
\begin{equation}\label{christa10}
id_\eta Q^+(\tau,y_\star,0)\,=\,d_pX(\tau,y_\star,\vo(0))\,,
\end{equation}
where $\vo(0)=\nabla\vp(y_\star)$ is the initial momentum of the Hamiltonian
trajectory from $y_\star$ to $x_\star$ as in the statement of
Theorem~\ref{mainthm}.
\end{lemma}

\begin{proof}
By Lemma~\ref{le-egon}
we have $Q^+(t,y_\star,0)=X(t,y_\star,0)=\gamma(t)$, $t\in[0,\tau]$.
We set $\rho(t):=(\gamma(t),\vo(t))
=(\gamma(t),\nabla\vp(\gamma(t)))$, $t\in[0,\tau]$.
On account of \eqref{a7}, \eqref{a8}, and \eqref{def-FFa}
we thus find
\begin{align*}
\frac{d}{dt}{d_\eta Q^+\choose d_\eta\Xi^+}(t,y_\star,0)
&=
\begin{pmatrix}
\BB(t)&-i\mathbb{A}(t)
\\
0&-\BB(t)^\top
\end{pmatrix}
{d_\eta Q^+\choose d_\eta\Xi^+}(t,y_\star,0)\,,
\\
{d_\eta Q^+\choose d_\eta\Xi^+}(0,y_\star,0)&={0\choose\id}\,,
\end{align*}
where $\mathbb{A}(t)=H_{pp}''(\rho(t))$
and $\BB(t):=\BB(t,y_\star)$ is defined as in \eqref{def-BB} with $y=y_\star$,
for $t\in[0,\tau]$.
On the other hand,
\begin{align}\label{michi1}
\frac{d}{dt}{d_pX\choose d_pP}(t,\rho(0))&=
\begin{pmatrix}
H_{px}''&H_{pp}''
\\
-H_{xx}''&-H_{xp}''
\end{pmatrix}(\rho(t))
{d_pX\choose d_pP}(t,\rho(0)),
\\\label{michi2}
{d_pX\choose d_pP}(0,\rho(0))&={0\choose\id}\,.
\end{align}
Hence, using 
$0=d_x(\nabla_xH(x,\nabla\vp))
=H_{xx}''(x,\nabla\vp)+H_{xp}''(x,\nabla\vp)\,\vp''$
on $K_0$ and \eqref{michi1}\&\eqref{michi2},
we can verify directly that
$$
{\wt{X}(t)\choose\wt{P}(t)}
\,:=\,{d_pX(t,\rho(t))\choose-\vp''(\gamma(t))\,d_pX(t,\rho(t))+d_pP(t,\rho(t))}
\,,\qquad t\in[0,\tau]\,,
$$
solves
\begin{align*}
\frac{d}{dt}{\wt{X}\choose\wt{P}}(t)=
\begin{pmatrix}
\BB(t)&\mathbb{A}(t)
\\
0&-\BB(t)^\top
\end{pmatrix}{\wt{X}\choose\wt{P}}(t),
\quad{\wt{X}\choose\wt{P}}(0)={0\choose\id},
\end{align*}
which implies \eqref{christa10}.
\end{proof}

\begin{lemma}\label{le-christa}
The following identity holds,
\begin{equation*}
\det\begin{pmatrix}
0&-v_{y_\star}^\top
\\
v_{x_\star}&i d_\eta Q^+(\tau,y_\star,0)
\end{pmatrix}=\frac{\dA(x_\star,y_\star)^{d-1}
\det\big(\exp_{y_\star}'(\exp^{-1}_{y_\star}(x_\star))\big)}{
|V(x_\star)||V(y_\star)|
(1-V^2(x_\star))^{\frac{d-2}{2}}(1-V^2(y_\star))^{\frac{d-2}{2}}}.
\end{equation*}
\end{lemma}

\begin{proof}
We drop the subscripts $\star$ of the distinguished
points $x_\star$, $y_\star$ in this proof.
First, we introduce some notation.
Let $b_1,\ldots,b_d$ denote some $G(y)$-orthonormal basis of 
$\RR^d$ such that $b_d=(1-V^2(y))^{-1/2}v_y/|v_y|$
and $c_1,\ldots,c_d$ be some $G(x)$-orthonormal basis of 
$\RR^d$ such that $c_d=(1-V^2(x))^{-1/2}v_x/|v_x|$.
Let $b_1^*,\ldots,b_d^*$ and $c_1^*,\ldots,c_d^*$ denote
the corresponding dual bases and let
$B$ and $C$ denote the matrices whose $i$-th row
is $b_i^*$ and $c_i^*$, respectively.
Then we have
$B\,v_y=(1-V^2(y))^{1/2}\,|v_y|\,{\sf e}_d$ and
$C\,v_x=(1-V^2(x))^{1/2}\,|v_x|\,{\sf e}_d$, where
${\sf e}_d$ is the $d$-th canonical basis vector of $\RR^d$.
Therefore, we find, using \eqref{christa10}, that is,
$id_\eta Q^+(\tau,y,0)=d_pX(\tau,y,\vo(0))$,
\begin{align*}
&\det\begin{pmatrix}
1&0\\0&C
\end{pmatrix}
\det\begin{pmatrix}
0&-v_y^\top\\v_x&id_\eta Q^+(\tau,y,0)
\end{pmatrix}
\det\begin{pmatrix}
1&0\\0&B^\top
\end{pmatrix}
\\
&=\,(1-V^2(x))^{1/2}\,|v_x|\,
(1-V^2(y))^{1/2}\,|v_y|\,
\det\begin{pmatrix}
0&-{\sf e}_d^\top\\{\sf e}_d&C\,d_pX(\tau,y,\vo(0))\,B^\top
\end{pmatrix}\,.
\end{align*}
Since $\det B=(1-V^2(y))^{d/2}$, $\det C=(1-V^2(x))^{d/2}$,
$$
|v_y|=\big|\nabla_pH(y,\vo(0))\big|=
\frac{|\vo(0)|}{(1-\vo(0)^2)^{1/2}}=\frac{(1-V^2(y))^{1/2}}{|V(y)|}\,,
$$
and, analogously, $|v_x|=(1-V^2(x))^{1/2}/|V(x)|$, we obtain
\begin{align}\label{michi3}
\det\begin{pmatrix}
0&-v_y^\top\\v_x&id_\eta Q^+(\tau,y,0)
\end{pmatrix}=
\frac{\det\big((c_i^*\,d_pX(\tau,y,\vo(0))\,(b_j^*)^\top)_{
1\klg i,j\klg d-1}\big)}{
|V(x)||V(y)|
(1-V^2(x))^{\frac{d-2}{2}}(1-V^2(y))^{\frac{d-2}{2}}}\,.
\end{align}
In order to compare the position space projection, $X$,
of the Hamiltonian flow associated with $H$ and
the exponential map at $y$ we observe that
\begin{align*}
X(\tau,y,p)=
\exp_y\Big(\dA\big(X(\tau,y,p),y\big)\,(1-V^2(y))^{-1/2}\,
p/|p|\Big)\,,
\end{align*}
for $p\in\RR^d$ in some neighborhood of $\vo(0)$, since
$\cL(p):=(1-V^2(y))^{-1/2}\,p/|p|$ is normalized with respect to
$G(y)$ and the initial momentum $p$ of a Hamiltonian
trajectory is collinear with its initial velocity in our case.
We set $r(p):=\dA(X(\tau,y,p),y)$.
Then it follows that
\begin{align*}
d_pX(\tau,y,p)=
\exp_y'\big(r(p)\,\cL(p)\big)\big[\cL(p)\otimes r'(p)\big]
+r(p)\,\exp_y'\big(r(p)\,\cL(p)\big)\,\cL'(p)\,.
\end{align*}
By Gau{\ss}' lemma we know that
$c_i^*\exp_y'(r(\vo(0))\,\cL(\vo(0)))\,\cL(\vo(0))=0$, for $i=1,\ldots,d-1$,
thus
\begin{align*}
c_i^*\,d_pX(\tau,y,\vo(0))\,(b_j^*)^\top=
r(\vo(0))\,c_i^*\,\exp_y'\big(r(\vo(0))\,\cL(\vo(0))\big)
\,\cL'(\vo(0))\,(b_j^*)^\top,
\end{align*}
for $i,j=1,\ldots,d-1$.
If $P^\bot_{\vo(0)}$ denotes the orthogonal projection
in $\RR^d$ onto the Euclidean orthogonal complement
of $\vo(0)$, we have
$$
\cL'(\vo(0))\,(b_j^*)^\top\,=\,
(1-V^2(y))^{-1/2}\,\frac{1}{|\vo(0)|}\,
P^\bot_{\vo(0)}\,(1-V^2(y))\,b_j=
b_j\,,
$$
since $|\vo(0)|=(1-V^2(y))^{1/2}$.
Using also the identities
$r(\vo(0))=\dA(x,y)$ and
$r(\vo(0))\,\cL(\vo(0))=\exp_y^{-1}(x)$
we arrive at
\begin{align*}
\det\big((c_i^*&\,d_pX(\tau,y,\vo(0))\,(b_j^*)^\top)_{1\klg i,j\klg d-1}\big)
\\
&=\,\dA(x,y)^{d-1}\,
\det\big((c_i^*\,
\exp_{y}'(\exp_{y}^{-1}(x))\,b_j)_{1\klg i,j\klg d-1}\big)\,.
\end{align*}
Finally, we use $c_d^*\exp_{y}'(\exp_{y}^{-1}(x))\,b_d=1$
to conclude that
$$
\det\big((c_i^*\,d_pX(\tau,y,\vo(0))\,(b_j^*)^\top)_{1\klg i,j\klg d-1}\big)
=\dA(x,y)^{d-1}
\det\big(
\exp_{y}'(\exp_{y}^{-1}(x))\big)\,.
$$
Inserting this identity into \eqref{michi3} we arrive at the assertion.
\end{proof}


\appendix

\section{Connection to the BMT equation for Thomas precession}
\label{app-BMT}

\noindent
In this appendix we consider only the case $d=3$
and choose the standard representation of the Dirac matrices,
$$
\alpha_j=\begin{pmatrix}
0&\sigma_j\\\sigma_j&0
\end{pmatrix}\,,\;\;j=1,2,3\,,\qquad
\alpha_0=\begin{pmatrix}
\id&0\\0&-\id
\end{pmatrix}\,.
$$
We know that
the matrix-valued term 
$M(x_\star,y_\star)=U(\tau)\,(-V(y_\star))\,\Lambda^+(i\vo(0))$
appearing in \eqref{asymp-DV} maps $\Ran\Lambda^+(i\vo(0))$ onto
$\Ran\Lambda^+(i\vo(\tau))$. If we choose appropriate bases
of these subspaces then the coefficient matrix of $M(x_\star,y_\star)$
corresponding to these basis vectors is a solution of some
spin transport equation which is closely related to the
Bargmann-Michel-Telegdi (BMT) equation for the Thomas
precession of a classical three-dimensional spin;
see, e.g., \cite{BolteKeppeler1999,RubinowKeller1963}. 
In our special case, where no magnetic field
is present, the Hamiltonian determining the particle
trajectory ${\gamma\choose\vo}$ is $H$ given by \eqref{def-H}, and
$\sqrt{1-\vo^2}=-V(\gamma)$,
this spin transport equation reads
\begin{equation}\label{BMT-eq}
\frac{d}{dt}\,\fs(t)\,=\,i\,\fM(\gamma(t),\vo(t))
\,\fs(t)\,,\qquad
\fM(x,p)\,:=\,\frac{\vsigma\cdot\big(E(x)\times p\big)}{
-2V(x)[1-V(x))]}
\,,
\end{equation}
where $E=-\nabla V$ is the electrical field and $\fs(t)$ is a
complex $(2\times2)$-matrix. Notice that the usual momentum
is replaced by an imaginary
momentum in $H$ since we are dealing with some sort of
tunneling regime. The BMT equation, or rather its analogue
for the Hamiltonian $H$,
is obtained from \eqref{BMT-eq}
by choosing some $u\in\CC^2$ and computing the differential
equation satisfied by the expectation value
$\mathbf{s}(t):=\SPn{\fs(t)\,u}{\vsigma\,\fs(t)\,u}_{\CC^2}$
of the vector of Pauli matrices. In our case the BMT
equation turns out to be
\begin{equation}\label{BMT-eq2}
\frac{d}{dt}\,\mathbf{s}(t)\,=\,
\frac{\mathbf{s}(t)\times(E(\gamma(t))\times\vo(t))}{
-V(\gamma(t))[1-V(\gamma(t))]}\,.
\end{equation}
In order to derive \eqref{BMT-eq} and connect
it to \eqref{asymp-DV} we start with Equation \eqref{bea4}
for $B_+^0$, which, on the domain $\sD_+$, reads
\begin{equation}\label{BMT-eq3}
-2V(x)\,\Lambda^+(i\nabla\vp(x))\,(i\partial_t+\valpha\cdot\nabla_x)\,
B_+^0(t,x,0)\,=\,0\,,\qquad(t,x,0)\in\sD_+\,. 
\end{equation}
For every $x\in\RR^3$,
the range of $\Lambda^+(i\nabla\vp(x))$ is spanned by the two
mutually orthonormal eigenvectors of $\wh{D}(x,i\nabla\vp(x))$
in the $(4\times2)$-matrix
\begin{equation}\label{albrecht}
W(x)\,:=\,\frac{1}{\sqrt{-2V(x)[1-V(x)]}}\begin{pmatrix}
[1-V(x)]\,\id\\\vsigma\cdot i\nabla\vp(x)
\end{pmatrix}\,,
\end{equation}
so that $\Lambda^+(i\nabla\vp(x))=W(x)\,W(x)^\top$.
Next, we write $B_+^0(t,x,0)=W(x)\,C(t,x)$,
for some $(2\times4)$-matrix $C(t,x)$. This is possible
since $B_+^0$ fulfills $(\boldsymbol{T_0})$.
Then \eqref{BMT-eq3} implies 
$W^\top\,(i\partial_t+\valpha\cdot\nabla_x)\,W\,C=0$
and we observe as in \cite{BolteKeppeler1999} that the operator
$-iW^\top\,(i\partial_t+\valpha\cdot\nabla_x)\,W$ equals
$$
\partial_t+F\cdot\nabla_x+\frac{1}{2}\,
\mathrm{div}F-i\,\fM(x,\nabla\vp)\,,
\quad F(x):=\nabla_pH(x,\nabla\vp(x))=V(x)^{-1}\nabla\vp(x)\,.
$$
Substituting $\gamma$ for $x$ and using
$\dot\gamma=F(\gamma)$, we deduce that
$$
\frac{d}{dt}\,C(t,\gamma(t))\,=\,
-\frac{1}{2}\,\mathrm{div}F(\gamma(t))\,C(t,\gamma(t))
+i\,\fM(\gamma(t),\vo(t))\,C(t,\gamma(t))\,.
$$
Notice that, unlike \eqref{hubert}, the previous equation
involves the complete divergence of $F$. Using the ansatz
$C(t,\gamma(t))=\vr(t)\,\fs(t)\,W(y_\star)^\top$, where
$\vr$ is scalar, and $\fs$ solves \eqref{BMT-eq} with $\fs(0)=\id_2$,
we thus find by means of Liouville's formula and the
initial condition $C(0,y_\star)=W(y_\star)^\top$ that
$\vr(t)=\det[d_yQ_+(t,y_\star,0)]^{-1/2}$, $t\in[0,\tau]$.
We arrive at the following formula for $B_+^0$ alternative to
\eqref{B0aufsD},
$$
B_+^0(\tau,x_\star,0)\,=\,
\det[d_yQ_+(t,y_\star,0)]^{-1/2}\,W(x_\star)\,
\fs(\tau)\,W(y_\star)^\top.
$$
Using the previous formula instead of \eqref{B0aufsD} in
the proof of Proposition~\ref{prop-christa} we obtain the
following asymptotics for the Green kernel,
\begin{align*}
&\D{h,V}^{-1}(x_\star,y_\star)\nonumber
\\
&=\,
\frac{(1-V^2(x_\star))^{\frac{1}{4}}(1-V^2(y_\star))^{\frac{1}{4}}
\big(V(x_\star)V(y_\star)\big)^{\frac{1}{2}}}{
h^3\,\det\big[\exp_{y_\star}'(\exp_{y_\star}^{-1}(x_\star))\big]^{1/2}}
\cdot
\frac{(1+\bigO(h))\,e^{-\dA(x_\star,y_\star)/h}}{
2\pi\,\dA(x_\star,y_\star)/h}
\\
&\qquad\cdot\nonumber 
W(x_\star)\,\fs(\tau)\,W(y_\star)^\top,\qquad
\fs\;\textrm{solves \eqref{BMT-eq}},\;\,\fs(0)=\id_2\,.
\end{align*}
Here $W$ is given by \eqref{albrecht} where
$\nabla\vp(x_\star)=\vo(\tau)$, $\nabla\vp(y_\star)=\vo(0)$.
We remark without proof that the scalar term in the second line is
the asymptotic expansion of the Green kernel of
the Weyl quantization of $\sqrt{1+\xi^2}+V$ in three dimensions.
This can be inferred by means of a procedure similar to the
one carried through in the present paper.
Finally, we remark that $\fs(\tau)$ can be represented
by means of the polar coordinates of a suitable solution
of \eqref{BMT-eq2} and additional dynamical and geometric
phases; see \cite[\textsection4]{BolteKeppeler1999}. 


\section{Proof of Lemma~\ref{le-fabrizio}}
\label{app-MeSj}

\noindent
In order to give a self-contained construction
of the phase functions $\psi_\pm$ we present
the proof of Lemma~\ref{le-fabrizio} in this appendix.
The proofs below are variants of those in \cite{MeSj2}
where the symbol is assumed to be homogeneous of degree one.
We recall the definition of $\SMS$ in \eqref{def-SMS}
and start with the following lemma which corresponds to
\cite[Lemma~1.7]{MeSj2}. 

\begin{lemma}
For every compact subset $K\subset\Omega$, we find some
$C_K\in(0,\infty)$ such that,
for all $(s,\rho)=(s,x,\xi)\in\CC\times K$,
\begin{align}\label{hans}
(\wh{\sK}_{a_\pm}\,\SMS)(s,\rho)\,&\grg\,-\frac{3}{4}\,\Im a_\pm(\Re\rho)-C_K\,
|\Im\rho|^3.
\end{align}
\end{lemma}

\begin{proof}
Let $a$ be $a_+$ or $a_-$.
Since $\wh{\sK}_a$ is a real differential operator we have
$[\wh{\sK}_a,\Re]=0$, whence
\begin{align*}
\wh{\sK}_a\big(s+\SPn{x}{\Re\xi}\big)\,&=\,
a-\SPn{\nabla_\xi a}{\xi}+\SPn{\nabla_\xi a}{\Re\xi}
-\SPn{x}{\Re\nabla_x a}
\\
&=\,
a-\SPn{\nabla_\xi a}{i\Im\xi}-\SPn{i\Im x}{\nabla_x a}
-\Re\SPn{\ol{x}}{\nabla_x a}
\end{align*}
on $\Omega$.
Taking the imaginary part and using that
$[\wh{\sK}_a,\Im]=0$ we obtain
\begin{align*}
-\wh{\sK}_a\,\SMS\,&=\,
\Im\big(a-\SPn{\nabla_\rho a}{i\Im\rho}\big)\,.
\end{align*}
Taylor expanding both $a$ and $\nabla_\rho a$ at $\Re\rho$
using $\partial_{(\Re\rho,\Im\rho)}^\alpha\nabla_{\ol{\rho}}a(\Re\rho)=0$,
$\alpha\in\NN_0^{4d}$,
we infer that
\begin{equation}\label{sarah2}
a-\SPn{\nabla_\rho a}{i\Im\rho}\,=\,
a(\Re\rho)\,+\,\frac{1}{2}\,\SPn{\Im\rho}{a_{\rho\rho}''(\Re\rho)\,\Im\rho}
+\bigO\big(|\Im\rho|^3\big)\,.
\end{equation}
Furthermore, a Taylor expansion of $\Im a(\Re\rho\pm t\Im\rho)$
with $t>0$ yields
\begin{align}\nonumber
\frac{1}{t^2}\,\Im& a(\Re\rho)
+\frac{1}{2}\,\SPn{\Im\rho}{\Im a_{\rho\rho}''(\Re\rho)\,\Im\rho}
\\ \label{sarah3}
&=\frac{1}{2t^2}\,\big(\Im a(\Re\rho+t\Im\rho)+\Im a(\Re\rho-t\Im\rho)\big)
+t\,\bigO\big(|\Im\rho|^3\big)\,.
\end{align}
The $\bigO$-symbols in \eqref{sarah2} and \eqref{sarah3} are
uniform when $\rho$ varies in a compact subset of $\Omega$.
Choosing $t=2$ and using that $\Im a\klg0$ we thus arrive at the assertion.
\end{proof}

\smallskip

\begin{proof}[Proof of Lemma~\ref{le-fabrizio}]
We drop all sub- and superscripts $\pm$ in this proof.
Taylor expanding the right side of 
$\frac{d}{dt}\:\Im \kappa_t\,=\,\Im\wh{\sH}_a(\kappa_t)$
at $\Re\kappa_t$ and using Duhamel's formula we obtain, exactly as
in \cite[pp. 351]{MeSj2}, the estimate
\begin{equation}\label{est-Duhamel}
|\Im\kappa_u|\,\klg\,\bigO(1)\,\Big(\,
|\Im\kappa_t|\,+\,\int_u^t|\Im a'(\Re\kappa_r)|\,dr\,\Big)\;.
\end{equation}
It holds for all $\rho$ in a compact, complex
neighborhood of $(y_0,\eta_0)$ and $u,t\in[0,\tau+\ve_1]$, for some $\ve_1>0$.
Since, for $\rho$ in a compact set, 
the curves $[0,\tau+\ve_1]\ni t\mapsto\Re\kappa_t(\rho)$ stay in a
compact set, we may apply the 
standard estimate for positive functions to \eqref{est-Duhamel},
which together with H\"{o}lder's inequality gives
\begin{equation}\label{est-Duhamel2}
|\Im\kappa_r|\,\klg\,\bigO(1)\,\big(
|\Im\kappa_t|+(I_u^t)^{1/2}\big)\,,
\qquad I_u^t(\rho):=\int_u^t-\Im a(\Re\kappa_v(\rho))\,dv\,,
\end{equation}
for $0\klg u\klg r\klg t\klg\tau+\ve_1$.
Next, we integrate the estimate \eqref{hans} for 
$\frac{d}{dt}\SMS(\vs_t,\kappa_t)=\wh{\sK}_a(\SMS)(\vs_t,\kappa_t)$
from $u$ to $t$, use \eqref{est-Duhamel2} to bound $|\Im\kappa_r|$,
$r\in[u,t]$, and arrive at
\begin{align}\label{est-Duhamel3}
\SMS(\vs_t,\kappa_t)&\grg\SMS(\vs_u,\kappa_u)
+\,\frac{3}{4}\,I_u^t
-\bigO(1)\,(t-u)\,\big(|\Im\kappa_t|^3+(I_u^t)^{3/2}\big).
\end{align}
In the case $u=0$ we further observe 
by means of \eqref{ina1} and \eqref{est-Duhamel2} that
\begin{align}\label{est-Duhamel3b}
\SMS(\vs_0,\kappa_0)\grg-\bigO(1)\,|\Im\rho |^3\grg-\bigO(1)\,
\big(|\Im\kappa_t|^3+(I_0^t)^{3/2}\big)\,.
\end{align}
If $\tau>0$ in the plus-case we use
the assumption $\Im a_+(y_0,\eta_0)=0$ and Lemma~\ref{le-konrad} to deduce 
that $\Im a_+(\kappa_t^+(y_0,\eta_0))=0$, $t\in[0,\tau]$, thus
$I_u^t(y_0,\eta_0)=0$, $0\klg u\klg t\klg\tau$.
Back in the general case we conclude that,
for $0\klg u\klg t\klg\tau+\ve_1$, we can ensure that 
the integrals $I_u^t(\rho)$ are as small as we please by assuming
that $\rho$ is contained in a sufficiently small 
neighborhood of $(y_0,\eta_0)$, and that $\ve_1>0$
is sufficiently small. Then \eqref{est-Duhamel3} and 
\eqref{est-Duhamel3b} give \eqref{fabrizio1} and the estimate
\begin{align}\label{est-Duhamel4}
\SMS(\vs_t,\kappa_t)&\grg\SMS(\vs_u,\kappa_u)
+\,\frac{1}{2}\,I_u^t
-\bigO(1)\,(t-u)\,|\Im\kappa_t|^3,
\end{align}
for $0\klg u\klg t\klg\tau+\ve_1$.
Squaring \eqref{est-Duhamel2} with $r=u$, dividing by some
suitable constant, and adding the result
to \eqref{est-Duhamel4} we obtain
\begin{align*}
\SMS(\vs_t,\kappa_t)+\frac{1}{4}|\Im\kappa_t|^2
&\grg\SMS(\vs_u,\kappa_u)+\frac{2}{C}|\Im\kappa_u|^2
+\Big(\frac{1}{2}-\frac{1}{4}\Big)\,I_u^t
-\bigO(1)(t-u)|\Im\kappa_t|^3,
\end{align*}
where we can assume that the constant (coming from \eqref{est-Duhamel2})
satisfies $C\grg4$. By possibly restricting the neighborhood
of $(y_0,\eta_0)$ further to ensure that $\max_{t\in[0,\tau+\ve_1]}|\Im\kappa_t|$
is sufficiently small, 
we infer from the previous estimate that
\begin{align}\label{est-Duhamel5}
\SMS(\vs_t,\kappa_t)+|\Im\kappa_t|^2
&\grg\SMS(\vs_u,\kappa_u)+\frac{1}{C}\,|\Im\kappa_u|^2
+\,\frac{1}{4}\,I_u^t\,,
\end{align}
for $0\klg u\klg t\klg\tau+\ve_1$.
On account of \eqref{fabrizio1}, where the integral
is non-negative, we may again restrict the neighborhood
of $(y_0,\eta_0)$ so that $|\Im\kappa_u(\rho)|$ is sufficiently
small to ensure that
$(1-\frac{1}{C})\SMS(\vs_u,\kappa_u)+\frac{1}{C}|\Im\kappa_u|^2\grg0$, for
all $u\in[0,\tau+\ve_1]$. Subtracting
the latter term
from the right hand side of \eqref{est-Duhamel5} 
and using $C\grg4$
we arrive at
\begin{align*}
\big(\SMS(\vs_t,\kappa_t)+|\Im\kappa_t|^2\big)
&\grg\frac{1}{C}\,\big(\SMS(\vs_u,\kappa_u)+|\Im\kappa_u|^2\big)
+\,\frac{1}{C}\,I_u^t\,,
\end{align*}
which is \eqref{fabrizio2}.
Subtracting 
$\frac{1}{C}(\SMS(\vs_r,\kappa_r)+\frac{1}{2}\,|\Im\kappa_r|^2)\grg0$
from the right hand side of \eqref{fabrizio2}, where the
integral is non-negative, we finally obtain
\eqref{fabrizio3}.
\end{proof}


\bigskip\noindent{\bf Acknowledgement.}
This work has been supported by the DFG (SFB/TR12).
O.M. thanks the Institute for Mathematical Sciences
and Center for Quantum Technologies of the National University
of Singapore, where parts of this manuscript have been
prepared, for their kind hospitality.

\end{document}